\documentclass[]{imsart}

\RequirePackage{amsthm,amsmath,amsfonts,amssymb}
\RequirePackage[numbers]{natbib}
\RequirePackage[colorlinks,citecolor=blue,urlcolor=blue]{hyperref}
\RequirePackage{graphicx}

\startlocaldefs

\newtheorem{theorem}{Theorem}[section]
\newtheorem{lemma}[theorem]{Lemma}

\newtheorem{prop}[theorem]{Proposition}
\newtheorem{corollary}[theorem]{Corollary}

\theoremstyle{remark}

\newtheorem{remark}[theorem]{Remark}

\usepackage{amsmath}
\usepackage{amssymb}
\usepackage{amsfonts}
\usepackage{amsthm}
\usepackage{bm}\label{key}
\usepackage{caption}
\usepackage{subcaption}
\usepackage{dsfont}
\usepackage{enumitem}
\usepackage{rotating}

\usepackage{accents}

\usepackage{algorithm}
\usepackage{algpseudocode}

\usepackage[]{natbib}
\usepackage{multirow}

\usepackage{tikz}
\usepackage{pgfplots}

\newcommand{\real}[1][]{\ensuremath{\mathbb{R}^{#1}}}
\newcommand{\naturalnumber}{\ensuremath{\mathbb{N}}}
\newcommand{\family}[1]{\ensuremath{\mathcal{#1}}}

\newcommand{\rbracket}[1]{\ensuremath{\left( #1\right)}}
\newcommand{\sbracket}[1]{\ensuremath{\left\lbrack #1\right\rbrack}}
\newcommand{\cbracket}[1]{\ensuremath{\left\lbrace #1\right\rbrace}}
\newcommand{\set}[2]{\ensuremath{\cbracket{ #1 :\ #2}}}
\newcommand{\map}[3]{\ensuremath{#1 : #2 \rightarrow #3}}
\newcommand{\mapping}[2]{\ensuremath{#1 \longmapsto #2}}
\newcommand{\function}[4]{\ensuremath{#1_{#2}^{#3}\left(#4\right)}}
\newcommand{\indicator}[1]{\ensuremath{{\mathds{1}}\cbracket{#1}}}
\newcommand{\mat}[1]{\ensuremath{\bm{\mathrm{#1}}}}
\newcommand{\tr}[1]{\ensuremath{{#1}^{T}}}
\newcommand{\norm}[2]{\ensuremath{\left\|{#1}\right\|_{#2}}}

\newcommand{\virtualref}[2]{\expandafter\trim@spaces\expandafter{\IfSubStr{#1}{,}{\newcount\numofelem\StrCount{#1}{,}[\numofelem]\newcount\counter\counter=0{{#2}s~}\@for\tmp:=#1\do{\ref{\tmp}\ifnum\counter=\numofelem\else\ifnum\counter=\numexpr(\numofelem-1)\relax{~and~}\else{,~}\fi\fi\advance\counter by 1}}{\hyperref[{#1}]{{#2}~\ref{#1}}}}}

\newcommand{\fref}[1]{Figure~\hyperref[{#1}]{\ref{#1}}}
\newcommand{\tref}[1]{Table~\hyperref[{#1}]{\ref{#1}}}
\newcommand{\eref}[1]{\hyperref[{#1}]{(\ref{#1})}}
\newcommand{\cref}[1]{Chapter~\hyperref[{#1}]{\ref{#1}}}
\newcommand{\sref}[1]{Section~\hyperref[{#1}]{\ref{#1}}}
\newcommand{\aref}[1]{Appendix~\hyperref[{#1}]{\ref{#1}}}
\newcommand{\thref}[1]{Theorem~\hyperref[{#1}]{\ref{#1}}}
\newcommand{\leref}[1]{Lemma~\hyperref[{#1}]{\ref{#1}}}
\newcommand{\algoref}[1]{Algorithm~\hyperref[{#1}]{\ref{#1}}}
\newcommand{\pref}[1]{Proposition~\hyperref[{#1}]{\ref{#1}}}
\newcommand{\asref}[1]{Assumption~\hyperref[{#1}]{\ref{#1}}}

\DeclareMathOperator*{\argmax}{arg\,max}
\DeclareMathOperator*{\argmin}{arg\,min}
\DeclareMathOperator*{\maximize}{maximize}
\DeclareMathOperator*{\minimize}{minimize}

\newcommand{\limit}[2]{\ensuremath{\lim\limits_{#1 \rightarrow #2}}}
\newcommand{\converge}[1][]{\ensuremath{\xrightarrow{#1}}}

\newcommand{\survival}[2][S]{\ensuremath{#1\rbracket{#2|\mat{x}, \mat{z}}}}
\newcommand{\uncuredsurvival}[4][S]{\ensuremath{#1_{u}^{#4}\rbracket{#3|#2}}}
\newcommand{\uncuredrate}[1][\mat{x}]{\ensuremath{p\rbracket{#1}}}

\newcommand{\link}[3][\tr{\mat{\gamma}}\mat{x}]{\ensuremath{\varphi_{#2}^{#3}\rbracket{#1}}}
\newcommand{\basecumhazard}[2][0]{\ensuremath{\Lambda_{#1}\rbracket{#2}}}
\newcommand{\basehazard}[2][0]{\ensuremath{\lambda_{#1}\rbracket{#2}}}

\newcommand{\prob}[2][\rbracket]{\ensuremath{\mathbb{P}#1{#2}}}
\newcommand{\expect}[2][\rbracket]{\ensuremath{\mathbb{E}#1{#2}}}
\newcommand{\empirical}[1]{\prob[_{n}\rbracket]{#1}}

\newcommand{\bracketnum}[3][\zeta]{\ensuremath{N_{B}\rbracket{{#1}, {#2}, {#3}}}}
\newcommand{\bracketentropy}[3][\zeta]{\ensuremath{H_{B}\rbracket{{#1}, {#2}, {#3}}}}

\newcommand{\ubar}[1]{\ensuremath{\underaccent{\bar}{#1}}}
\newcommand{\dd}{\mathrm{d}}

\endlocaldefs

\begin{document}
	
	\begin{frontmatter}
		\title{Single-index mixture cure model \\under monotonicity constraints}
		\runtitle{Monotone single-index cure model}

		\begin{aug}
			\author{\fnms{Eni} \snm{Musta}\ead[label=e1]{e.musta@uva.nl}}
			\and
			\author{\fnms{Tsz Pang} \snm{Yuen}\ead[label=e2]{t.p.yuen@uva.nl}}
			
			\address{Korteweg-de Vries Institute for Mathematics, University of Amsterdam, Netherlands\\
				\printead{e1,e2}
			}
			\runauthor{E. Musta and T. P. Yuen}
		\end{aug}
		
		\begin{abstract}
			We consider survival data in the presence of a cure fraction, meaning
			that some subjects will never experience the event of interest. We assume a mixture cure model consisting of two sub-models: one for the probability of being uncured (incidence) and one for the survival of  the uncured subjects (latency). 
			Various approaches, ranging from parametric to nonparametric, have been used to model the effect of covariates on the incidence, with the logistic model being the most common one. We propose a monotone single-index model for the incidence and introduce a new estimation method that is based on the profile maximum likelihood approach and techniques from isotonic regression. The monotone single-index structure relaxes the parametric logistic assumption while maintaining interpretability of the regression coefficients. We investigate the consistency of the proposed estimator and show through a simulation study that, when the monotonicity assumption is satisfied,  it performs better compared to the non-constrained single-index/Cox mixture cure model. To illustrate its practical use, we use the new method to study melanoma cancer survival data. 
		\end{abstract}
		
	\end{frontmatter}

	\section{Introduction}
	Modelling time-to-event data in the presence of subjects that will never experience the event of interest has gained popularity over the recent decades. For instance, the advancement of cancer treatments has led to a larger fraction of patients being cured of their diseases \citep{LB2019}. Another example can be found in fertility studies \citep{VVBHZM2012}, where one is interested in the time to pregnancy while there are infertile couples for whom natural conception is impossible. Cure models have also been utilized in credit scoring to model the time to default of a loan applicant and default does not occur for the majority of debtors \citep{DCB2017}. In all these scenarios, the subjects that are immune to the event of interest are referred to as ‘cured’ (non susceptible). There are two types of cure models: mixture cure models and promotion time models. We refer the reader to \cite{AK2018} and \cite{PY2022} for  a comprehensive review of these models.
	
	Mixture cure models assume that the population is a mixture of cured and susceptible subjects and consist of two sub-models: one for the uncured probability (incidence) and one for the conditional survival function of the susceptibles (latency).  Initially, fully parametric models with logistic regression assumption for the incidence and different parametric distributions for the latency were proposed \citep{F1977,F1982}. Later on, extensions to semi-parametric models for the latency, such as the Cox proportional hazard (PH) model \citep{PD2000,ST2000} and the accelerated failure time model \citep{LT2002,ZP2007}, were introduced. For the incidence, a nonparametric estimator was developed in \cite{XP2014} based on the Beran estimator for the conditional survival function. However, such method is problematic for multivariate covariates since it requires multi-dimensional smoothing. To circumvent the curse-of-dimensionality, a single-index model for the incidence was introduced in \cite{AKL2019}, linking a linear predictor (index) to the incidence probability via an unspecified link function. The index achieves dimension reduction and alleviates the dimensionality issues when estimating the link function nonparametrically. In certain contexts, the link function is expected to be monotone, meaning that the cure probability increases/decreases as the risk score of an individual, given by the index, increases. The widely used logistic model for the incidence is in particular a monotone single-index model. In such cases, it is more appropriate to estimate the link function under monotonicity constraints, while the estimate proposed in \cite{AKL2019} is not guaranteed to be monotone. Another advantage of a monotone link function is interpretability since the sign of the coefficients of the index tell us whether a given covariate increases or decreases the cure chances. This motivates us to investigate estimation of a single-index model for the incidence under monotonicity constraints. We focus mainly on the incidence component and assume a Cox PH model for the latency.
	
	Single-index models have been thriving because of their flexibility over linear and parametric models, while avoiding dimensionality problems of general nonparametric models. The monotonicity of the link function appears in numerous applications leading to the popularity of generalized linear models. Therefore, it is appealing to impose monotonicity constraint on the link function, which leads to monotone single-index models and has recently become an active research area \cite{balabdaoui2019score,groeneboom2018current,BDJ2019,balabdaoui2021profile}. It is also worth mentioning that the binary choice model in econometrics and the current status linear regression model \citep{groeneboom2018current} are special cases of the monotone single-index model. More in general, there has been a growing interest in  statistical inference under shape-constraints, such as monotonicity, convexity, log-concavity, etc., which arise naturally in a wide range of applications \cite{GJ2014}. One advantage of such methods is that they allow for nonparametric estimation without using tuning parameters. However, a combination of smoothing and shape-constrained estimation often leads to better finite sample performance \cite{lopuhaa2017smooth,lopuhaa2018smoothed}.
	
	In this paper we introduce an estimation method for the monotone single-index mixture cure model that is based on the profile maximum likelihood principle and techniques from isotonic regression in combination with kernel smoothing. Despite the fact that the monotone single-index model has already been studied in the literature, its use within the mixture cure model has some unique features that make the problem more challenging (see Section 3 for a more detailed discussion). 
	First, in contrast to the standard monotone single-index models, including the current status linear regression model, where the response is directly observed, the cure status in mixture cure model is unknown for the censored subjects. As a result, the likelihood has a more complicated expression and iterative procedures such as the EM algorithm are required to solve the optimization problem. Secondly, apart from the coefficients of the index and the link function, our model contains additional parameters (finite and infinite dimensional) because of the extra latency component. Challenges also arise when studying the theoretical properties of the estimators, which are established in a less straightforward manner as compared to the monotone single-index model or the current status model, see Section 4 for a more detailed discussion. We study consistency of the proposed estimator, which to our best knowledge has not been investigated even for the single-index mixture cure model without monotonicity constraints \cite{AKL2019}. We illustrate through simulations that imposing monotonicity improves the behavior of the estimator and makes it more stable with respect to the choice of the bandwidth compared to the smooth non-monotone estimator. In particular, one does not need to use time consuming bandwidth selection procedures since a simple bandwidth choice performs reasonably well. In addition, we observe that our decision  to incorporate a smoothing step to the isotonic estimation of the link function is motivated by its improved practical performance for finite sample sizes.
	
	The paper is organized as follows. Section 2 describes the monotone single-index mixture cure model and the conditions for model identification. Section 3 introduces the estimation procedure, while Section 4 focuses on establishing consistency of the estimator. The finite sample properties of the proposed method are investigated through a simulation study and the results are reported in Section 5. Finally, an illustration of the practical use through a study of a medical dataset of melanoma cancer patients is provided in Section 6. The proofs  and additional simulation results can be found in the Appendix. Software in the form of R code is available on the GitHub repository \url{https://github.com/tp-yuen/msic}.

	\section{Model description \label{sec:model}}
	Let $T$ be a nonnegative random variable denoting the survival time, i.e time until occurrence of an event of interest, which can be equal to infinity indicating the possibility of cure. Under the assumption that the survival time is subject to random right censoring, we observe the follow-up time $Y = \min{\rbracket{T, C}}$ and the censoring indicator $\Delta = \indicator{T \leq C}$, where $C$ is the censoring time. {Since the duration of the studies is in practice limited, we assume that $C$ has bounded support.} As a result of censoring, the cure status $B = \indicator{T < \infty}$ is a latent variable and the cured subjects cannot be distinguished from the censored uncured ones. Since the cure probability and the survival time of the uncured do not necessarily depend on the same predictors, we use two sets of covariates $\mat{X}\in\real[d]$ and $\mat{Z}\in\real[q]$ that can possibly be the same or partially/completely different. {We assume that $C$ and $T$ are conditionally independent given the covariates $(\mat{X},\mat{Z})$, which is a rather standard assumption in survival analysis. }%
	In the mixture cure model the survival function is given by
	\begin{equation}
		\survival{t} = \prob{T > t | \mat{X} = \mat{x}, \mat{Z} = \mat{z}}=1 - \uncuredrate + \uncuredrate\uncuredsurvival{\mat{z}}{t}{},
	\end{equation}
	where $\uncuredrate = \prob{B = 1 | \mat{X} = \mat{x}}$ is the conditional uncure probability (incidence) and $\uncuredsurvival{\mat{z}}{t}{}$ is the conditional survival function for the uncured (latency). Note that $S_u$ is a proper survival function, while $\lim_{t\to\infty}\survival{t} = 1 - \uncuredrate$. In terms of distribution functions we have  $F(t|\mat{x},\mat{z}) = \uncuredrate F_u(t|\mat{z})$.
	We consider a monotone single-index model for the incidence component, that is 
	\begin{equation}
		\label{eq:incidence}
		\uncuredrate = \link[\tr{\mat{\gamma}}_{0}\mat{x}]{0}{},
	\end{equation}
	for some unknown regression coefficient $\mat{\gamma}_{0}\in\real[d]$ and an unknown link function $\varphi_{0}$ belonging to the set
	{$\family{M}:= \set{\map{\varphi}{\real}{\rbracket{0,1}}}{\varphi~\text{is monotone non-decreasing}}$.} For the latency we assume a Cox proportional hazard model, i.e.  
	\begin{equation}
		\label{eq:uncured_survival}
		\uncuredsurvival{\mat{z}}{t}{} =\prob{T > t | B=1,\mat{Z} = \mat{z}}=\exp\sbracket{-\Lambda_{0}\rbracket{t}\exp\rbracket{\tr{\mat{\beta}}_{0}\mat{z}}},
	\end{equation} 
	where $\Lambda_0$ denotes the baseline cumulative hazard function and $\mat{\beta}_{0}\in\real[q]$ is a vector of regression parameters. Both $\Lambda_{0}$ and $\beta_0$ are left unspecified.%
	
	A crucial issue for both the single-index and the mixture cure model is identifiability meaning that \[
	\begin{split}
	&l(Y,\Delta,X,Z;\varphi,\mat\gamma,\mat\beta,\Lambda)=	l(Y,\Delta,X,Z;\tilde\varphi,\tilde{\mat\gamma},\tilde{\mat\beta},\tilde\Lambda)\text{ a.s.}\\
	& \implies \varphi=\tilde{\varphi},\quad	\mat\gamma=\tilde{\mat{\gamma}},\quad\mat\beta=\tilde{\mat{\beta}},\quad\Lambda=\tilde\Lambda,
	\end{split}
	\]  where $	l(Y,\Delta,X,Z;\varphi,\mat\gamma,\mat\beta,\Lambda)$ denotes the log-likelihood of the model given the parameters. 
	For the single-index model, the parameters $(\varphi_0,\mat{\gamma}_0)$ are not identifiable without further restrictions since for any $a\in\real$ we can define $\tilde{\varphi}(t)=\varphi_0(t/a)$, $\tilde{\mat\gamma}\mat=a\mat{\gamma}_0$ and have $\tilde\varphi(\tilde{\mat\gamma}^T\mat{x})=\varphi_0(\mat\gamma_0^T\mat{x})$. For the mixture cure model, the parameters are not identifiable if the follow-up of the study does not contain the support of the event times since it is not possible to distinguish the event of being cured from the one of being uncured with survival time larger than the follow-up of the study. To guarantee identifiability of the model  we require the following set of assumptions, where $\family{X}$ denotes the support of the covariate $\mat{X}$ and $\family{I}_{\mat{\gamma}} := \set{\tr{\mat{\gamma}}\mat{x}}{\mat{x}\in\family{X}}$ denotes the support of the index $\tr{\mat{\gamma}}\mat{X}$. 
	\begin{enumerate}[label=(A\arabic*), series=assumptions]
		\item	\label{enum:si_assumptions}
		\begin{enumerate}[label=(\roman*)]
			\item\label{enum:link_assumptions}	$\varphi_{0}$ is differentiable and not constant on $\family{I}_{\mat{\gamma}_{0}}$.
			\item The parameter $\mat{\gamma}_{0}$ does not contain an intercept and it belongs to the $d-1$ dimensional unit sphere $\family{S}_{d-1} := \{\mat{\gamma} \in \real[d]\colon{\norm{\mat{\gamma}}{2} = 1}\}$ with respect to the Euclidean norm $\norm{\cdot}{2}$. 
			\item	The covariate $\mat{X}$ contains at least one continuous variable and the continuous components have a joint probability density function.
			\item	$\family{X}$ is not contained in a proper linear subspace of $\real[d]$.
			\item	$\family{I}_{\mat{\gamma}_{0}}$ is not divided into disjoint intervals by different values of the discrete components.
		\end{enumerate}
		\item	
		\begin{enumerate}[label=(\roman*)]
			\item	 $\mat{\beta}_{0}$ does not have an intercept term.
			\item	The covariance matrix of $\mat{Z}$ has full rank.
		\end{enumerate}
		\item\label{enum:latency_assumptions}	
		\begin{enumerate}[label=(\roman*)]
			\item\label{enum:cure_threshold}	There exists a cure threshold $\tau_{0} < \infty$ such that $T > \tau_{0} \iff T = \infty$. Moreover $\prob{C > \tau_{0}|\mat{X},\mat{Z}} > 0$ for almost all $\mat{X}$ and $\mat{Z}$.
			\item	The incidence $\uncuredrate$ in \eqref{eq:incidence} satisfies $0 < \uncuredrate < 1$ for all $\mat{x} \in \family{X}$.
		\end{enumerate}
	\end{enumerate}
	Assumptions \ref*{enum:si_assumptions}-\ref*{enum:latency_assumptions} are almost identical to the ones in \cite{AKL2019}. Note that we do not need the assumption that the sign of the first entry of $\mat{\gamma}_{0}$ is fixed because we are fixing the direction of monotonicity for the link function. 
	Assumption \ref*{enum:si_assumptions} entails the identifiability of the monotone single-index model in \eqref{eq:incidence} (Theorem 2.1 in \cite{H2009}). Given that the incidence in \eqref{eq:incidence} is identifiable, using the same argument as in the proof of Proposition~1 of \cite{AKL2019}, it follows that the monotone-single-index/Cox mixture cure model is identifiable. 
	
	\section{Estimation method \label{sec:model_est}}
	Assume that we have i.i.d. realizations $\rbracket{y_{i},\delta_{i},\mat{x}_{i},\mat{z}_{i}}$, $i=1,\cdots,n$ of $\rbracket{Y, \Delta, \mat{X}, \mat{Z}}$. The observed likelihood function of the mixture cure model is given by
	\begin{align}
		\begin{split}
			\label{eq:mcm_likelihood}
			L_{n}\rbracket{\mat{\gamma}, \mat{\beta}, \Lambda, \varphi} 
			= \prod_{i=1}^{n}
			&\sbracket{
				\link[\tr{\mat{\gamma}}\mat{x}_{i}]{}{}
				\basehazard[]{y_{i}}e^{\tr{\mat{\beta}}\mat{z}_{i}}
				e^{
					-\basecumhazard[]{y_{i}}\exp\rbracket{\tr{\mat{\beta}}\mat{z}_{i}}
				}
			}^{\delta_{i}}\\%
			&\times
			\sbracket{
				1 - \link[\tr{\mat{\gamma}}\mat{x}_{i}]{}{} + \link[\tr{\mat{\gamma}}\mat{x}_{i}]{}{}
				e^{
					-\basecumhazard[]{y_{i}}\exp\rbracket{\tr{\mat{\beta}}\mat{z}_{i}}
				}
			}^{1-\delta_{i}}.
		\end{split}
	\end{align}
	When the link function is assumed to be known, the parameters $(\mat\gamma,\mat\beta,\Lambda)$ are estimated via the maximum likelihood principle. Here, we treat $\varphi$ as a nuisance (infinite dimensional) parameter and, for any fixed $(\mat\gamma,\mat\beta,\Lambda)$, we construct a smooth monotone estimator of $\varphi$. Finally, we consider a new likelihood criteria with the plug-in estimator of $\varphi$ and apply the maximum likelihood method.
	Hence, the estimation procedure consists of the following three steps:
	\begin{enumerate}
		\item	
		For fixed $\mat{\theta} = \rbracket{\mat{\gamma}, \mat{\beta}, \Lambda}$, we estimate the link by
		\begin{equation}
			\label{eq:monotone_link_mle}
			\hat{\varphi}_{n, \mat{\theta}} 
			=
			\argmax_{\varphi \in \family{M}_{\epsilon^{\prime}}}
			\function{L}{n}{}{\mat{\gamma},\mat{\beta},\Lambda,\varphi},
		\end{equation}
		where $\family{M}_{\epsilon^\prime} = \set{\map{\varphi}{\real}{\sbracket{\epsilon^\prime,1 - \epsilon^\prime}}}{\varphi~\text{is monotone non-decreasing}}$ and $\epsilon^{\prime} > 0$ is a fixed small constant for a truncation on the uncured probability. {See Remark \ref{re:truncation}  below for a discussion on this truncation and the choice of $\epsilon'$.} The estimator  $\hat{\varphi}_{n,\mat{\theta}} $ is computed using the EM algorithm and techniques from isotonic estimation as explained in Subsection \ref{seq:computation} below. The maximizer is not unique but it is uniquely defined at the points $\mat\gamma^T\mat{x}_i$, $i=1,\dots,n$. We consider $\hat{\varphi}_{n,\mat{\theta}} $ to be a left-continuous step function that extends  constantly to the entire real line. 
		\item A kernel smoothed version of $\hat{\varphi}_{n,\mat{\theta}}$ is defined by
		\begin{equation}
			\label{eq:smooth_isotonic_estimator}
			\function{\hat{\varphi}}{n,\mat{\theta}}{s}{u}
			=
			\int_{u-h}^{u+h}\frac{1}{h}
			\function{k}{}{}{\frac{u - t}{h}}\function{\hat{\varphi}}{n,\mat{\theta}}{}{t}
			\dd t,
		\end{equation}
		where $k$ is a symmetric kernel with bounded support $\sbracket{-1, 1}$ and $h$ is a chosen bandwidth. By definition and the monotonicity of $\hat\varphi_{n,\mat\theta}$ it follows that $\hat{\varphi}_{n,\mat\theta}^s$ is a smooth non-decreasing function. We illustrate in Appendix~\ref{sec:smoothing} that this smoothing step indeed improves the behavior of the estimator. For the bandwidth $h$ we follow a common choice in the literature of smooth isotonic estimators by taking $h=rn^{-1/5}$, where $r$ is the range of the observed index $\mat\gamma^T\mat{X}$. In Appendix~\ref{sec:bandwidth} we investigate the sensitivity of the estimators with respect to the choice of the bandwidth and conclude that, despite not being the optimal bandwidth, this is a satisfactory and quick solution. We do not apply any boundary correction for the kernel estimator but instead extend the isotonic estimator $\hat{\varphi}_{n,\mat{\theta}}$ to be constant outside of the range of the observed data.
		\item Using the plug-in approach and the maximum likelihood principle,  $\mat{\theta}_0$ is estimated by 
		\begin{equation}
			\label{eq:theta_mle}
			\hat{\mat{\theta}}_{n} = \rbracket{\hat{\mat{\gamma}}_{n}, \hat{\mat{\beta}}_{n}, \hat{\Lambda}_{n}} 
			=
			\argmax_{\mat{\theta} 
			}\function{L}{n}{}{\mat{\gamma},\mat{\beta},\Lambda,\hat{\varphi}_{n,\mat{\theta}}^{s}},
		\end{equation}
		where {the maximization is done over $\mat\gamma\in \family{S}_{d-1}$, $\mat{\beta}\in\real[q]$ and non-decreasing positive functions~$\Lambda$. }  The estimator is computed iteratively using the EM algorithm as explained in Subsection \ref{seq:computation} below. The whole estimation procedure is described in  Algorithm \ref*{algo:model_estimation_algo} 
		in Appendix~\ref{sec:algorithms}. 
	 As in the standard logistic/Cox mixture cure model, we impose the zero tail constraint meaning that the observations in the plateau are assumed to be cured. This corresponds to setting $S_u(t|\mat{z};\Lambda,\mat\beta)=\exp(
		-\basecumhazard[]{t}\exp(\tr{\mat{\beta}}\mat{z}))=0$ for $t>y_{(r)}$ where $y_{(r)}$ denotes the largest observed event time.
	\end{enumerate}
	
	\subsection{Comparison with similar problems in the literature}
	Before explaining the computation of the proposed estimators, we comment on how our problem and method relate to the existing literature on the standard monotone single-index model \cite{BDJ2019,balabdaoui2019score,balabdaoui2021profile,groeneboom2019estimation}, which assumes that $\expect[\sbracket]{Y\mid \mat{X}}=\psi_0(\mat\alpha_{0}^T\mat{X})$ for some unknown $\mat\alpha_{0}$ and monotone link function $\psi_0$. Different methods for estimation of $(\mat\alpha_0,\psi_0)$ have been proposed based on the least-squares principle and adaptations of it without imposing any smoothness assumptions. The main idea is  the following. For fixed $\mat\alpha$, one can minimize the least squares criterion
	$
	h_n(\psi,\mat\alpha)=\frac1n\sum_{i=1}^n \{Y_i-\psi(\mat\alpha^T\mat{X}_i)\}^2$
	with respect to $\psi$ on the class of monotone functions, which gives a $\mat\alpha$-dependent function $\psi_{n,\mat\alpha}$. In a second step, the function 
	$h_n(\psi_{n,\mat\alpha},\mat\alpha)$ is then minimized over $\mat\alpha$. This would be the standard profile least squares estimator. Note that since this criterion function for $\mat\alpha$ is not smooth but piecewise constant, the estimator of $\mat\alpha$ is not unique. Moreover, alternative ways to estimate $\mat\alpha$ in the second step have been proposed by using the score approach and computing the zero-crossings of 
	$
	\frac1n\sum_{i=1}^n \{Y_i-\psi_{n,\mat\alpha}(\mat\alpha^T\mat{X}_i)\}\mat{X}_i
	$
	or minimizing its squared norm.
	If the criterion function was continuous in $\mat\alpha$, these alternative approaches would result in the same least squares estimator. 
	
	The current status linear regression problem can also be seen as a  monotone single-index model where the link is actually a distribution function \citep{groeneboom2019estimation}. In that setting estimation can be performed via the maximum likelihood principle, again by first maximizing the likelihood for a fixed index on the class of distribution functions and then maximizing with respect to $\mat\alpha$ or solving score equations \citep{groeneboom2018current}. 
	
	In our setting, the response variable that corresponds to the single-index model for the incidence is the latent cure status $B$. The fact that $B$ is  not always observed makes the use of the least-squares approach not suitable. Hence, our method is based on the maximum likelihood principle similarly to the one for the current status linear regression problem. However, in the current status model, the response (the current status $\Delta$) is observed and the only unknown parameters are the index and the link function. In our model,  the presence of additional unknown finite and infinite dimensional parameters $\mat\beta,\Lambda$ makes the estimation problem much more challenging. In particular, both optimization problems in \eqref{eq:monotone_link_mle} and \eqref{eq:theta_mle} cannot be solved directly but only through iterative procedures such as the EM algorithm. To the best of our knowledge, this is the first case for which a maximum likelihood estimator under monotonicity constraints, as in \eqref{eq:monotone_link_mle},  cannot be characterized explicitly. In the proof of Proposition~\ref{prop:monotone_link_est} below we comment that, even if one would try to use the standard techniques from isotonic estimation to characterize the maximizer as the left derivative of a greatest convex minorant, would end up with an iterative procedure that is the same as the EM algorithm.  
	In addition, we include a smoothing step which leads to a smooth monotone estimator of the link function and a continuous criterion for estimation of $\mat\theta$ in the next step. In Appendix~\ref{sec:smoothing} we illustrate that smoothing indeed improves the performance of the estimator compared to the monotone (non-smooth) estimator. As in our case, a truncation is also needed in the current status regression problem in order to avoid the link function from being close to 0 and 1. However, since for that setting the link is a distribution function which necessarily obtains the values 0 and 1, the truncation is imposed for the likelihood criterion excluding the extreme observations. For our model, given the assumption (A3)(ii), it seems easier and more reasonable to restrict to link functions that are bounded away from 0 and 1.

	\subsection{Computation of the estimators}
	\label{seq:computation}
	Unlike the Cox proportional hazard model \citep{C1972}, for which the regression coefficients $\mat{\beta}$ can be estimated using a profile likelihood approach independently of the baseline cumulative hazard function $\Lambda$, the mixture cure model does not possess a likelihood function that can take advantage of such approach due to the latent uncure status $B$. The maximization problems in \eqref{eq:monotone_link_mle} and \eqref{eq:theta_mle} are solved via the expectation-maximization (EM) algorithm as in the standard logistic/Cox mixture cure model \cite{ST2000}.%
	
	\subsubsection*{EM Algorithm}
	The uncure status $B_i$ of the i-th subject is  $B_i=1$ (uncured) if $\delta_{i} = 1$ and  is unknown otherwise. Given the observed data $\mat{v}_{i} = \rbracket{y_{i},\delta_{i},\mat{x}_{i},\mat{z}_{i}}$, $i=1,\cdots,n$, the complete-data likelihood function is
	\begin{align}
		\begin{split}
			\label{eq:mcm_complete_likelihood}
			L_{nc}\rbracket{\mat{\gamma}, \mat{\beta}, \Lambda, \varphi} =			&
			\prod_{i=1}^{n}
			\varphi(\tr{\mat{\gamma}}\mat{x}_{i})
			^{B_{i}}
			[1 - 	\varphi(\tr{\mat{\gamma}}\mat{x}_{i})]^{1-B_{i}} \\
			&\times 	\prod_{i=1}^{n}
			\sbracket{
				\basehazard[]{y_{i}}e^{\tr{\mat{\beta}}\mat{z}_{i}}
				e^{
					-\basecumhazard[]{y_{i}}\exp\rbracket{\tr{\mat{\beta}}\mat{z}_{i}}
				}
			}^{\delta_{i}B_{i}} \\
		&\times\prod_{i=1}^{n}\sbracket{
				e^{
					-\basecumhazard[]{y_{i}}\exp\rbracket{\tr{\mat{\beta}}\mat{z}_{i}}
				}
			}^{(1-\delta_{i})B_{i}}.
		\end{split}
	\end{align}%
	
	In the $(k+1)$-th iteration of the EM algorithm the parameters are updated as follows. 	The E-step of the EM algorithm computes the conditional expectation of the complete-data log-likelihood $\log L_{nc}$ with respect to the uncured status $B$ given the parameters of the previous iteration 
	and the observed data. 
	Specifically, by virtue of the partially observed nature of $B$ and its linearity in the complete-data log-likelihood function, the E-step is equivalent to computing
	\begin{equation}
		\label{eqn:w_i_EM}
		\begin{split}
	w_{i}^{(k)}&=	\mathbb{E}^{(k)}[{B_{i} \mid \mat{v}_{i}}]\\
	&= \delta_{i} + \rbracket{1 - \delta_{i}}\frac{\varphi^{(k)}({\mat{\gamma}^{(k)T}}\mat{x}_{i})S_u(y_i|\mat{z}_{i};\mat{\beta}^{(k)},\Lambda^{(k)})}{1-\varphi^{(k)}({\mat{\gamma}^{(k)T}}\mat{x}_{i})+\varphi^{(k)}({\mat{\gamma}^{(k)T}}\mat{x}_{i})S_u(y_i|\mat{z}_{i};\mat{\beta}^{(k)},\Lambda^{(k)})},	
		\end{split}
		\end{equation}
	where the expectation $	\mathbb{E}^{(k)}$ is computed using the parameters ${\mat{\theta}^{(k)}, \varphi^{(k)}}$ and $S_{u}(\cdot|\cdot;\mat{\beta},\Lambda)$ is obtained according to \eqref{eq:uncured_survival}. Substituting $B_{i}$ with $w_{i}$ in \eqref{eq:mcm_complete_likelihood}, we obtain the expected complete-data likelihood
	\begin{equation}
		\label{eq:mcm_expected_complete_likelihood}
		\begin{split}
&\prod_{i=1}^{n}\varphi(\tr{\mat{\gamma}}\mat{x}_{i})^{w_{i}}[1 - \varphi(\tr{\mat{\gamma}}\mat{x}_{i})]^{1-w_{i}} 
\prod_{i=1}^{n}f_{u}({y_{i}|\mat{z}_{i}})^{\delta_{i}w_{i}} S_u(y_i|\mat{z}_{i})^{(1-\delta_{i})w_{i}}\\
&=
\tilde{L}_{nc}^{1}\rbracket{\mat{\gamma}, \varphi} 
\tilde{L}_{nc}^{2}\rbracket{\mat{\beta}, \Lambda},
		\end{split}
	\end{equation}
	where, to simplify the notation, we have denoted $w_i^{(k)}$ by just $w_i$.	
	The M-step of the algorithm consists in maximizing the expected complete-data likelihood with respect to the parameters of interest. Specifically, in \eqref{eq:monotone_link_mle}, we maximize with respect to  $\varphi$ over $\family{M}_{\epsilon^{\prime}}$ while keeping $\rbracket{\mat{\gamma}, \mat{\beta}, \Lambda}$ fixed in all iterations, while in \eqref{eq:theta_mle} we maximize over	
	$\rbracket{\mat{\gamma}, \mat{\beta}, \Lambda}$ for a given $\varphi=\hat\varphi^s_{n,\mat{\theta}}$. {To simplify the maximization problem in \eqref{eq:theta_mle}, the link function could be kept fixed, equal to the one obtained in the previous iteration of the EM algorithm  with $\mat{\theta}=\mat{\theta}^{(m-1)}$. However, we observe that in some cases such procedure does not behave well computationally since the link estimator is very sensitive to the estimator of $\mat\gamma$. On the other hand, the link estimate seems to be stable with respect to small changes of the latency parameters $\mat\beta$ and $\Lambda$ from one iteration to the other. Hence, we only fix $\mat\beta$ and $\Lambda$ as the estimates of the previous iteration and allow $\varphi=\hat\varphi^s_{n,\mat{\theta}}$ with $\mat{\theta}=(\mat\gamma,\mat\beta^{(m-1)},\Lambda^{(m-1)})$ to still depend on $\mat\gamma$.}
	From \eqref{eq:mcm_expected_complete_likelihood}, we can see that the expected complete-data likelihood can be factorized into two parts. The first part only consists of the parameters of the incidence part, while the second part  contains the parameters of the latency part only. Therefore we can maximize the likelihood for the two parts separately. 
	
	\subsubsection*{Monotone Link Estimator}
	Algorithm \ref*{algo:link_estimator_algo} in Appendix~\ref{sec:algorithms} describes the procedure of the EM algorithm for estimation of the monotone link function in \eqref{eq:monotone_link_mle}. The M-step of the EM algorithm is equivalent to the following maximization problem for a fixed $\mat\gamma$
	\begin{equation}
		\label{eq:link_mle}
		\maximize_{\varphi \in \family{M}_{\epsilon^{\prime}}} 
		\sum_{i=1}^{n} 
		\cbracket{w_{i} \log\link[\tr{\mat{\gamma}}\mat{x}_{i}]{}{} 
			+ \rbracket{1 - w_{i}}\log\sbracket{1 - \link[\tr{\mat{\gamma}}\mat{x}_{i}]{}{}}},
	\end{equation}
	which belongs to the class of order-restricted maximum likelihood estimation problems \citep{RWD1988}. In addition to the order restriction, a uniform bound restriction is imposed to the maximum likelihood estimation (MLE) of $\varphi$ for fixed $\mat{\gamma}$. Such class of order and uniform bound restricted problems is studied by \cite{H1997}. Using results from \cite{RWD1988} and \cite{H1997} we obtain the following characterization, a proof of which can be found in Appendix~\ref{sec:proofs}. 
	\begin{prop}
		\label{prop:monotone_link_est}
		The maximizer in \eqref{eq:link_mle} exists, it is not unique but it is uniquely defined at the ordered points $\mat\gamma^T\mat{x}_{(1)}<\dots<\mat\gamma^T\mat{x}_{(n)}$ with correspondent values  $\hat{\varphi}_i=\max(\epsilon',\min(\tilde{\varphi}_i,1-\epsilon'))$, where $(\tilde{\varphi}_1,\dots,\tilde{\varphi}_n)$ are the left derivatives of the greatest convex minorant of the cumulative sum diagram 
		$$\cbracket{\rbracket{0,0}, \rbracket{i, \sum_{j=1}^{i}w_{(j)}}, i=1,\cdots,n}.$$
		Here $w_{(j)}$ corresponds to the $j$-th order statistic of $\tr{\mat{\gamma}}\mat{x}_{i}$, $i=1,\dots,n$ .
	\end{prop}
	
	\begin{remark}
		\label{re:truncation}
		The truncation of the link function is introduced in order to avoid that it takes the extreme values $0$ and $1$ which would create theoretical problems with terms that explode to infinity. When $\epsilon'$ is chosen to be very small (for example of the order $10^{-6}$), in practice there would be almost no difference between the truncated and not truncated version of the link function since by construction $\tilde{\varphi}_i\in[0,1]$. Truncation is also 
		compatible with Assumption (A3)(ii) above. When the support of $\mat{x}$ is assumed to be bounded, there exists $\epsilon'$ such that the true link belongs to $\mathcal{M}_{\epsilon'}.$
	\end{remark}
	
	\subsubsection*{Incidence Regression Parameter Estimator}
	For the estimation of the incidence regression parameter $\mat\gamma$ in \eqref{eq:theta_mle}, the M-step of the EM algorithm is equivalent to the following maximization problem 
	\begin{equation}
		\label{eq:gamma_mle}
		\maximize_{\mat{\gamma} \in \family{S}_{d-1}} 
		\sum_{i=1}^{n} 
		{	\cbracket{w_{i} \log\varphi_{\mat\gamma}(\tr{\mat{\gamma}}\mat{x}_{i}) 
				+ \rbracket{1 - w_{i}}\log\sbracket{1 - \varphi_{\mat\gamma}(\tr{\mat{\gamma}}\mat{x}_{i})}},}
	\end{equation}
	{where $\varphi_{\mat\gamma}$ denotes the smooth monotone link estimate $\hat\varphi^s_{n,\mat\theta}$ for $\mat{\theta}=(\mat\gamma,\mat\beta^{(m-1)},$ $\Lambda^{(m-1)})$, i.e. the parameters $\mat\beta$, $\Lambda$ are fixed to the estimates of the previous iteration. }
	Here, we impose the identifiability constraint, $\norm{\mat{\gamma}}{2} = 1$, leading to a maximization problem with a nonlinear constraint. This problem can be solved by the augmented Lagrangian method. Such method first reformulates the problem as an unconstrained problem by introducing penalty terms for the equality constraints and solving this unconstrained problem by some interior-point algorithms. The penalty terms are then updated. These two steps are repeated until convergence. See Chapter 17 of \cite{NW2006} for more details.%
	
	\subsubsection*{Latency estimation}
	For the latency component, the M-step of the EM algorithm can  be performed as in \cite{ST2000}. Specifically, the estimator  $\widehat{\mat\beta}_n$ for $\mat{\beta}$ is computed using the profile likelihood approach as the maximizer of 
	\begin{equation*}
		\label{eq:partial_likelihood}
		\prod_{i=1}^{n}
		\rbracket{
			\frac{
				\exp(\tr{\mat{\beta}}\mat{z}_{i})
			}{
				\sum_{j \in \family{R}_{i}}
				w_{j}\exp(\tr{\mat{\beta}}\mat{z}_{j})
			}
		}^{\delta_{i}},	
	\end{equation*}
	where $\family{R}_{i}$ denotes the risk set just before time $y_{i}$. The  nonparametric estimator of $\Lambda$ is  given by
	\begin{equation*}
		\label{eq:npmle_Lambda}
		\function{\hat{\Lambda}}{n}{}{t}
		=
		\sum_{i: y_{i} \leq t}
		\frac{
			d_{i}
		}{
			\sum_{j \in \family{R}_{i}}
			w_{j}\exp(\tr{\widehat{\mat{\beta}}}_n\mat{z}_{j})},
	\end{equation*}
	where $d_{i}$ denotes the number of events at time $y_{i}$. As suggested in \cite{T1995}, the conditional survival function $\hat{S}_{u}(t|\mat{z})$ is set to zero when $t> y_{(r)}$.
	
	\section{Asymptotic properties \label{sec:prop}}
	We start by providing some technical intuition on the approach and illustrate the additional challenges that we face compared to the existing literature on the monotone single-index model. 
	We will derive consistency of our estimators using results from the theory of semiparametric M-estimation. Indeed, our estimator $\hat{\mat\theta}_n$ corresponds to the maximizer of an empirical criterion function that depends on an infinite dimensional nuisance parameter  $\varphi$. The unknown $\varphi$ is  replaced by a nonparametric estimator $\hat\varphi^s_{n,\mat\theta}$ depending on $\mat\theta= \rbracket{\mat{\gamma}, \mat{\beta}, \Lambda}.$  
	Denote the log-likelihood for a single observation by
	\begin{equation*}
		l(y, \delta, \mat{x}, \mat{z}; \mat{\theta}, \varphi) = \delta\log f_{u}(y|\mat{z}) + \delta\log\varphi(\mat\gamma^T\mat{x}) + (1 - \delta)\log[1 - \varphi(\mat\gamma^T\mat{x})F_u(y|\mat{z})].
	\end{equation*}
	Then we can write 
	\[
	\hat{\mat\theta}_n=\argmax_{\mat\theta} M_n(\mat\theta,\hat\varphi^s_{n,\mat\theta}),\qquad M_n(\mat\theta,\hat\varphi^s_{n,\mat\theta})=\frac{1}{n}\sum_{i=1}^nl\rbracket{y_i, \delta_i, \mat{x}_i, \mat{z}_i; \mat{\theta}, \hat\varphi^s_{n,\mat\theta}}.
	\]
	We will show in Proposition \ref{prop:existence_mle_theta} below that such maximizer exists and is finite. The asymptotic version of $M_n$ is given by
	$
	M(\mat\theta,\varphi_{\mat\theta})=\expect[\sbracket]{l\rbracket{Y, \Delta, \mat{X}, \mat{Z}; \mat{\theta}, \varphi_{\mat\theta}}},
	$
	where the infinite dimensional nuisance parameter is allowed to depend on $\mat\theta$.
	On the other hand, for fixed $\mat{\theta} = \rbracket{\mat{\gamma}, \mat{\beta}, \Lambda}$, we define $l_{\mat{\theta}}\rbracket{\varphi} := \expect[\sbracket]{l\rbracket{Y, \Delta, \mat{X}, \mat{Z}; \mat{\theta}, \varphi}}$ and 
	\begin{equation}
		\label{eqn:phi_theta}
		\varphi_{0,\mat{\theta}} := \argmax_{\varphi \in \family{M}_{\epsilon^{\prime}}}l_{\mat{\theta}}\rbracket{\varphi},
	\end{equation}
	where $\family{M}_{\epsilon^\prime} := \set{\map{\varphi}{\real}{\sbracket{\epsilon^\prime,1 - \epsilon^\prime}}}{\varphi~\text{is monotone non-decreasing}}$. This corresponds to the asymptotic version of the estimator $\hat\varphi_{n,\mat\theta}$ in \eqref{eq:monotone_link_mle}. Proposition \ref{prop:link_maximizer} below guarantees that the maximizer $\varphi_{0,\mat{\theta}}$ exists and is unique.  Note that from Assumptions \ref*{enum:si_assumptions}\ref*{enum:link_assumptions} and \ref*{enum:bounded_support} below, it follows that there exists $\epsilon>0$ such that $\epsilon\leq\varphi_0(\mat\gamma_0^T\mat{x})\leq 1-\epsilon$ for all $\mat{x}\in\mathcal{X}$. We assume that $\epsilon'$ in the definition of $\hat\varphi_{n,\mat\theta}$ is chosen such that $\epsilon'<\epsilon$. In this way, $\varphi_{0}\in\family{M}_{\epsilon'}$ and in particular we have that, if $\mat\theta=\mat\theta_0$, then $\varphi_{0,\mat\theta}=\varphi_0$ (Proposition \ref{prop:true_link}). We will show that
	$
	\mat\theta_0=\argmax_{\mat\theta}M(\mat\theta,\varphi_{0,\mat\theta}),
	$
	which is the foundation behind the estimation strategy. 
	Then, to obtain consistency of of the estimator for $\mat\theta_0$ we will check the conditions of Theorem 1 in \cite{DK2020}. Specifically we need that $\hat\varphi^s_{n,\mat\theta}$ is a consistent estimator for $\varphi_{0,\mat\theta}$ uniformly over $\mat\theta$, the empirical criterion function $M_n(\mat\theta,\varphi)$ is a good approximation of the asymptotic criterion $M(\mat\theta,\varphi)$ uniformly over $\mat\theta$ and $\varphi$, the function $M(\mat\theta,\varphi)$ is continuous with respect to $\varphi$ at $\varphi_{0,\mat\theta}$ uniformly over $\mat\theta$. Such conditions will be proved in Theorem \ref{prop:theta_estimator_consistency}.
	
	The function $\varphi_{0,\mat\theta}$ defined in \eqref{eqn:phi_theta} plays a fundamental role in the theoretical analysis of our estimators and the main challenges we face arise from the fact that we do not have an explicit characterization of this function. The counterpart of this function in the monotone single-index model is 
	$
	\psi_{\mat\alpha}(u)=\mathbb{E}[\psi_0(\mat\alpha_0^T\mat{X})\mid\mat\alpha^T\mat{X}=u]
	$
	(see for example equation (5) in \cite{balabdaoui2021profile}), while in the current status model is 
	$
	F_{\mat\beta}(u)=\mathbb{E}[F_0(T-\mat\beta_0^T\mat{})\mid T-\mat\beta^T\mat{X}=u]
	$
	(see equation 3.2 in \cite{groeneboom2018current}). In both cases, this function can be seen as the expected value of the true single-index model when we fix the index to a given value.  Having this explicit characterization makes it easier to deal with this function and in particular, properties such as continuity, differentiability of $\psi_{\mat\alpha}$  or $F_{\mat\beta}$ can be derived from assumptions on the true link function. In our case, we can characterize $\varphi_{0,\mat\theta}$ as 
	\begin{equation}
		\label{eqn:characterization_phi_theta}
		\varphi_{0,\mat\theta}(u)=\expect[\sbracket]{\varphi_0(\mat\gamma_0^T\mat{X})\mid\mat\gamma^T\mat{X}=u}
	\end{equation}
	only if the covariates $\mat{X}$ and $\mat{Z}$ are independent, censoring is independent of all the other variables and only for $\mat\theta=(\mat\gamma,\mat\beta_0,\mat\Lambda_0)$ (see details in Appendix~\ref{sec:proofs}). However, such assumptions are too strong for practical purposes since it is quite common in particular to have $\mat{X}=\mat{Z}$. Hence, we prefer not to restrict ourselves to such scenario.  Without such characterization, even just arguing continuity of $\varphi_{0,\mat\theta}$ is quite challenging and technical (see Proposition \ref{prop:link_continuity}) and requires assumptions that are more difficult to interpret such as assumption \ref*{enum:expectation_ratio_continuity} below. For a fixed $\mat\theta$, the smooth kernel estimator $\hat{\varphi}^s_{n,\mat\theta}$ is an estimator of $\varphi_{0,\mat\theta}$, hence in order to obtain the rate of convergence of the estimators, one would need also (twice) differentiability of $\varphi_{0,\mat\theta}$. Given the technicalities of proving by contradiction that $\varphi_{0,\mat\theta}$ is continuous, we do not explore this direction further. 
	
	The second challenge arises from the fact that our parameter $\mat\theta$ contains not only the index $\mat\gamma$ but also the latency parameters $\mat\beta$, $\Lambda$. In particular, $\mat\theta$ is not finite dimensional. The theory of semiparametric M-estimators in \cite{DK2020} allows $\mat\theta$ to be infinite dimensional only for the consistency part but obtaining the rate of convergence requires an Euclidean parameter $\mat\theta$. Hence, one would first need to extend the standard results of semiparametric M-estimation to this scenario. Note that obtaining the limit distribution of the latency parameters $\hat{\mat\beta}_n,\hat\Lambda_n$ is equally challenging since they are dependent on $\hat\varphi^s_{n,\mat\theta}$ and $\hat{\mat\gamma}_n$ (see Definition \eqref{eq:theta_mle}) and cannot be dealt with separately. Because of these two main issues, within this paper, we focus on the consistency property of the estimators and leave the rate of convergence and asymptotic normality to be subject of future research.
	
	The results that we previously described intuitively are formulated rigorously  below. First we list the required assumptions. 
	The identifiability assumptions \ref*{enum:si_assumptions}-\ref*{enum:latency_assumptions} are assumed throughout this section.  In what follows $B_{\mat{c}}(r)$ will denote a closed ball of center $\mat{c}$ and radius $r$ in a given metric space, which is some $\real[d]$ with the Eucledian norm if not specified otherwise.  
	\begin{enumerate}[label=(A\arabic*), resume=assumptions]
		\item\label{enum:bounded_support}
		The covariates $\mat{X}$ and $\mat{Z}$ have bounded supports $\family{X}$ and $\family{Z}$ respectively. That is, $\family{X} \subset B_{\mat{0}}(r_{1})$ for some $r_{1} > 0$ and $\family{Z} \subset B_{\mat{0}}(r_{2})$ for some $r_{2} > 0$.
		\item\label{enum:index_density_assumption} There exists $\delta_0>0$ such that, for all $\mat{\gamma} \in B_{\mat\gamma_0}(\delta_{0})$, 
		$\tr{\mat{\gamma}}\mat{X}$ has a density with respect to the Lebesgue measure. We denote such density by $g_{\tr{\mat{\gamma}}\mat{X}}(\cdot)$.
		\item\label{enum:compact_param_space}
		$\mat{\gamma}_{0}$ and $\mat{\beta}_{0}$ lie in the interior of  compact sets $\Gamma=\mathcal{S}_{d-1}\cap B_{\mat{\gamma}_{0}}(\delta_{0})$ and $\mathcal{B}$ respectively. 
		\item\label{enum:Lambda} The function $ \Lambda_0(t)$ defined on $[0,\tau_0)$ is non-decreasing and $\Lambda_0(\tau_0-)<\infty$. We define $\Lambda_0(\tau_{0})=\Lambda_0(\tau_0-)$, i.e. extend $\Lambda$ at $\tau_0$ by left-continuity.
	\end{enumerate}
	Consider $\mat\theta\in \Theta=\Gamma\times\mathcal{B}\times\mathcal{D}$, where $\mathcal{D}$ is the space of nondecreasing functions $\Lambda$ on $[0,\tau_0]$ such that $\Lambda(0)=0$ and $\Lambda(\tau_0)<\infty$.
	\begin{enumerate}[label=(A\arabic*), resume=assumptions]
		\item\label{enum:expectation_ratio_continuity}	 
		For any $\mat\theta\in\Theta$, the function
		\begin{equation*}
			\mapping{u}{
				\expect[\sbracket]{\left.
					\frac{
						1 - \link[\tr{\mat{\gamma}}_{0}\mat{X}]{0}{}F_{u,0}\rbracket{Y | \mat{Z}}
					}{
						1 - \link[\tilde{u}]{0,\mat{\theta}}{}F_{u}\rbracket{Y | \mat{Z}}
					} \right\vert \tr{\mat{\gamma}}\mat{X} = u
				}
			}
		\end{equation*}
		is continuous for all {$\tilde{u}\in \mathcal{I}_{\mat{\gamma}}$}, where $F_{u,0}\rbracket{Y | \mat{Z}} = 1 - \uncuredsurvival{\mat{Z}}{Y}{}$ with the true parameters $\mat{\beta}_{0}$ and $\Lambda_{0}$ as in \eqref{eq:uncured_survival}, $F_{u}(Y | \mat{Z})= 1 - S_u(Y|\mat{Z})$ with parameters $\mat{\beta}$ and~$\Lambda$.
		\item\label{enum:cont_diff_index_density}
		The family of density functions $\{g_{\tr{\mat{\gamma}}\mat{X}}:\mat{\gamma} \in \Gamma\}$ is uniformly equicontinuous. 
		\item\label{enum:bound_dens} The density functions $g_{\tr{\mat{\gamma}}\mat{X}}$ for $\mat\gamma\in\Gamma$ and $g_{\tr{\mat{\beta}}\mat{Z}}$ for $\mat\beta\in\mathcal{B}$ are uniformly bounded from above by some positive constants $\bar{q}_1,$ $\bar{q}_2$ respectively. 
	\end{enumerate}
	
	Assumptions \ref*{enum:bounded_support}-\ref*{enum:compact_param_space} are standard assumptions made also in the standard single index model and the current status model. Assumption \ref*{enum:Lambda} is a standard assumption of the mixture cure model, see \cite{L2008,musta2022presmoothing}. Together with assumption \ref*{enum:latency_assumptions}, it essentially means that the distribution of the survival times for the uncured subjects has a jump at $\tau_0$, i.e. there is positive probability for the event to happen at $\tau_0$. This is mainly a technical condition for consistency of the estimators (see discussion in \cite{musta2022presmoothing}) and the probability mass at $\tau_0$ can be arbitrarily small, hence reasonable in practice.
	Assumption \ref*{enum:expectation_ratio_continuity} is needed to guarantee the continuity of $\varphi_{0,\mat\theta}$ since we do not have an explicit expression for such function. If there was no latency component $F_u$ then this assumption reduces to continuity of $\mathbb{E}[\link[\tr{\mat{\gamma}}_{0}\mat{X}]{0}{}\mid \tr{\mat{\gamma}}\mat{X} = u]$, which is standard in the single-index model. 
	Assumptions \ref*{enum:cont_diff_index_density} and \ref*{enum:bound_dens} are required in order to get uniform consistency of the link estimate. \ref*{enum:cont_diff_index_density} is for example satisfied if the density functions $\function{g}{\tr{\mat{\gamma}}\mat{X}}{}{\cdot}$ are continuously differentiable on their support $\family{I}_{\mat{\gamma}}$ for all $\mat{\gamma} \in \Gamma$ with uniformly bounded derivative.
	
	\begin{prop}
		\label{prop:link_maximizer}
		Suppose that Assumptions \ref*{enum:bounded_support} and \ref*{enum:index_density_assumption} 
		hold. Then, for any  {$\mat{\theta}\in \Theta$}, the maximizer $\varphi_{0,\mat{\theta}}$ in \eqref{eqn:phi_theta} exists and is unique.
	\end{prop}
	
	\begin{prop}
		\label{prop:true_link}
		For $\mat{\theta} = \rbracket{\mat{\gamma}_{0},\mat{\beta}_{0},\Lambda_{0}}$, we have $\varphi_{0,\mat{\theta}} = \varphi_{0}$.
	\end{prop}
	
	\begin{prop}
		\label{prop:existence_mle_theta}
		Suppose that Assumptions \ref*{enum:bounded_support} and \ref*{enum:compact_param_space} 
		hold. Then the maximum likelihood estimator $(\hat{\mat{\gamma}}_{n},\hat{\mat{\beta}}_{n},\hat{\Lambda}_{n})$ defined as in \eqref{eq:theta_mle} exists and is finite.
	\end{prop}
	
	\begin{prop}
		\label{prop:link_continuity}
		Suppose that Assumptions \ref*{enum:bounded_support}, \ref*{enum:index_density_assumption} and \ref*{enum:expectation_ratio_continuity} 
		hold. Then, for any $\mat\theta\in\Theta$, the function $\mapping{u}{\function{\varphi}{0,\mat{\theta}}{}{u}}$ is continuous.
	\end{prop}
	As in Lemma 1 of \cite{L2008}, using Assumption~\ref*{enum:bounded_support}, it can be shown that there exists $M>0$ such that $\sup_n\hat\Lambda_n(\tau_0)\leq M$ a.s.. Hence, we can restrict to $\mat\theta\in \tilde\Theta=\Gamma\times\mathcal{B}\times\tilde{\mathcal{D}}$, where $\tilde{\mathcal{D}}$ is the subset of $\mathcal{D}$ consisting of functions $\Lambda$ that are uniformly bounded by $M$.
	\begin{prop}
		\label{prop:link_mle_convergence}
		Suppose that Assumptions \ref*{enum:bounded_support}-\ref*{enum:Lambda} and \ref*{enum:bound_dens}
		 hold,
		then
		\begin{equation*}
			\prob[\sbracket]{
				\limit{n}{\infty}\sup_{\mat{\theta}\in\tilde\Theta}
				\int_{\family{X}}\cbracket{
					\function{\hat{\varphi}}{n,\mat{\theta}}{}{\tr{\mat{\gamma}}\mat{x}} - \function{\varphi}{0,\mat{\theta}}{}{\tr{\mat{\gamma}}\mat{x}}
				}^{2}\dd\function{Q}{\mat{X}}{}{\mat{x}} = 0
			} = 1,
		\end{equation*}
		where $Q_{\mat{X}}(\cdot)$ denotes the distribution function of $\mat{X}$.
	\end{prop}
	
	Let $k$ be a symmetric kernel density function with support $\sbracket{-1, 1}$ that satisfies $\function{k}{}{}{x} \leq K < \infty$ for all $x \in \sbracket{-1, 1}$. To simplify the notation, we define $\function{k}{h}{}{u} = h^{-1}\function{k}{}{}{u/h}$. Here $h = h_{n} > 0$ is a bandwidth that depends on the sample size $n$ and satisfies $h_{n} \converge 0$ as $n \converge \infty$. Since within this paper we do not investigate further the rate of convergence of the estimators, no additional restrictions are imposed on the bandwidth. However, it is known that the optimal order bandwidth for estimation of a twice differentiable function is $n^{-1/5}$ and, if the bandwidth is of order $n^{-1/3}$ or smaller, the smooth estimator is not much different from the non-smooth one. Hence, in practice we take $h=rn^{-1/5}$, where $r$ is the range of the index $\mat{\gamma}^{T}\mat{X}$. This is a common choice in the literature of  smooth isotonic estimators that is simple and behaves well. For a more detailed investigation of the role of the bandwidth see Appendix~\ref{sec:bandwidth}.
	\begin{prop}
		\label{prop:smooth_estimator_convergence}
		Suppose that Assumptions \ref*{enum:bounded_support}-\ref*{enum:bound_dens} 
		hold, then
		\begin{equation*}
			\prob[\sbracket]{
				\limit{n}{\infty}\sup_{\mat{\theta}\in\tilde\Theta}
				\int_{\family{X}}\cbracket{
					\function{\hat{\varphi}}{n,\mat{\theta}}{s}{\tr{\mat{\gamma}}\mat{x}} - \function{\varphi}{0,\mat{\theta}}{}{\tr{\mat{\gamma}}\mat{x}}
				}^{2}\dd\function{Q}{\mat{X}}{}{\mat{x}} = 0
			} = 1.
		\end{equation*}
	\end{prop}
	
	Next  we show that the estimator $\hat{\mat{\theta}}_{n}=(\hat{\mat\gamma}_n,\hat{\mat\beta}_n,\hat\Lambda_n)$ of $\mat{\theta}_{0}$ is weakly consistent.  
	\begin{theorem}
		\label{prop:theta_estimator_consistency}
		Suppose that Assumptions \ref*{enum:bounded_support}-\ref*{enum:bound_dens} 
		hold, then
		\begin{equation*}
			\Vert \hat{\mat\gamma}_n-\mat\gamma_0\Vert_2, \qquad\Vert\hat{\mat{\beta}}_n-\mat\beta_0\Vert_2\qquad\text{and}\qquad \sup_{t \in \sbracket{0, \tau_{0}}}|\hat\Lambda_{n}\rbracket{t} - \Lambda_{0}\rbracket{t}|
		\end{equation*}
		converge to zero in probability as $n\to\infty$. 
	\end{theorem}
	The following corollary shows the consistency of the estimated cure probabilities and the estimated survival function for the uncured.
	\begin{corollary}
		\label{prop:link_s_u_consistency}
		Suppose that Assumptions \ref*{enum:bounded_support}-\ref*{enum:bound_dens} 
		hold. We suppose further that $\varphi_{0}$ has bounded derivative on $\family{I}_{0}$. 
		Then, for any $\mat{z} \in \family{Z}$,
		\begin{equation*}
			\int_{\family{X}}\cbracket{
				{\hat{\varphi}}_{n,\hat{\mat{\theta}}_{n}}^{s}({\tr{\hat{\mat{\gamma}}}_{n}\mat{x}}) - \link[\tr{\mat{\gamma}}_{0}\mat{x}]{0}{}
			}^{2}\dd\function{Q}{\mat{X}}{}{\mat{x}}\qquad\text{and}\qquad \sup_{t \in \sbracket{0, \tau_{0}}}|\hat{S}_{u}(t|\mat{z}) - S_{u}(t|\mat{z})|
		\end{equation*}
		converge to zero in probability as $n\to\infty$, where $S_{u}(t|\mat{z})$ and $\hat{S}_{u}(t|\mat{z})$ are defined as in \eqref{eq:uncured_survival} using ${\mat{\beta}}_0$, $\Lambda_{0}$ and the estimated parameters $\hat{\mat{\beta}}_n$, $\hat\Lambda_{n}$ respectively.
	\end{corollary}

	\section{Simulation study \label{sec:sim}}
	In the simulation study, we consider different settings to evaluate  the finite sample behavior of the estimator proposed in \sref{sec:model_est} and compare it with the SIC method proposed in \cite{AKL2019}. We simulate 500 datasets from the mixture cure model introduced in \sref{sec:model}, 
	where $\mat{X} = \tr{\rbracket{X_{1}, \cdots, X_{4}}}$ contains four independent covariates: $X_{1} \sim U[0,1]$, $X_{2} \sim N(0, 1)$,  $X_{3}$ and $X_4$ are Bernoulli variables with parameters $0.3$ and $0.6$ respectively, and $\mat{Z} = \tr{\rbracket{Z_{1}, Z_{2}}}$ with $Z_{1} = X_{1}$ and $Z_{2} = X_{4}$. We consider a Weibull model  with parameters $\lambda = 1.5$ and $k = 2.2$ for the baseline distribution of the uncured subjects. 
	The random right censoring time $C$ follows the exponential distribution with rate $\lambda_{C}$. %
	We consider three experiments A, B and C, with three different non-decreasing link functions.
	\begin{equation*}
		{
			\varphi_A(u) = \frac{
				\exp(c + u)
			}{
				1 + \exp(c + u)
			},\,
		\varphi_B(u) = \frac{
				\exp[\psi(c,u)]
			}{
				1 + \exp[\psi(c,u)]
			},\,
		\varphi_C(u)= \frac{
				1 + \tanh(c + u ^ {3})
			}{
				2
			},
		}
	\end{equation*}
	where $c$ is an intercept term, $\psi(c,u)=0.75\Phi\{\rbracket{c + u} + 0.5\} +
	0.25\Phi\{0.5\rbracket{c + u}^{3}\}$ and $\Phi$ is the cdf of the standard normal distribution.
	The first two links correspond to Scenarios 1 and 2 of the simulation study investigated by \cite{AKL2019}, while $\varphi_{C}$ is a scaled hyperbolic tangent function. 
	These three link functions are considered for exploring the influence of the shape and the steepness of the true link function to the model estimation performance. \fref{fig:links} shows the plots of the three link functions over $[-4, 4]$ when the intercept term $c = 0$. Among these three link functions, $\varphi_{C}$ is the steepest and $\varphi_{B}$ is the flattest.
	\begin{figure}
		\centering
		\includegraphics[scale=0.5]{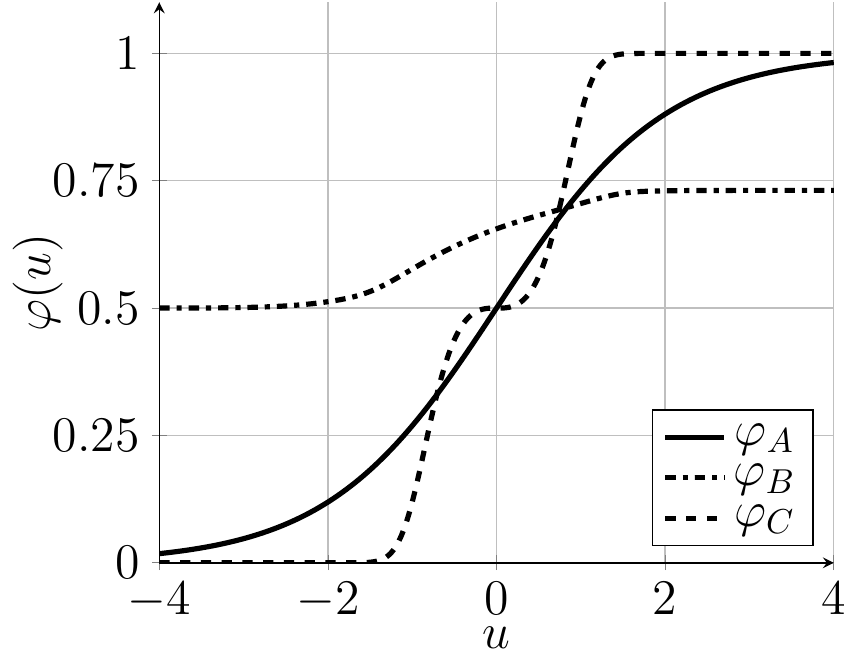}
		\caption{{Link functions when the intercept term $c = 0$.}\label{fig:links}}
	\end{figure}

	\begin{table}
		\centering
		\caption[Simulation settings]{Simulation settings \label{tab:sim_settings}}
		{
			\renewcommand{\arraystretch}{1}
			\small
			\setlength{\tabcolsep}{2pt}

\begin{tabular}{cc@{\extracolsep{3pt}}ccccc@{\extracolsep{3pt}}cc@{\extracolsep{3pt}}cccc}
	\hline
	Expt.                  & $c$ & ${\gamma}_{01}$     & ${\gamma}_{02}$    & ${\gamma}_{03}$     & ${\gamma}_{04}$    & ${\beta}_{01}$   & ${\beta}_{02}$  & $\lambda_{C}$ & \shortstack{Cure\\prop.}         & \shortstack{Cens.\\rate} & Plateau \\ \hline
	\multirow{2}{*}{A} & \multirow{2}{*}{1.2} & \multirow{2}{*}{-0.2383} & \multirow{2}{*}{0.7423} & \multirow{2}{*}{0.3156}  & \multirow{2}{*}{0.5409} & \multirow{2}{*}{-0.8} & \multirow{2}{*}{0.5} & 0.1           & \multirow{2}{*}{0.2090} & 0.2674         & 0.1675  \\
	&                      &                          &                         &                          &                         &                       &                      & 0.3           &                         & 0.3691         & 0.1108  \\ \hline
	\multirow{2}{*}{B} & \multirow{2}{*}{0.5} & \multirow{2}{*}{-0.7826} & \multirow{2}{*}{0.4368} & \multirow{2}{*}{-0.2599} & \multirow{2}{*}{0.3594} & \multirow{2}{*}{-0.6} & \multirow{2}{*}{0.8} & 0.1           & \multirow{2}{*}{0.3352} & 0.3792         & 0.2734  \\
	&                      &                          &                         &                          &                         &                       &                      & 0.4           &                         & 0.4902         & 0.1554  \\ \hline
	\multirow{2}{*}{C} & \multirow{2}{*}{0.2} & \multirow{2}{*}{0.1057}  & \multirow{2}{*}{0.7899} & \multirow{2}{*}{-0.4883} & \multirow{2}{*}{0.3556} & \multirow{2}{*}{0.6}  & \multirow{2}{*}{0.4} & 0.15          & \multirow{2}{*}{0.3912} & 0.4415         & 0.3092  \\
	&                      &                          &                         &                          &                         &                       &                      & 0.5           &                         & 0.5390         & 0.1867  \\ \hline
\end{tabular}

		}
	\end{table}
	
	The choices for the parameters  $\mat{\gamma}_{0}$, $\mat{\beta}_{0}$, and $\lambda_{C}$ are given in \tref{tab:sim_settings} as well as the averages, over the 500 simulated datasets for each setting, of the cure proportion, the censoring rate and the proportion of observations in the plateau. $\mat{\gamma}_{0}$ is chosen such that it has unit Euclidean norm to ensure model identifiability and it leads to different cure proportions. For each experiment we also consider two different censoring scenarios while maintaining a reasonable percentage of observations in the plateau.
	
	The simulation experiments are carried out with sample size $n$ of $250$ and $500$. The SIC method proposed by \cite{AKL2019} and our estimation method (mSIC) are applied to each dataset. For comparison we compute the mean squared error (MSE) of the estimate of the cure probability, bias and variance of the estimates of both $\mat{\gamma}$ and $\mat{\beta}$. The MSE for the cure probability  is defined as 
	$$	\text{MSE}(\hat{p}, p_{0})= \frac{1}{K}\sum_{k=1}^{K}\{
	\hat{\varphi}(\tr{\hat{\mat{\gamma}}}\mat{x}_{k}) - {\varphi_{0}}(\tr{\mat{\gamma}}_{0}\mat{x}_{k})\}^{2},
	$$
	where the summation is over a grid of points. 
	For $x_{k1}$ and $x_{k2}$ we take a grid of size $0.01$ on $\sbracket{0, 1}$ and $\sbracket{-3, 3}$ respectively, while $x_{k3},x_{k4}\in\cbracket{0, 1}$. The bias of the coefficient estimates is the mean of the Euclidean norms of the differences between the coefficient estimate and the true parameter over 500 replications. The variance of the coefficient estimates is the sample variance of the Euclidean norms of the coefficient estimates.%
	
	For both methods, we initialize the algorithms  as follows. The initial link estimate is the logistic function,  $\hat{\mat{\gamma}}_{(0)}$ is the estimate from fitting a logistic regression model to the censoring indicator against the covariates $\mat{X}$,  initial estimates for the latency are obtained from fitting the standard Cox model to the uncensored observations. Both algorithms terminate when the difference of the estimators from one iteration to the other, in Eucledian norm, is smaller then $10^{-5}$.
	Other configurations for the SIC method are set as stated in Section 3 of \cite{AKL2019}. Since the SIC method sets $\hat{\gamma}_{1} = \pm1$ for the model identifiability condition, at the end we normalize the obtained estimate. For our  method we used the triweight kernel function for smoothing the monotone link estimate and set the bandwidth parameter at the $k$-th iteration to be $h_k = r_{k}n^{-1/5}$, where $r_{k}$ is the range of the index $\mat{\gamma}^{T}\mat{X}$ computed at the $k$-th iteration. This is a common choice for smooth isotonic estimators and behaves well in practice. We note that the bandwidth choice is less problematic in the case of smooth monotone estimators than for the standard kernel estimator. A more detailed investigation of the role of the bandwidth is provided in Appendix~\ref{sec:bandwidth}. For the SIC method, cross-validation is used to select the bandwidth but in our experience that is not stable in practice and the search interval needs to be chosen depending on the range of the index (see discussion in Appendix~\ref{sec:app_application}). Regarding the truncation parameter, we set $\epsilon'= 10^{-6}$ for both experiments A and C. For experiment B, given the small range of values for the incidence, the estimation is very challenging. Since the isotonic estimator is known to have problems at the boundary, in these type of situations it is better to determine the upper and lower truncation in a data-driven way instead of taking a very small $\epsilon'$. We  use the method proposed by \cite{LR2017} for the range-regularized isotonic regression problem. Details are given in Appendix~\ref{sec:appendix_setting_B}.
	\begin{table}
		{
			\centering
			\caption[Simulation results]{Simulation results\label{tab:sim_results}} 
			{
				\renewcommand{\arraystretch}{1}
				\small
				\setlength{\tabcolsep}{2pt}
				\begin{tabular}{cccl@{\extracolsep{6pt}}cc@{\extracolsep{6pt}}cc@{\extracolsep{6pt}}cc}
	\hline
	\multirow{2}{*}{Expt.} & \multirow{2}{*}{Size} & \multirow{2}{*}{$\lambda_{C}$} & \multirow{2}{*}{Method} & \multicolumn{2}{c}{$\text{MSE}\rbracket{\hat{p},p_{0}}$} & \multicolumn{2}{c}{$\hat{\mat{\gamma}}$} & \multicolumn{2}{c}{$\hat{\mat{\beta}}$} \\ \cline{5-6} \cline{7-8} \cline{9-10}
	&                              &                                &                         & Mean                       & Variance                    & Bias               & Variance            & Bias               & Variance           \\ \hline
	\multirow{8}{*}{A}          & \multirow{4}{*}{250}         & \multirow{2}{*}{0.1}           & mSIC                    & 0.00939                    & 4.39E-05                    & 0.57368            & 0.05038             & 0.28080            & 0.02709            \\
	&                              &                                & SIC                     & 0.01688                    & 1.41E-04                    & 0.71096            & 0.05077             & 0.28151            & 0.02708            \\ \cline{3-10} 
	&                              & \multirow{2}{*}{0.3}           & mSIC                    & 0.01141                    & 5.92E-05                    & 0.61280            & 0.05145             & 0.30371            & 0.03408            \\
	&                              &                                & SIC                     & 0.01940                    & 1.55E-04                    & 0.76052            & 0.06128             & 0.30605            & 0.03442            \\ \cline{2-10} 
	& \multirow{4}{*}{500}         & \multirow{2}{*}{0.1}           & mSIC                    & 0.00562                    & 1.50E-05                    & 0.43937            & 0.03810             & 0.19793            & 0.01321            \\
	&                              &                                & SIC                     & 0.00967                    & 4.57E-05                    & 0.58563            & 0.04825             & 0.19799            & 0.01331            \\ \cline{3-10} 
	&                              & \multirow{2}{*}{0.3}           & mSIC                    & 0.00668                    & 2.01E-05                    & 0.48043            & 0.03989             & 0.22269            & 0.01667            \\
	&                              &                                & SIC                     & 0.01183                    & 7.44E-05                    & 0.61136            & 0.04907             & 0.22371            & 0.01716            \\ \hline
	\multirow{8}{*}{B}          & \multirow{4}{*}{250}         & \multirow{2}{*}{0.1}           & mSIC                   & 0.00697                    & 6.06E-05                    & 0.99875            & 0.20072             & 0.30717            & 0.03557            \\
	&                              &                                & SIC                     & 0.00831                    & 9.70E-05                    & 1.02873            & 0.26451             & 0.30798            & 0.03525            \\ \cline{3-10} 
	&                              & \multirow{2}{*}{0.4}           & mSIC                    & 0.00912                    & 9.79E-05                    & 0.94890            & 0.18028             & 0.35858            & 0.04353            \\
	&                              &                                & SIC                     & 0.01041                    & 1.34E-04                    & 1.04764            & 0.25388             & 0.35971            & 0.04412            \\ \cline{2-10} 
	& \multirow{4}{*}{500}         & \multirow{2}{*}{0.1}           & mSIC                    & 0.00433                    & 2.25E-05                    & 0.88094            & 0.17064             & 0.21738            & 0.01595            \\
	&                              &                                & SIC                     & 0.00504                    & 4.17E-05                    & 0.91734            & 0.26178             & 0.21715            & 0.01596            \\ \cline{3-10} 
	&                              & \multirow{2}{*}{0.4}           & mSIC                    & 0.00591                    & 3.50E-05                    & 0.86886            & 0.17217             & 0.24466            & 0.02036            \\
	&                              &                                & SIC                     & 0.00578                    & 4.44E-05                    & 0.97191            & 0.25303             & 0.24623            & 0.02122            \\ \hline
	\multirow{8}{*}{C}          & \multirow{4}{*}{250}         & \multirow{2}{*}{0.15}          & mSIC                    & 0.00583                    & 3.49E-06                    & 0.29442            & 0.01686             & 0.32893            & 0.03600            \\
	&                              &                                & SIC                     & 0.01474                    & 1.79E-04                    & 0.40554            & 0.01987             & 0.32966            & 0.03599            \\ \cline{3-10} 
	&                              & \multirow{2}{*}{0.5}           & mSIC                    & 0.00666                    & 6.10E-06                    & 0.33582            & 0.02433             & 0.37579            & 0.05071            \\
	&                              &                                & SIC                     & 0.01847                    & 2.82E-04                    & 0.45743            & 0.02754             & 0.38216            & 0.05140            \\ \cline{2-10} 
	& \multirow{4}{*}{500}         & \multirow{2}{*}{0.15}          & mSIC                    & 0.00446                    & 4.52E-05                    & 0.20682            & 0.01075             & 0.22675            & 0.01747            \\
	&                              &                                & SIC                     & 0.00811                    & 5.61E-05                    & 0.32336            & 0.01965             & 0.22658            & 0.01761            \\ \cline{3-10} 
	&                              & \multirow{2}{*}{0.5}           & mSIC                    & 0.00459                    & 1.41E-06                    & 0.23880            & 0.01169             & 0.26455            & 0.02287            \\
	&                              &                                & SIC                     & 0.01038                    & 8.51E-05                    & 0.35384            & 0.02201             & 0.26582            & 0.02345            \\ \hline
\end{tabular}
			}\par
		}
	\end{table}
	
	\tref{tab:sim_results} summarizes the simulation results, including the MSE of the link estimates, bias and variance of the coefficient estimates, for both SIC and mSIC methods. 
	In terms of the MSE for  the cure probability, the mSIC method has lower mean and variance compared to the SIC method among all simulation settings, except for Experiment B with sample size of $500$ and $\lambda_{C} = 0.4$. This indicates that mSIC performs better in inferring the incidence and gives less dispersed estimates. The mSIC method behaves better in estimating $\mat{\gamma}$ and $\mat{\beta}$, in terms of bias and variance among all simulation configurations. As expected, as the sample size increases the performance of both methods improves, while it deteriorates when the censoring rate increases. 
	
	For experiment B, we find that range-regularization method shrinks the range of the upper and lower truncation to zero for about one-fifth of the replications for each simulation settings. This leads to a constant estimate of the link function which poses a practical identification issue for $\mat{\gamma}$ and causes the high variability of $\hat{\mat{\gamma}}$. We look separately at the cases with a constant (non-constant) estimate for the link function and report the correspondent MSE, bias and variance in Tables~\ref{tab:sim_results_B_const_link}-\ref{tab:sim_results_B_nonconst_link}
	 in Appendix~\ref{sec:appendix_setting_B}.
	Results indicate that, when our method estimates a constant link function,  the SIC method behaves no better than a constant link estimate. Table~\ref{tab:sim_results_B_nonconst_link} 
	also suggests mSIC performs better in estimating the incidence for almost all cases with a non-constant link estimate. 
	As expected, the variance of $\hat{\mat{\gamma}}$ for the cases with a constant estimated link function is higher than that for a non-constant link. This practical identifiability issue may be a result of the flatness of the true link function $\varphi_{B}$, although it is not a constant function. Note that the true incidence probability  ranges from around $0.52$ to around $0.73$. We recommend that the range-regularized method should be used in practice to determine the upper and lower truncation in order to identify situations for which there is risk of practical identifiability as in experiment B. In such cases, one should be careful in interpreting the results.
	
	We note that if the true link function is not monotone, using mSIC would lead to model mispecification and SIC would be preferred over mSIC. We illustrate this through an additional simulation setting in Appendix~\ref{sec:non_monotone}. 
 However, we expect that for small deviations from monotonicity, mSIC would still provide reasonable estimates. Finally, we also investigated the role of the additional smoothing step 2 in the estimation procedure described in Section~\ref{sec:model_est}. Results in Appendix~\ref{sec:smoothing} illustrate that the smooth monotone estimator behaves better than the piecewise constant isotonic estimator.
	
	For the computational aspect, the proposed method consumes slightly more time to complete the model estimation comparing to the SIC method. For instance, the average elapsed time, over the 500 replications, on the model estimation for experiment C with $\lambda_{C} = 0.15$ and a sample size of $250$ ($500$) are 12.28 (14.20) seconds for mSIC and 11.64 (13.67) seconds for SIC, with a Core i7-1165G7 CPU laptop.

	\section{Real data application\label{sec:application}}
	We apply the proposed method to a dataset of melanoma patients extracted from the Surveillance, Epidemiology and End Results (SEER) database. Such dataset has also been studied by \cite{musta2022presmoothing} and it consists of 1445 melanoma patients diagnosed between 2004 -- 2015. The event time of interest is the time to death because of melanoma and the follow-up time ranges from 1 to 155 months. The age of the patients varies between 11 to 104 years old. Among the 1445 patients, 596 are females and 849 are males. The cancer stage at diagnosis (localized: 1302 cases, regional: 101 cases, distant: 42 cases) is also recorded. The Kaplan-Meier estimate of the survival function in  \fref{fig:km_seer} has a long plateau, which contains around 20\% of the observations. Combined with medical evidence of possibility of cure for melanoma, this suggests that the cure model is appropriate for this dataset.
	
	\begin{figure}
		\centering
		\includegraphics[scale=0.4]{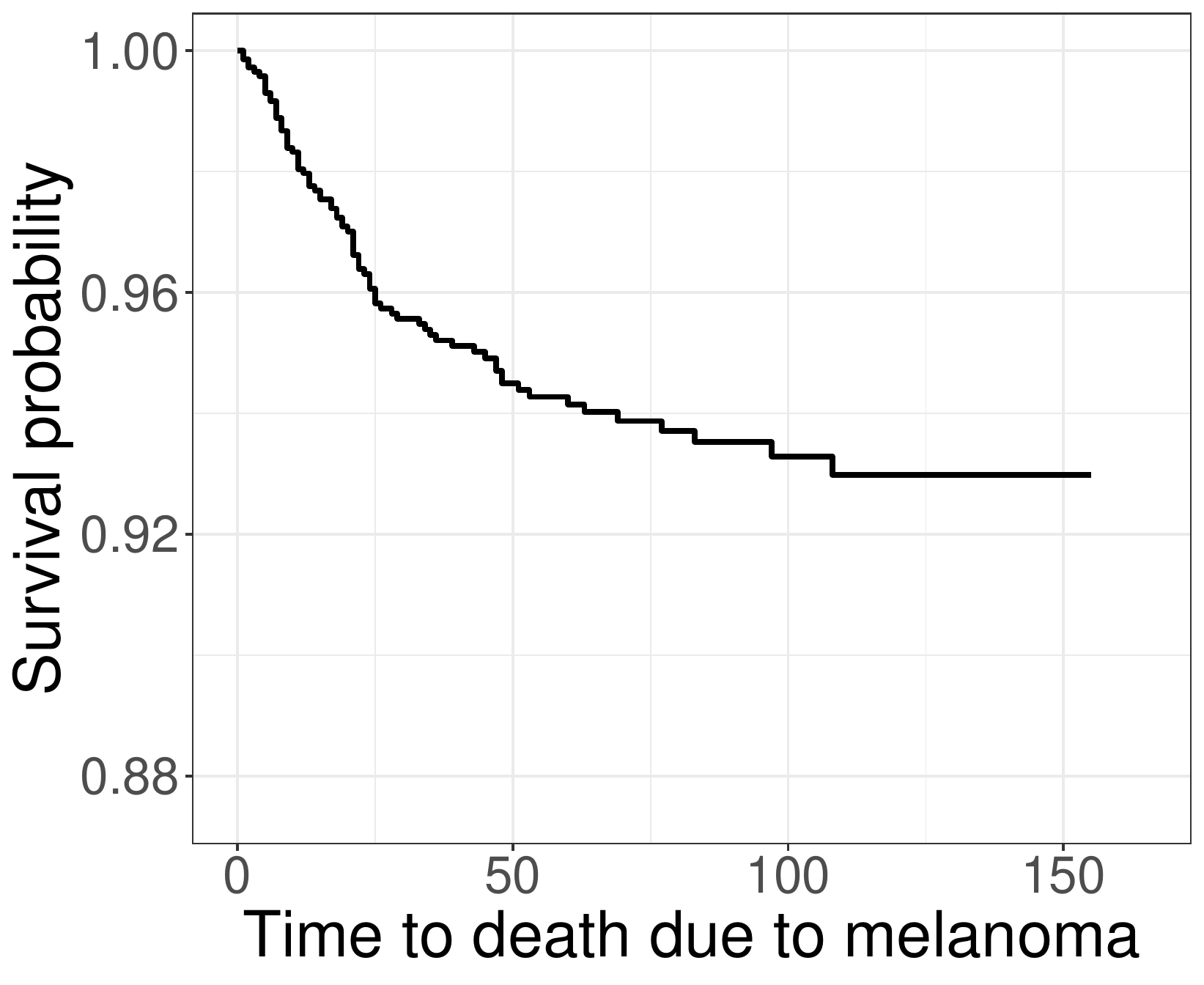}
		\caption{Kaplan-Meier estimate of the survival function for the SEER data.\label{fig:km_seer}}
	\end{figure}
	
	To compare the estimates of the logistic/Cox (LC), the single-index/Cox (SIC) and the monotone single-index/Cox (mSIC) models, we split data into a training set (with size of 964) and a testing set (with size of 481) and estimate the models using the training set. The standardized age is considered in the incidence in order to select a reasonable bandwidth for the SIC model. We note that the SIC method selected an optimal bandwidth of $1$, which is the upper bound of the search interval, $[0.4, 1]$, considered in the SIC algorithm. However, extending further the search interval for the bandwidth gave similar results.  \tref{tab:seer_results} shows the parameter estimates, and the percentile-based bootstrap confidence intervals with approximate 95\% confidence level. The confidence intervals are computed using 500 naive bootstrap samples. The effects of the covariates have the same direction for all methods. 
	Note that the $\mat\gamma$ estimate of the LC model are normalized, so that the coefficients corresponding to the non-intercept terms possess a norm of one for better comparison with the other two methods.
	To compare the performance of the three models in predicting the uncure probability, we compute the prediction error (PE) using the testing data
	$$PE = -\sum_{i=1}^{481}\hat{w}_{i}\log\rbracket{1 - \hat{p}(\mat{x}_{i}^{\text{test}})} + (1-\hat{w}_{i}) \log\hat{p}(\mat{x}_{i}^{\text{test}}),$$
	where $\hat{w}_{i}$ is computed using \eqref{eqn:w_i_EM} with the estimated parameters (and the estimated link for SIC and mSIC.) The prediction errors are 95.47 for the LC model, 101.14 for the SIC model and 81.46 for the mSIC model. This indicates that mSIC performs better in predicting the uncure status among the three approaches. As mentioned previously, the SIC algorithm selected an optimal bandwidth of 1. The PE for SIC reduced slightly (to 97.15) when we allowed the SIC algorithm to search for some larger bandwidth. 
	\begin{table}[h]
		{
			\centering
			\caption[SEER Results]{Parameter estimates for the SEER data {(LC's $\hat{\mat{\gamma}}_{n}$ is normalized)}\label{tab:seer_results}}
			{
				\renewcommand{\arraystretch}{1}
				\footnotesize
				\setlength{\tabcolsep}{1pt}
				\begin{tabular}{ll@{\extracolsep{5pt}}ccc@{\extracolsep{5pt}}ccc@{\extracolsep{5pt}}ccc}
	\hline
	\multirow{2}{*}{}          & \multirow{2}{*}{Covariates} & \multicolumn{3}{c}{LC}                                    & \multicolumn{3}{c}{SIC}                                   & \multicolumn{3}{c}{mSIC}                                  \\ \cline{3-5}\cline{6-8}\cline{9-11} 
	&                             & Est.    & \shortstack{Lower\\CI} & \shortstack{Upper\\CI} & Est.    & \shortstack{Lower\\CI} & \shortstack{Upper\\CI} & Est.    & \shortstack{Lower\\CI} & \shortstack{Upper\\CI} \\ \hline
	\multirow{5}{*}{{$\hat{\mat{\gamma}}_{n}$}} & Intercept                   & -4.2129 & -4.8391                & -3.3821                & -       & -                      & 0.3640                 & -       & -                      & -                      \\
	& Age                         & 0.0620  & -0.0189                & 0.1463                 & 0.1221  & 0.0586                 & 0.3640                 & 0.0487  & -0.0633                & 0.0843                 \\
	& Gender                      & 0.1298  & -0.0475                & 0.3019                 & 0.2167  & -0.0132                & 0.7204                 & 0.1713  & 0.0246                 & 0.3119                 \\
	& Regional                    & 0.5842  & 0.4812                 & 0.7208                 & 0.5217  & 0.1992                 & 0.6674                 & 0.5658  & 0.4843                 & 0.7121                 \\
	& Distant                     & 0.7987  & 0.7266                 & 0.8935                 & 0.8161  & 0.7393                 & 1.1943                 & 0.8051  & 0.7270                 & 0.8794                 \\ \hline
	\multirow{4}{*}{{$\hat{\mat{\beta}}_{n}$}}   & Age                         & -0.0079 & -0.0277                & 0.0163                 & -0.0086 & -0.0086                & -0.0283                & -0.0038 & -0.0196                & 0.0188                 \\
	& Gender                      & -0.1809 & -1.0741                & 0.9811                 & -0.201  & -0.2010                & -1.0711                & -0.1609 & -0.8706                & 0.7711                 \\
	& Regional                    & 0.3908  & -0.5542                & 1.4638                 & 0.479   & 0.4790                 & -0.3423                & 0.2916  & -0.5600                & 1.0817                 \\
	& Distant                     & 1.4290  & 0.1771                 & 2.5478                 & 1.5064  & 1.5064                 & 0.2783                 & 1.1356  & -0.0604                & 2.0322                 \\ \hline
\end{tabular}

			}\par
		}
	\end{table}
	\begin{figure}[h]
		\centering
		\includegraphics[scale=0.35]{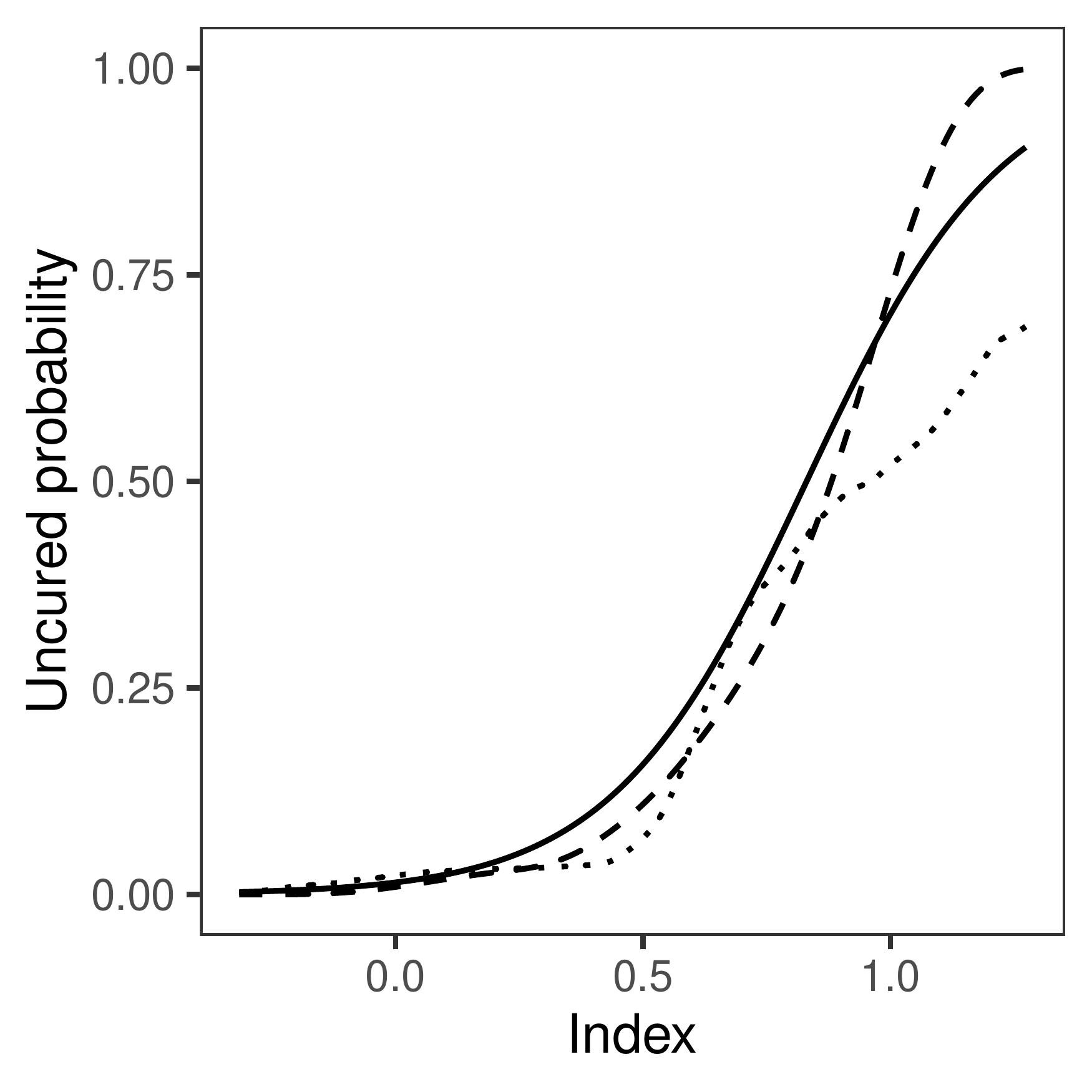}
		\caption{Link estimates (solid: LC, dotted: SIC, dashed: mSIC)\label{fig:seer_links}}
	\end{figure}

	Figure \ref{fig:seer_links} shows the link estimates of the three methods, plotted over the {same} range of the index. The LC estimate is rescaled according to the normalization of $\hat\gamma_n$. For this particular split of the data, all three methods give a monotone link function, and the links of LC and mSIC have similar  shape. The estimated link function from SIC differs considerably from the other two methods but it has the largest prediction error. However, as illustrated in Appendix~\ref{sec:app_application}, the SIC estimate is in general non-monotone and unstable. Part of this behavior seems to be due to its sensitivity to the bandwidth choice and the behavior of the cross-validation bandwidth selection method. mSIC on the other hand, apart from guaranteeing a monotone link estimate, is also more stable, less sensitive to the bandwidth and behaves well even for a simple (non-optimal) bandwidth choice. To achieve a more robust comparison among the three models using the prediction error, we also generate 10 random splits of the data (training: 964, testing: 481) and compute the prediction error (PE) using the testing data for each split. The averages (standard deviations) of the prediction errors over the 10 random splits are 89.64 (5.73) for LC; 91.84 (9.51) for SIC; and 71.84 (6.66) for mSIC. This suggests the same conclusion as before that mSIC behaves better in predicting the uncure probability among the three methods. 
	
	\appendix
	\section{Technical lemmas and proofs}
	\label{sec:proofs}
	
	\begin{proof}[Proof of Proposition \ref{prop:monotone_link_est}]
		We start by making an analogy with the standard isotonic regression problem, which corresponds to solving the following minimization problem
		\[
		\hat{\mat{y}}=\argmin_{\mat{y}\in\family{C}}\sum_{i=1}^{n}\rbracket{x_{i} - y_{i}}^{2},
		\]
		where $\family{C} = \set{\mat{y} \in \real[n]}{y_{1} \leq y_{2} \leq \cdots \leq y_{n}}$, for given observations $x_1,\dots,x_n$. From the theory of isotonic regression \citep[Section 1.5]{RWD1988}, it is known that the isotonic least squares estimate $\hat{\mat{y}}$ is also solution of the following optimization problem 
		\begin{equation}
			\label{eq:isoreg}
			\hat{\mat{y}} = \argmin_{\mat{y} \in \family{C}} \sum_{i=1}^{n}\function{\Delta}{\Psi}{}{x_{i},y_{i}},
		\end{equation}
		where 
		\begin{equation}
			\label{eqn:Delta_fun}
			\function{\Delta}{\Psi}{}{x,y} = \function{\Psi}{}{}{x} - \function{\Psi}{}{}{y} - \rbracket{x - y}\function{\psi}{}{}{y},
		\end{equation}
		and $\Psi$ is a convex function with derivative $\psi$. This equivalence is usually used to transform maximum likelihood isotonic estimation into an isotonic least-squares problem. Moreover, $\hat{\mat{y}}$ can be characterized as the left derivative of the greatest convex minorant (GCM) of the cumulative sum diagram (CSD) $$\cbracket{\rbracket{0,0}, \rbracket{i, \sum_{j=1}^{i}x_{(j)}}, i=1,\cdots,n},$$  where $x_{(1)} \leq x_{(2)} \leq \cdots \leq x_{(n)}$ correspond to the ordered observations. If in addition to the order restriction, a uniform bound restriction is imposed on the solution, the problem becomes
		\[
		\hat{\mat{y}}=\argmin_{\mat{y}\in\family{C'}}\sum_{i=1}^{n}\rbracket{x_{i} - y_{i}}^{2},
		\]
		with $\family{C}^{\prime} = \set{\mat{y} \in \real[n]}{a \leq y_{1} \leq y_{2} \leq \cdots \leq y_{n} \leq b}$ and is studied by \cite{H1997}. Specifically, it has shown that the minimizer under the order and uniform bound restrictions is $\tilde{\mat{y}} = \rbracket{\tilde{y}_{1}, \cdots, \tilde{y}_{n}}$ given by 
		\[
		\tilde{y}_{i} = \max\rbracket{a, \min\rbracket{\hat{y}_{i}, b}}, \qquad i=1,\cdots,n,
		\]
		where $\hat{\mat{y}} = \rbracket{\hat{y}_{1}, \cdots, \hat{y}_{n}}$ is the solution of \eqref{eq:isoreg}.

		Going back to our maximization problem in \eqref{eq:link_mle},
		first note that since the criterion depends only on the values of $\varphi$ at the observed points $\mat\gamma^T\mat{x}_{(1)}<\dots<\mat\gamma^T\mat{x}_{(n)}$ we can only identify the correspondent values $\mat\phi=(\phi_1,\dots,\phi_n)$. If we choose the convex function $\function{\Psi}{}{}{u} = u \log u + \rbracket{1-u}\log\rbracket{1-u}$ in \eqref{eqn:Delta_fun}, we obtain
		\begin{equation}
			\label{eq:link_mle_isoreg}
			\begin{split}
&\sum_{i=1}^{n}\function{\Delta}{\Psi}{}{w_{i}, \link[\tr{\mat{\gamma}}\mat{x}_{i}]{}{}}\\
& = \sum_{i=1}^{n}\cbracket{- w_{i} \log\link[\tr{\mat{\gamma}}\mat{x}_{i}]{}{} - \rbracket{1 - w_{i}}\log\sbracket{1 - \link[\tr{\mat{\gamma}}\mat{x}_{i}]{}{}}}+C,
			\end{split}
					\end{equation}
		for some $C$ that depends on the observations and the fixed parameters but not on $\varphi$. Therefore, the maximization problem in \eqref{eq:link_mle} 
		is equivalent to solving
		\[
		\widehat{\mat\phi}=\argmin_{\mat\phi\in\family{C}_{\epsilon^{\prime}}}\sum_{i=1}^{n}\function{\Delta}{\Psi}{}{w_{(i)},\phi_{i}{}{}}
		\]
		where  $\family{C}_{\epsilon^{\prime}} = \set{\mat{\phi} \in \real[n]}{\epsilon^{\prime} \leq \phi_{1} \leq \phi_{2} \leq \cdots \leq \phi_{n} \leq 1 - \epsilon^{\prime}}$. Consequently, using the previous results for the isotonic regression problem,  $\widehat{\mat\phi}$ can be characterized as $$\hat\phi_{i} = \max\rbracket{\epsilon^{\prime}, \min\rbracket{\tilde{\phi}_{i}, 1-\epsilon^{\prime}}},\qquad i = 1, \cdots, n,$$ where $\rbracket{\tilde{\phi}_{1}, \cdots, \tilde{\phi}_{n}}$ are the left derivatives of the GCM of the CSD $$\cbracket{\rbracket{0,0}, \rbracket{i, \sum_{j=1}^{i}w_{(j)}}, i=1,\cdots,n}$$
		with $w_{(j)}$ corresponding to the $j$-th order statistic of $(\tr{\mat{\gamma}}\mat{x}_{i})_i$. Note also that the procedure can easily accomodate ties in the observations of the index by summing the correspondent $w_{j}$. 
		
		Finally we note that, if one would try to directly characterize the maximizer of the likelihood criterion in \eqref{eq:monotone_link_mle} 
		over the non-decreasing function $\varphi$ would still end up with an iterative procedure that is the same as the EM algorithm. Indeed, the problem is equivalent to maximizing
		\[
		\mathcal{L}_n(\mat\phi)=\sum_{i=1}^n \Delta_{(i)}\log \phi_i+(1-\Delta_{(i)})\log\left\{1-\phi_i+\phi_iS_u(Y_{(i)}|Z_{(i)})\right\}
		\]
		over $\phi_1\leq\dots\leq \phi_n$ denoting again the values of the link function on the ordered points $\mat\gamma^T\mat{x}_{(1)}<\dots<\mat\gamma^T\mat{x}_{(n)}$. If we assume to know 
		\[
		W_{(i)}=\Delta_{(i)}+(1-\Delta_{(i)})\frac{\phi_iS_u(Y_{(i)}|Z_{(i)})}{1-\phi_i+\phi_iS_u(Y_{(i)}|Z_{(i)})}
		\] 
		and choose $\function{\Psi}{}{}{u} = u \log u + \rbracket{1-u}\log\rbracket{1-u}$ as before, we obtain 
		\[
		\sum_{i=1}^{n}\function{\Delta}{\Psi}{}{W_{(i)}, \phi_i} =-\mathcal{L}_n(\mat\phi)-\sum_{i=1}^n(1-\Delta_{(i)})W_{(i)}\log S_u(Y_{(i)}|Z_{(i)}).
		\]
		As a result, maximizing $\mathcal{L}_n(\mat\phi)$ is equivalent to minimizing the expression on the left hand side of the equation. The solution can be characterized as the slope of the GCM of the cumulative sum diagram  $\{(0,0),(i,\sum_{i=1}^nW_{(i)}), i=1,\cdots,n\}$. However, the $W_{(i)}$ are actually not known and depend on $\mat\phi$. One could construct an iterative algorithm in which $\mat\phi$ from the previous step is used to compute $W_{(i)}$ and then update $\mat\phi$ again. This is the same as what EM algorithm does since the $W_{(i)}$ coincide with the ones defined in \eqref{eqn:w_i_EM}.
	\end{proof} 
	\begin{proof}[Deriavation of the characterization in \eqref{eqn:characterization_phi_theta}] 
		By definition we have 
		\[
		\begin{split}
	\varphi_{0,\mat\theta}&=\argmax_{\varphi\in\family{M_{\epsilon'}}}\mathbb{E}\left[\Delta\log\varphi(\mat\gamma^TX)\right.\\
		&\left.\qquad\qquad\qquad+(1-\Delta)\log\left\{1-\varphi(\mat\gamma^TX)+\varphi(\mat\gamma^TX)S_u(Y|Z;\mat\beta,\Lambda)\right\}\right].
		\end{split}
			\]
		Let us assume for this that the covariates $X$ and $Z$ are independent and the censoring is independent of the other variables. Define $U=\mat\gamma^TX$. Then the true likelihood for the observations $(Y,\Delta,U,Z)$ is 
		\[
		\begin{split}
			&\Delta\prob{C>Y,B=1,T=Y\mid U,Z}+(1-\Delta)\prob{C=Y,T>Y\mid U,Z}\\
			&=\Delta\prob{C>Y}\prob{B=1\mid U}\prob{T=Y\mid B=1,Z}\\
			&\quad+(1-\Delta)\prob{C=Y}\left\{\prob{B=0\mid U}+\prob{B=1\mid U}\prob{T>Y\mid B=1,Z}\right\}
		\end{split}
		\]
		where with a slight abuse of notation we have used $\prob{\cdot}
		$ instead of the density function for the continuous variables. By the law of conditional expectations we have 
		\[
		\prob{B=1\mid U}=\expect[\sbracket]{\expect[\sbracket]{\indicator{B=1}\mid X}\mid U}=\expect[\sbracket]{\varphi_0(\mat\gamma_0^TX)\mid U}=:\varphi^*_{\mat\gamma}(U)
		\]	
		Then the true log-likelihood for the observations $(Y,\Delta,U,Z)$ can be written as
		\[
		\begin{split}
			&=\Delta\log \varphi^*_{\mat\gamma}(U)+\Delta\log f_u(Y|Z;\mat\beta_0,\Lambda_0)\\
			&\quad+(1-\Delta)\log\left\{1-\varphi^*_{\mat\gamma}(U)+\varphi^*_{\mat\gamma}(U)S_u(Y|Z;\mat\beta_0,\Lambda_0)\right\}+c
		\end{split}
		\]
		where $c$ denotes other terms that do not depend on the model parameters (related to the censoring distribution). If we fix $\mat\theta=(\mat\gamma,\mat\beta_0,\Lambda_0)$, it follows from the Kullback-Leibler inequality that 
		\[
		\argmax_{\varphi\in\family{M}}\expect[\sbracket]{\Delta\log\varphi(U)+(1-\Delta)\log\left\{1-\varphi(U)+\varphi(U)S_u(Y|Z;\mat\beta_0,\Lambda_0)\right\}}
		\]
		is obtained for $\varphi=\varphi^*_{\mat\gamma}$. On the other hand, Assumptions (A4), (A1)(i-ii), (A3)(ii) imply that there exists $\epsilon>0$ such that $\epsilon\leq\varphi_0(\mat\gamma_0^TX)\leq 1-\epsilon$. As a result, if we choose $\epsilon'<\epsilon$,  $\varphi^*_{\mat\gamma}\in\family{M_{\epsilon'}}$. Hence we conclude that under these more restrictive assumptions, $\varphi_{0,\mat\theta}$ is given by the expression in \eqref{eqn:characterization_phi_theta}. 
		
		Next we comment on the necessity of such assumptions. To apply the previous argument we needed that $C\perp (B,T)\mid U,Z$. Hence just conditional independence $C\perp (B,T)\mid X,Z$ would not be sufficient to split the probability into a product of probabilities. In addition, if $X$ and $Z$ would be dependent then we also would not have  $\prob{B=1\mid U,Z}=\prob{B=1\mid U}$. For example, if $X=Z$, we would have $\prob{B=1\mid U,Z}=\varphi_0(\mat\gamma_0^TZ)$ and there would be no connection between the function $\varphi^*_{\mat\gamma}$ and the true likelihood. Finally, if the latency parameters were not fixed to their true values $\mat\beta_0,\Lambda_0$, then the Kullback-leibler inequality could not be used to conclude that the maximizer is $\varphi^*_{\mat\gamma}$.
		
	\end{proof}
	\begin{proof}[Proof of Proposition \ref{prop:link_maximizer}]
		Let $Q_{\mat{X}}$ denote the distribution of $\mat{X}$. 	For any $\varphi\in\family{M}_{\epsilon'}$ and $\mat\gamma\in\Gamma$ define the $L_2$-norm
		\[
		\norm{\varphi}{{Q}_{\mat{X}},\mat\gamma} = \rbracket{\int_{\family{X}}\varphi(\mat\gamma^T\mat{x})^{2}\dd{Q_{\mat{X}}}\rbracket{\mat{x}}}^{\frac{1}{2}}
		\]
		To show the existence of a maximizer of $l_{\mat{\theta}}\rbracket{\varphi}$ over $\family{M}_{\epsilon^\prime}$, it suffices to show that $l_{\mat{\theta}}$ is continuous on $(\family{M}_{\epsilon^\prime}, \norm{\cdot}{{Q}_{\mat{X}},\mat\gamma})$ and $(\family{M}_{\epsilon^\prime}, \norm{\cdot}{{Q}_{\mat{X}},\mat\gamma})$ is compact. To obtain  uniqueness of the maximizer, we show that $\mapping{\varphi}{l_{\mat{\theta}}\rbracket{\varphi}}$ is strictly concave using the Gateaux derivative and $\family{M}_{\epsilon^{\prime}}$ is a convex set.
		
		For simplicity we consider separately the two terms $l_{1, \mat{\theta}}\rbracket{\varphi} = \expect[\sbracket]{\Delta\log\link[\tr{\mat{\gamma}}\mat{X}]{}{}}$ and $l_{2, \mat{\theta}}(\varphi) = \mathbb{E}[(1 - \Delta)\log\{1 - \varphi(\tr{\mat{\gamma}}\mat{X})F_u(Y|\mat{Z})\}]$. 
		Let $\rbracket{\varphi_{m}}_{m} \subset \family{M}_{\epsilon^\prime}$ be any sequence on $\family{M}_{\epsilon^\prime}$ that converges to $\tilde{\varphi} \in \family{M}_{\epsilon^\prime}$. For any $\eta > 0$, we can find an $N\rbracket{\eta} \in \naturalnumber$ such that $\norm{\varphi_{m} - \tilde{\varphi}}{{Q}_{\mat{X}},\mat\gamma} < \eta$, whenever $m \geq N\rbracket{\eta}$. For $m \geq N\rbracket{\eta}$, we have
		{		\begin{align*}
			\begin{split}
				|l_{1,\mat{\theta}}\rbracket{\varphi_{m}} - l_{1,\mat{\theta}}\rbracket{\tilde{\varphi}}|
				&= \left|\expect[\sbracket]{\Delta\log \varphi_m(\tr{\mat{\gamma}}\mat{X})} - \expect[\sbracket]{\Delta\log\tilde{\varphi}(\tr{\mat{\gamma}}\mat{X})}\right|\\%
				&\leq
	 			\frac{1}{\epsilon^{\prime}}\rbracket{\expect[\sbracket]{\Delta^{2}}\expect[\sbracket]{\cbracket{\link[\tr{\mat{\gamma}}\mat{X}]{m}{} - \tilde{\varphi}\rbracket{\tr{\mat{\gamma}}\mat{X}}}^{2}}}^{\frac{1}{2}}\\%
				&\leq
				\frac{1}{\epsilon^{\prime}}\norm{\varphi_{m} - \tilde{\varphi}}{{Q}_{\mat{X}},\mat\gamma} < \frac{\eta}{\epsilon^{\prime}}.
			\end{split}
		\end{align*}
		Here, the second inequality follows from the Cauchy-Schwarz inequality, the mean value theorem and the bound on $\epsilon'\leq\tilde{\varphi},\varphi_{m}\leq1-\epsilon'$. In a similar way, for $m \geq N\rbracket{\eta}$,
		\begin{align*}
			\begin{split}
				|l_{2,\mat{\theta}}\rbracket{\varphi_{m}} - l_{2,\mat{\theta}}\rbracket{\tilde{\varphi}}|
				&\leq
				\frac{1}{\epsilon^{\prime}}\rbracket{
					\expect[\sbracket]{\cbracket{\rbracket{1 - \Delta}\uncuredsurvival[F]{\mat{Z}}{Y}{}}^{2}}\right.\\%
				&\qquad\qquad\left.
					\times\expect[\sbracket]{\cbracket{\link[\tr{\mat{\gamma}}\mat{X}]{m}{} - \tilde{\varphi}\rbracket{\tr{\mat{\gamma}}\mat{X}}}^{2}}}^{\frac{1}{2}}\\%
				&\leq
				\frac{1}{\epsilon^{\prime}}\norm{\varphi_{m} - \tilde{\varphi}}{{Q}_{\mat{X}},\mat\gamma} < \frac{\eta}{\epsilon^{\prime}}.
			\end{split}
		\end{align*}
		}
		Therefore, take $\xi = {2\eta}/{\epsilon^{\prime}} > 0$ and $N = N\rbracket{{\xi\epsilon^{\prime}}/{2}}$, we have $|l_{\mat{\theta}}\rbracket{\varphi_{m}} - l_{\mat{\theta}}\rbracket{\tilde{\varphi}}| < \xi$, whenever $m \geq N$, and hence $l_{\mat{\theta}}$ is continuous on $(\family{M}_{\epsilon^\prime}, \norm{\cdot}{Q_{\mat{X}},\mat\gamma})$.%
		
		Next, we show the compactness of $(\family{M}_{\epsilon^\prime}, \norm{\cdot}{{Q}_{\mat{X}},\mat\gamma})$. By Helly's selection theorem, since $(\varphi_{m})_{m} \subset \family{M}_{\epsilon^\prime}$ is a sequence of uniformly bounded monotonically increasing functions, there exists a convergent subsequence $(\varphi_{m_{k}})_{k \in \naturalnumber}$ of $(\varphi_{m})_{m \in \naturalnumber}$ such that $\varphi_{m_{k}}\rbracket{u} \rightarrow \varphi^{\ast}\rbracket{u}$ for almost all $u \in \real$ {with respect to the Lebesgue measure}. $\varphi_{m_{k}}$ are all uniformly bounded monotonically increasing functions, so as its limit $\varphi^{\ast}$. By Assumption \ref{enum:index_density_assumption}, the distribution of $\tr{\mat{\gamma}}\mat{X}$ has a density with respect to the Lebesgue measure, we have $\norm{\varphi_{m_{k}} - \varphi^{\ast}}{Q_{\mat{X}},\mat\gamma} \rightarrow 0$ almost everywhere by the dominated convergence theorem.%
		
		We show the convexity of $\family{M}_{\epsilon^{\prime}}$ by showing that the convex combination $t\varphi_{1} + \rbracket{1 - t}\varphi_{2}$ belongs to  $\family{M}_{\epsilon^{\prime}}$, for any $\varphi_{1}, \varphi_{2}\in\family{M}_{\epsilon^{\prime}}$, and $t \in \sbracket{0, 1}$. Let $-\infty < x \leq y < \infty$. We have $\epsilon^{\prime} \leq \function{\varphi}{1}{}{x} \leq \function{\varphi}{1}{}{y} \leq 1 - \epsilon^{\prime}$ and $\epsilon^{\prime} \leq \function{\varphi}{2}{}{x} \leq \function{\varphi}{2}{}{y} \leq 1 - \epsilon^{\prime}$. These inequalities imply that $\epsilon^{\prime} \leq t\function{\varphi}{1}{}{x} + \rbracket{1-t}\function{\varphi}{2}{}{x} \leq t\function{\varphi}{1}{}{y} + \rbracket{1-t}\function{\varphi}{2}{}{y} \leq 1 - \epsilon^{\prime}$. Hence, $t\varphi_{1} + \rbracket{1 - t}\varphi_{2} \in \family{M}_{\epsilon^{\prime}}$.%
		
		Define $H\rbracket{t} = l_{\mat{\theta}}\rbracket{t\varphi_{1} + \rbracket{1 - t}\varphi_{2}}$, where $0 < t < 1$ and $\varphi_{1}, \varphi_{2} \in \family{M}_{\epsilon^{\prime}}$ with $\varphi_{1} \neq \varphi_{2}$. We show that the second derivative of $H\rbracket{t}$ is negative on $\rbracket{0, 1}$ using the Gateaux derivative and hence $\mapping{\varphi}{l_{\mat{\theta}}\rbracket{\varphi}}$ is strictly concave. For simplicity, denote $F_{u,0} = 1 - \uncuredsurvival{\mat{Z}}{Y}{}$ with the true parameters $\mat{\beta}_{0}$ and $\Lambda_{0}$, as in \eqref{eq:uncured_survival}, and $F_{u} = 1 - \uncuredsurvival{\mat{Z}}{Y}{}$ with parameters $\mat{\beta}$ and $\Lambda$. One can show that
		\[
		\begin{split}
		H''\rbracket{t} 
		&=
		\expect[\sbracket]{
			\frac{
				-\Delta\sbracket{\link[\tr{\mat{\gamma}}\mat{X}]{1}{} - \link[\tr{\mat{\gamma}}\mat{X}]{2}{}}^{2}
			}{
				\sbracket{t\link[\tr{\mat{\gamma}}\mat{X}]{1}{} + \rbracket{1-t}\link[\tr{\mat{\gamma}}\mat{X}]{2}{}}^{2}
			}}\\
		& + \expect[\sbracket]{\frac{
				-\rbracket{1 - \Delta}F_{u}^{2}\sbracket{\link[\tr{\mat{\gamma}}\mat{X}]{1}{} - \link[\tr{\mat{\gamma}}\mat{X}]{2}{}}^{2}
			}{
				\sbracket{1 - \cbracket{t\link[\tr{\mat{\gamma}}\mat{X}]{1}{} + \rbracket{1-t}\link[\tr{\mat{\gamma}}\mat{X}]{2}{}}F_{u}}^{2}
			}					
		}.
	\end{split}	\]
		Bounding the denominators by one and replacing $\Delta$ by $$\expect[\sbracket]{\Delta | \mat{X}, Y, \mat{Z}} = \link[\tr{\mat{\gamma}}_{0}\mat{X}]{0}{}F_{u,0}\rbracket{Y | \mat{Z}}$$ using the law of iterated expectations, we obtain 
		\begin{equation}
					\label{eq:second_deriv_expectation}
				H''\rbracket{t} 
				\leq
				-\mathbb{E}\left[
					\varphi_{0}(\tr{\mat{\gamma}}_{0}\mat{x}) F_{u,0}
					\cbracket{1 - \varphi_{0}(\tr{\mat{\gamma}}_{0}\mat{x}) F_{u,0}}
					\cbracket{\varphi_{1}(\tr{\mat{\gamma}}\mat{x})  - \varphi_{2}(\tr{\mat{\gamma}}\mat{x}) }^{2}
				\right].
		\end{equation}
		By Assumption \ref{enum:latency_assumptions}\ref{enum:cure_threshold}, we have 
		$\inf_{\mat{x} \in \family{X}}\prob[\sbracket]{C > \tau_{0}|\mat{X} = \mat{x}} = c > 0$. Therefore, since $\varphi_{0}(\tr{\mat{\gamma}}_{0}\mat{x}) \geq \epsilon$,
		\begin{equation}
			\begin{split}
				\label{eq:positive_expect_censor}
				\expect[\sbracket]{\Delta | \mat{X}} 
				&= \prob[\sbracket]{B = 1, T \leq C | \mat{X}}\geq \prob[\sbracket]{B = 1, C \geq \tau_{0} | \mat{X}}= \prob[\sbracket]{B = 1 | \mat{X}}\prob[\sbracket]{C \geq \tau_{0} | \mat{X}}\\%
				&\geq \epsilon \cdot  \inf_{\mat{x} \in \family{X}}\prob[\sbracket]{C > \tau_{0}|\mat{X}=\mat{x}} = c\epsilon > 0.
			\end{split}
		\end{equation}		
		Thus, using the law of iterated expectations and $\link[\tr{\mat{\gamma}}_{0}\mat{X}]{0}{}F_{u,0} \leq 1 -\epsilon$, we obtain
		\begin{align}
			\label{eq:second_deriv_expectation_lb}
			\begin{split}
				&\expect[\sbracket]{
					\link[\tr{\mat{\gamma}}_{0}\mat{X}]{0}{}F_{u,0}
					\sbracket{1 - \link[\tr{\mat{\gamma}}_{0}\mat{X}]{0}{}F_{u,0}}
					\sbracket{\link[\tr{\mat{\gamma}}\mat{X}]{1}{} - \link[\tr{\mat{\gamma}}\mat{X}]{2}{}}^{2}
				}\\%
				&\geq
				\epsilon\cdot\expect[\sbracket]{
					\expect[\sbracket]{\Delta | \mat{X}, Y, \mat{Z}}
					\sbracket{\link[\tr{\mat{\gamma}}\mat{X}]{1}{} - \link[\tr{\mat{\gamma}}\mat{X}]{2}{}}^{2}
				}\\%
				&=\epsilon\cdot\expect[\sbracket]{
					\expect[\sbracket]{\Delta | \mat{X}}
					\sbracket{\link[\tr{\mat{\gamma}}\mat{X}]{1}{} - \link[\tr{\mat{\gamma}}\mat{X}]{2}{}}^{2}
				}\\%
				&\geq
				c\epsilon^{2}\cdot\expect[\sbracket]{
					\cbracket{\link[\tr{\mat{\gamma}}\mat{X}]{1}{} - \link[\tr{\mat{\gamma}}\mat{X}]{2}{}}^{2}
				} > 0.
			\end{split}
		\end{align}
		Hence $H''\rbracket{t}  < 0$, for $0 < t < 1$.
	\end{proof}
	
	\begin{proof}[Proof of Proposition \ref{prop:true_link}]
		If $\mat{\theta} = \mat{\theta}_{0} = \rbracket{\mat{\gamma}_{0},\mat{\beta}_{0},\Lambda_{0}}$, by the definition of $\varphi_{0,\mat{\theta}}$ in Proposition \ref{prop:link_maximizer}, we have
		\begin{equation*}
			\expect[\sbracket]{l\rbracket{Y, \Delta, \mat{X}, \mat{Z}; \mat{\theta}_{0}, \varphi}} \leq \expect[\sbracket]{l\rbracket{Y, \Delta, \mat{X}, \mat{Z}; \mat{\theta}_{0}, \varphi_{0,\mat{\theta}_{0}}}}\quad\text{for all }\varphi \in \family{M}_{\epsilon^{\prime}},
		\end{equation*}
		with equality if $\varphi = \varphi_{0,\mat{\theta}}$. From the Kullback-Leibler inequality, we also have
		\begin{equation*}
			\expect[\sbracket]{l\rbracket{Y, \Delta, \mat{X}, \mat{Z}; \mat{\theta}_{0}, \varphi}} \leq \expect[\sbracket]{l\rbracket{Y, \Delta, \mat{X}, \mat{Z}; \mat{\theta}_{0}, \varphi_{0}}}\quad\text{for all }\varphi \in \family{M}_{\epsilon^{\prime}},
		\end{equation*}
		with equality if $\varphi = \varphi_{0}$. Since,  as shown in Proposition \ref{prop:link_maximizer}, $\varphi_{0,\mat{\theta}_{0}}$ is the unique maximizer of $\varphi\mapsto\expect[\sbracket]{l\rbracket{Y, \Delta, \mat{X}, \mat{Z}; \mat{\theta}_{0}, \varphi}}$, it follows that $\varphi_{0,\mat{\theta}_{0}} = \varphi_{0}$.
	\end{proof}
	
	\begin{proof}[Proof of Proposition \ref{prop:existence_mle_theta}]
		We observe that, as in the standard logistic-Cox cure model, the maximizer $\hat\Lambda_n$ must be a step function with jumps at the observed times. 	Let  $\lambda\rbracket{t}$ be the jump size of $\Lambda$ at $t$.  Also, $\mapping{u}{\function{\hat{\varphi}}{n,\mat{\theta}}{s}{u}}$ is continuous on $\real$, therefore the likelihood function $L_{n}$ is continuous with respect to {$\mat{\gamma}$, $\mat\beta$} and the jump sizes of $\Lambda$. Hence the existence and finiteness of the maximum likelihood estimator  $(\hat{\mat{\gamma}}_{n},\hat{\mat{\beta}}_{n},\hat{\Lambda}_{n})$ follows as in the proof of Theorem 1 in \cite{L2008}.
	\end{proof}
	
	\begin{proof}[Proof of Proposition \ref{prop:link_continuity}]
		Recall that $\varphi_{0,\mat{\theta}} = \argmax_{\varphi \in \family{M}_{\epsilon^{\prime}}}l_{\mat{\theta}}\rbracket{\varphi}$ for fixed $\mat{\theta}$. 
		Let $\tilde{\varphi} \in \family{M}_{\epsilon^{\prime}}$ and define
		$H(t) = l_{\mat{\theta}}\rbracket{t\tilde{\varphi} + \rbracket{1 - t}\varphi_{0,\mat{\theta}}}$, where $0 \leq t \leq 1$. By the definition of $\varphi_{0,\mat{\theta}}$ and the convexity of $\family{M}_{\epsilon^{\prime}}$, $H'(0)\leq 0$ for all $\tilde{\varphi} \in \family{M}_{\epsilon^{\prime}}$. We show the continuity of $\varphi_{0,\mat{\theta}}$ by contradiction.%
		
		For simplifying the notation, we denote $\tilde{\varphi} = \tilde{\varphi}(\tr{\mat{\gamma}}\mat{X})$, $\varphi_{0,\mat{\theta}} = \varphi_{0,\mat{\theta}}(\tr{\mat{\gamma}}\mat{X})$, $F_{u,0} = F_{u}(Y | \mat{Z};\mat{\beta}_0,\Lambda_0)$, and $F_{u} = F_{u}(Y | \mat{Z};\mat{\beta},\Lambda)$. Then, for all $\tilde{\varphi} \in \family{M}_{\epsilon^{\prime}}$,
		\begin{align}
			\begin{split}
				\label{eq:first_deriv_expectation}
				H'(0)
				&=
				\expect[\sbracket]{
					\frac{
						\Delta\rbracket{\tilde{\varphi} - \varphi_{0,\mat{\theta}}}
					}{
						\varphi_{0,\mat{\theta}}
					} - 
					\frac{
						\rbracket{1 - \Delta}\rbracket{\tilde{\varphi} - \varphi_{0,\mat{\theta}}}F_{u}
					}{
						1 - \varphi_{0,\mat{\theta}}F_{u}
					}					
				}\\
			&=
				\expect[\sbracket]{{
						\frac{
							\tilde{\varphi} - \varphi_{0,\mat{\theta}}
						}{
							\varphi_{0,\mat{\theta}}
						}
					}{
						\frac{
							\Delta - \varphi_{0,\mat{\theta}}F_{u}
						}{
							1 - \varphi_{0,\mat{\theta}}F_{u}
						}
					}	
				}\\%
				&=
				\expect[\sbracket]{{
						\frac{
							\tilde{\varphi} - \varphi_{0,\mat{\theta}}
						}{
							\varphi_{0,\mat{\theta}}
						}
					}{
						\frac{
							\expect[\sbracket]{\Delta| \mat{X}, Y, \mat{Z}} - \varphi_{0,\mat{\theta}}F_{u}
						}{
							1 - \varphi_{0,\mat{\theta}}F_{u}
						}
					}	
				}\\
			&=
				\expect[\sbracket]{{
						\frac{
							\tilde{\varphi} - \varphi_{0,\mat{\theta}}
						}{
							\varphi_{0,\mat{\theta}}
						}
					}{
						\frac{
							\link[\tr{\mat{\gamma}}_{0}\mat{X}]{0}{}F_{u,0} - \varphi_{0,\mat{\theta}}F_{u}
						}{
							1 - \varphi_{0,\mat{\theta}}F_{u}
						}
					}	
				}\\%
				&=
				\expect[\sbracket]{{
						\frac{
							\tilde{\varphi} - \varphi_{0,\mat{\theta}}
						}{
							\varphi_{0,\mat{\theta}}
						}
					}\cbracket{1 - \expect[\sbracket]{\left.
							\frac{
								1 - \link[\tr{\mat{\gamma}}_{0}\mat{X}]{0}{}F_{u,0}
							}{
								1 - \varphi_{0,\mat{\theta}}F_{u}
							}
							\right\vert\tr{\mat{\gamma}}\mat{X}
						}
					}	
				}\leq 0.
			\end{split}
		\end{align}
		Next we argue that, if $\varphi_{0,\mat{\theta}}$ was not continuous, we can construct $\tilde{\varphi} \in \family{M}_{\epsilon^{\prime}}$ for which the inequality \eqref{eq:first_deriv_expectation} is not satisfied. Hence, by contradiction we can conclude that $\varphi_{0,\mat{\theta}}$ is continuous. 
		
		Since $\varphi_{0,\mat{\theta}}$ is monotone non-decreasing, discontinuity points would be points of jump. Assume that $\varphi_{0,\mat{\theta}}$ has a jump at $\tr{\mat{\gamma}}\tilde{\mat{x}}$ with a size of $\kappa$. Let
		\[
		A(u)=	\expect[\sbracket]{\left.
			\frac{
				1 - \link[\tr{\mat{\gamma}}_{0}\mat{X}]{0}{}F_{u,0}
			}{
				1 - \link[\tr{\mat{\gamma}}\tilde{\mat{x}}]{0,\mat{\theta}}{}F_{u}
			}
			\right\vert\tr{\mat{\gamma}}\mat{X}=u
		} 
		\]
		
		We consider three cases based on whether $A(\tr{\mat{\gamma}}\tilde{\mat{x}})$
		is smaller than one, larger than one or equal to one. 
		
		\textit{Case 1.}
		Suppose that $A(\tr{\mat{\gamma}}\tilde{\mat{x}})< 1$. By  Assumption (A8), there exists $\eta > 0$, such that 
		$A(u) < 1$ for $u \in ({\tr{\mat{\gamma}}\tilde{\mat{x}} - \eta, \tr{\mat{\gamma}}\tilde{\mat{x}} + \eta})$. Since $\varphi_{0,\mat{\theta}}$ is monotone non-decreasing and $\varphi_{0,\mat{\theta}}$ has a jump at $\tr{\mat{\gamma}}\tilde{\mat{x}}$, we consider that $\varphi_{0,\mat{\theta}}$ is either right- or left-continuous at this point and construct $\tilde{\varphi}$ which coincides with $\varphi_{0,\mat\theta}$ apart from in a small neighborhood of $\tr{\mat{\gamma}}\tilde{\mat{x}}$.
		\begin{enumerate}[label=(\roman*)]
			\item
			If $\varphi_{0,\mat{\theta}}$ is right-continuous at $\tr{\mat{\gamma}}\tilde{\mat{x}}$,
			$\link[u]{0,\mat{\theta}}{} < \link[\tr{\mat{\gamma}}\tilde{\mat{x}}]{0,\mat{\theta}}{}$ for $ u \in \lbrack\tr{\mat{\gamma}}\tilde{\mat{x}} - \frac{\eta}{2}, \tr{\mat{\gamma}}\tilde{\mat{x}})$. We define $\tilde{\varphi}\rbracket{u} = \link[\tr{\mat{\gamma}}\tilde{\mat{x}}]{0,\mat{\theta}}{} > \link[u]{0,\mat{\theta}}{}$ for $u \in \lbrack\tr{\mat{\gamma}}\tilde{\mat{x}} - \frac{\eta}{2}, \tr{\mat{\gamma}}\tilde{\mat{x}})$ and $\tilde{\varphi}\rbracket{u} = \link[u]{0,\mat{\theta}}{}$ otherwise.
			\item
			If $\varphi_{0,\mat{\theta}}$ is left-continuous at $\tr{\mat{\gamma}}\tilde{\mat{x}}$,
			$\link[u]{0,\mat{\theta}}{} \leq \link[\tr{\mat{\gamma}}\tilde{\mat{x}}]{0,\mat{\theta}}{} < \link[\tr{\mat{\gamma}}\tilde{\mat{x}}]{0,\mat{\theta}}{} + \kappa$ for $ u \in \lbrack\tr{\mat{\gamma}}\tilde{\mat{x}} - \frac{\eta}{2}, \tr{\mat{\gamma}}\tilde{\mat{x}})$.
			We define $\tilde{\varphi}\rbracket{u} = \link[\tr{\mat{\gamma}}\tilde{\mat{x}}]{0,\mat{\theta}}{} + \kappa > \link[u]{0,\mat{\theta}}{}$ for $u \in \lbrack\tr{\mat{\gamma}}\tilde{\mat{x}} - \frac{\eta}{2}, \tr{\mat{\gamma}}\tilde{\mat{x}})$ and $\tilde{\varphi}\rbracket{u} = \link[u]{0,\mat{\theta}}{}$ otherwise.
		\end{enumerate}
		In addition, monotonicity of $\varphi_{0,\mat\theta}$ implies that
		\begin{equation*}
			\begin{split}
			\expect[\sbracket]{\left.
				\frac{
					1 - \varphi_{0}(\tr{\mat{\gamma}}_{0}\mat{X})F_{u,0}
				}{
					1 - \varphi_{0,\mat{\theta}}(u)F_{u}
				}
				\right\vert\tr{\mat{\gamma}}\mat{X}=u
			}
			&\leq
			\expect[\sbracket]{\left.
				\frac{
					1 - \varphi_{0}(\tr{\mat{\gamma}}_{0}\mat{X})F_{u,0}
				}{
					1 - \varphi_{0,\mat{\theta}}(\tr{\mat{\gamma}}\tilde{\mat{x}})F_{u}
				}
				\right\vert\tr{\mat{\gamma}}\mat{X}=u
			}=A(u)
			< 1,
			\end{split}
		\end{equation*}
		for $ u \in \lbrack\tr{\mat{\gamma}}\tilde{\mat{x}} - \frac{\eta}{2}, \tr{\mat{\gamma}}\tilde{\mat{x}})$. This inequality together with the construction of $\tilde{\varphi}$ in either case give
		\begin{equation*}
			\begin{split}
				H'\rbracket{0}
				&=
				\expect[\sbracket]{{
						\frac{
							\tilde{\varphi} - \varphi_{0,\mat{\theta}}
						}{
							\varphi_{0,\mat{\theta}}
						}
					}\cbracket{1 - \expect[\sbracket]{\left.
							\frac{
								1 - \link[\tr{\mat{\gamma}}_{0}\mat{X}]{0}{}F_{u,0}
							}{
								1 - \varphi_{0,\mat{\theta}}F_{u}
							}
							\right\vert\tr{\mat{\gamma}}\mat{X}
						}
					}	
				} \\
				&=	\mathbb{E}\left[{\indicator{\tr{\mat{\gamma}}\tilde{\mat{x}} - \frac{\eta}{2}\leq \tr{\mat{\gamma}}\mat{X}< \tr{\mat{\gamma}}\tilde{\mat{x}}}
						\frac{
							\tilde{\varphi} - \varphi_{0,\mat{\theta}}
						}{
							\varphi_{0,\mat{\theta}}
						}
					}\right.\\
				&\left.\qquad\qquad\cbracket{1 - \expect[\sbracket]{\left.
							\frac{
								1 - \link[\tr{\mat{\gamma}}_{0}\mat{X}]{0}{}F_{u,0}
							}{
								1 - \varphi_{0,\mat{\theta}}F_{u}
							}
							\right\vert\tr{\mat{\gamma}}\mat{X}
						}
					}	
				\right],
			\end{split}
		\end{equation*}
		is strictly larger than zero, 
		which contradicts the fact that  $H'(0) \leq 0$ for all $\tilde{\varphi} \in \family{M}_{\epsilon^{\prime}}$.
		
		\textit{Case 2.}
		Suppose that $A(\tr{\mat{\gamma}}\tilde{\mat{x}}) > 1$. By Assumption \ref{enum:expectation_ratio_continuity},
		there exists $\eta > 0$, such that 
		$A(u)> 1$ for $u \in ({\tr{\mat{\gamma}}\tilde{\mat{x}} - \eta, \tr{\mat{\gamma}}\tilde{\mat{x}} + \eta})$. Similar to the argument in the first case, we can show that
		\begin{equation*}
			\expect[\sbracket]{\left.
				\frac{
					1 - \varphi_{0}(\tr{\mat{\gamma}}_{0}\mat{X})F_{u,0}
				}{
					1 - \varphi_{0,\mat{\theta}}(u)F_{u}
				}
				\right\vert\tr{\mat{\gamma}}\mat{X}=u
			}
			\geq
				\expect[\sbracket]{\left.
				\frac{
					1 - \varphi_{0}(\tr{\mat{\gamma}}_{0}\mat{X})F_{u,0}
				}{
					1 - \varphi_{0,\mat{\theta}}(\tr{\mat{\gamma}}\tilde{\mat{x}})F_{u}
				}
				\right\vert\tr{\mat{\gamma}}\mat{X}=u
			}=A(u)
			> 1
		\end{equation*}
		for $u \in (\tr{\mat{\gamma}}\tilde{\mat{x}}, \tr{\mat{\gamma}}\tilde{\mat{x}} + \frac{\eta}{2}\rbrack$. Similarly to the first case, we can construct $\tilde{\varphi}$ which coincides with $\varphi_{0,\mat{\theta}}$ apart from in a small neighbourhood of $\tr{\mat{\gamma}}\tilde{\mat{x}}$ where it is constant and strictly smaller than $\varphi_{0,\mat{\theta}}$.  We then have $H'\rbracket{0} > 0$ and leads to a contradiction.
		
		\textit{Case 3.}
		Suppose that $A(\tr{\mat{\gamma}}\tilde{\mat{x}})
		= 1$.
		
\noindent	(i)				If $\varphi_{0,\mat{\theta}}$ is right-continuous at $\tr{\mat{\gamma}}\tilde{\mat{x}}$, we have $\link[u]{0,\mat{\theta}}{} \leq \link[\tr{\mat{\gamma}}\tilde{\mat{x}}]{0,\mat{\theta}}{} - \kappa < \link[\tr{\mat{\gamma}}\tilde{\mat{x}}]{0,\mat{\theta}}{}$ for $u < \tr{\mat{\gamma}}\tilde{\mat{x}}$. This implies that for $u < \tr{\mat{\gamma}}\tilde{\mat{x}}$,
			\begin{align}
				\begin{split}
					\label{eq:cond_expect_ineq}
				\expect[\sbracket]{\left.
					\frac{
						1 - \varphi_{0}(\tr{\mat{\gamma}}_{0}\mat{X})F_{u,0}
					}{
						1 - \varphi_{0,\mat{\theta}}(u)F_{u}
					}
					\right\vert\tr{\mat{\gamma}}\mat{X}=u
				}
					&\leq
					\expect[\sbracket]{\left.
						\frac{
							1 - \varphi_{0}(\tr{\mat{\gamma}}_{0}\mat{X})F_{u,0}
						}{
							1 - \sbracket{
								\link[\tr{\mat{\gamma}}\tilde{\mat{x}}]{0,\mat{\theta}}{} - \kappa
							}F_{u}
						}
						\right\vert\tr{\mat{\gamma}}\mat{X}=u
					}\\
					&<
					\expect[\sbracket]{\left.
						\frac{
							1 - \varphi_{0}(\tr{\mat{\gamma}}_{0}\mat{X})F_{u,0}
						}{
							1 - \link[\tr{\mat{\gamma}}\tilde{\mat{x}}]{0,\mat{\theta}}{}F_{u}
						}
						\right\vert\tr{\mat{\gamma}}\mat{X}=u
					}=A(u).\\
				\end{split}
			\end{align}
			We want to show that the left hand side of the previous inequality is smaller than one in a left neighborhood of $\tr{\mat{\gamma}}\tilde{\mat{x}}$ and then proceed with the construction of $\tilde{\varphi}$ as in Case 1. 
			By the mean value theorem, we have 
			\begin{equation*}
				\frac{
					1
				}{
					1 - \link[\tr{\mat{\gamma}}\tilde{\mat{x}}]{0,\mat{\theta}}{}F_{u}\rbracket{Y|\mat{Z}}
				} - 
				\frac{
					1
				}{
					1 - \sbracket{ 	\link[\tr{\mat{\gamma}}\tilde{\mat{x}}]{0,\mat{\theta}}{} - \kappa
					}F_{u}\rbracket{Y|\mat{Z}}
				} = \frac{\kappa F_{u}\rbracket{Y|\mat{Z}}}{\rbracket{1 - \xi}^{2}}
			\end{equation*}
			for some $\sbracket{ 	\link[\tr{\mat{\gamma}}\tilde{\mat{x}}]{0,\mat{\theta}}{} - \kappa}F_{u}\rbracket{Y|\mat{Z}} < \xi < \link[\tr{\mat{\gamma}}\tilde{\mat{x}}]{0,\mat{\theta}}{}F_{u}\rbracket{Y|\mat{Z}}$. Since $\varphi_{0,\mat{\theta}} \in \family{M}_{\epsilon^{\prime}}$, we have $\link[\tr{\mat{\gamma}}\tilde{\mat{x}}]{0,\mat{\theta}}{} - \kappa \geq \epsilon^{\prime}$ and $1>\xi > \epsilon^{\prime}F_{u}\rbracket{Y|\mat{Z}} \geq 0$.	
			It follows that 
			\begin{equation}
					\label{eqn:difference}
			\begin{split}
							&\expect[\sbracket]{\left.
						\frac{
							1 - \link[\tr{\mat{\gamma}}_{0}\mat{X}]{0}{}F_{u,0}
						}{
							1 - \link[\tr{\mat{\gamma}}\tilde{\mat{x}}]{0,\mat{\theta}}{}F_{u}\rbracket{Y|\mat{Z}}
						} - 
						\frac{
							1 - \link[\tr{\mat{\gamma}}_{0}\mat{X}]{0}{}F_{u,0}
						}{
							1 - \sbracket{ 	\link[\tr{\mat{\gamma}}\tilde{\mat{x}}]{0,\mat{\theta}}{} - \kappa
							}F_{u}\rbracket{Y|\mat{Z}}
						}
						\right\vert\tr{\mat{\gamma}}\mat{X}=u
					}\\%
					&\geq
					\expect[\sbracket]{\left.
						\cbracket{
							1 - \link[\tr{\mat{\gamma}}_{0}\mat{X}]{0}{}F_{u,0}
						}\kappa F_{u}\rbracket{Y|\mat{Z}}
						\right\vert\tr{\mat{\gamma}}\mat{X}=u
					}\\%
					&=
					\kappa\cdot\expect[\sbracket]{\left.
						\expect[\sbracket]{
							1 - \Delta | \mat{X}, Y, \mat{Z}	
						}F_{u}\rbracket{Y|\mat{Z}}
						\right\vert\tr{\mat{\gamma}}\mat{X}=u
					}\\%
					&\geq
					\kappa\cdot\expect[\sbracket]{\left.
						\rbracket{1 - \Delta}F_{u}\rbracket{Y|\mat{Z}}
						\indicator{Y > \tau_{0}}
						\right\vert\tr{\mat{\gamma}}\mat{X}=u
					}\\%
					&=
					\kappa\cdot\prob[\sbracket]{\left.
						Y > \tau_{0}
						\right\vert\tr{\mat{\gamma}}\mat{X}=u
					}\\%
					&=
					\kappa\cdot\prob[\sbracket]{\left.
						B = 0, C > \tau_{0}
						\right\vert\tr{\mat{\gamma}}\mat{X}=u
					}\\%
					&=
					\kappa\cdot\expect[\sbracket]{\left.
						\expect[\sbracket]{\left.
							\indicator{B = 0}\indicator{C > \tau_{0}}
							\right\vert\mat{X}
						}
						\right\vert\tr{\mat{\gamma}}\mat{X}=u
					}\\%
					&\geq
					c\epsilon\kappa>0.
			\end{split}
		\end{equation}
			Here the second equality follows from the definition of $\tau_{0}$ in Assumption \ref{enum:latency_assumptions}\ref{enum:cure_threshold} 
			and $F_{u}\rbracket{y|\mat{z}} = 1$ if $y > \tau_{0}$. The last inequality follows from $\prob[\sbracket]{B = 1, C \geq \tau_{0} | \mat{X}} \geq c\epsilon$ as shown in \eref{eq:positive_expect_censor}. 
			By Assumption \ref{enum:expectation_ratio_continuity}, 
			there exists $\eta > 0$, such that 
			$A(u)< 1 + c\epsilon\kappa$ for $u \in ({\tr{\mat{\gamma}}\tilde{\mat{x}} - \eta, \tr{\mat{\gamma}}\tilde{\mat{x}} + \eta})$.  Equation \eqref{eqn:difference} implies that, for $u \in (\tr{\mat{\gamma}}\tilde{\mat{x}} - \eta,$ $ \tr{\mat{\gamma}}\tilde{\mat{x}})$, we have
			\begin{align*}
				\begin{split}
					&\expect[\sbracket]{\left.
						\frac{
							1 - \link[\tr{\mat{\gamma}}_{0}\mat{X}]{0}{}F_{u,0}
						}{
							1 - \sbracket{ 	\link[\tr{\mat{\gamma}}\tilde{\mat{x}}]{0,\mat{\theta}}{} - \kappa
							}F_{u}
						}
						\right\vert\tr{\mat{\gamma}}\mat{X}=u
					} + c\epsilon\kappa\\%
					&\leq
					\expect[\sbracket]{\left.
						\frac{
							1 - \link[\tr{\mat{\gamma}}_{0}\mat{X}]{0}{}F_{u,0}
						}{
							1 - \link[\tr{\mat{\gamma}}\tilde{\mat{x}}]{0,\mat{\theta}}{}F_{u}
						}
						\right\vert\tr{\mat{\gamma}}\mat{X}=u
					}=A(u)
					< 1 + c\epsilon\kappa.
				\end{split}
			\end{align*}
			From \eqref{eq:cond_expect_ineq}, it follows that, for $u \in ({\tr{\mat{\gamma}}\tilde{\mat{x}} - \eta, \tr{\mat{\gamma}}\tilde{\mat{x}}})$,
			\begin{equation*}
				\begin{split}
				&\expect[\sbracket]{\left.
					\frac{
						1 - \varphi_{0}(\tr{\mat{\gamma}}_{0}\mat{X})F_{u,0}
					}{
						1 - \link[u]{0,\mat{\theta}}{}F_{u}
					}
					\right\vert\tr{\mat{\gamma}}\mat{X}=u
				}\\
			&
				\leq 
				\expect[\sbracket]{\left.
					\frac{
						1 -\varphi_{0}(\tr{\mat{\gamma}}_{0}\mat{X})F_{u,0}
					}{
						1 - \sbracket{ 	\link[\tr{\mat{\gamma}}\tilde{\mat{x}}]{0,\mat{\theta}}{} - \kappa
						}F_{u}
					}
					\right\vert\tr{\mat{\gamma}}\mat{X}=u
				}
				< 1.
				\end{split}
			\end{equation*}
			By a similar construction of $\tilde{\varphi}$ as in the first case, we get $H'(0) > 0$ which leads to a contradiction.%
			
	\noindent (ii)
			If $\varphi_{0,\mat{\theta}}$ is left-continuous at $\tr{\mat{\gamma}}\tilde{\mat{x}} $, we have
			$\link[\tr{\mat{\gamma}}\tilde{\mat{x}}]{0,\mat{\theta}}{} < \link[\tr{\mat{\gamma}}\tilde{\mat{x}}]{0,\mat{\theta}}{} + \kappa \leq \link[u]{0,\mat{\theta}}{}$ for $u > \tr{\mat{\gamma}}\tilde{\mat{x}}$. This implies that
			\begin{align*}
				\begin{split}
					\expect[\sbracket]{\left.
						\frac{
							1 - \varphi_{0}(\tr{\mat{\gamma}}_{0}\mat{X})F_{u,0}
						}{
							1 - \link[\tr{\mat{\gamma}}\tilde{\mat{x}}]{0,\mat{\theta}}{}F_{u}
						}
						\right\vert\tr{\mat{\gamma}}\mat{X}=u
					}
					&\leq
					\expect[\sbracket]{\left.
						\frac{
							1 - \link[\tr{\mat{\gamma}}_{0}\mat{X}]{0}{}F_{u,0}
						}{
							1 - \sbracket{
								\link[\tr{\mat{\gamma}}\tilde{\mat{x}}]{0,\mat{\theta}}{} + \kappa
							}F_{u}
						}
						\right\vert\tr{\mat{\gamma}}\mat{X}=u
					}\\
					&<
					\expect[\sbracket]{\left.
						\frac{
							1 - \link[\tr{\mat{\gamma}}_{0}\mat{X}]{0}{}F_{u,0}
						}{
							1 - \link[u]{0,\mat{\theta}}{}F_{u}
						}
						\right\vert\tr{\mat{\gamma}}\mat{X}=u
					}.\\
				\end{split}
			\end{align*}
			As in (i),
			\begin{equation*}
				\frac{
					1
				}{
					1 - \sbracket{ 	\link[\tr{\mat{\gamma}}\tilde{\mat{x}}]{0,\mat{\theta}}{} + \kappa
					}F_{u}\rbracket{Y|\mat{Z}}
				} -
				\frac{
					1
				}{
					1 - \link[\tr{\mat{\gamma}}\tilde{\mat{x}}]{0,\mat{\theta}}{}F_{u}
				} 
				= \frac{
					\kappa F_{u}
				}{\rbracket{1 - \xi}^{2}}
			\end{equation*}
			for some $\link[\tr{\mat{\gamma}}\tilde{\mat{x}}]{0,\mat{\theta}}{}F_{u}
			< \xi < \sbracket{ 	\link[\tr{\mat{\gamma}}\tilde{\mat{x}}]{0,\mat{\theta}}{} + \kappa}F_{u}
			$. Using a similar argument in (i), one can show that
			\begin{equation*}
				\expect[\sbracket]{\left.
					\frac{
						1 - \link[\tr{\mat{\gamma}}_{0}\mat{X}]{0}{}F_{u,0}
					}{
						1 - \sbracket{ 	\link[\tr{\mat{\gamma}}\tilde{\mat{x}}]{0,\mat{\theta}}{} + \kappa
						}F_{u}
					}
					\right\vert\tr{\mat{\gamma}}\mat{X}=u
				}
				>A(u)+ c\epsilon\kappa.
			\end{equation*}
			By Assumption \ref{enum:expectation_ratio_continuity}, 
			there exists $\eta > 0$, such that $A(u)> 1 - c\epsilon\kappa$ for $u \in ({\tr{\mat{\gamma}}\tilde{\mat{x}} - \eta, \tr{\mat{\gamma}}\tilde{\mat{x}} + \eta})$.
			Hence, we have for $u \in ({\tr{\mat{\gamma}}\tilde{\mat{x}}, \tr{\mat{\gamma}}\tilde{\mat{x}}} + \eta)$,
			\begin{equation*}
				\expect[\sbracket]{\left.
					\frac{
						1 - \varphi_{0}(\tr{\mat{\gamma}}_{0}\mat{X})F_{u,0}
					}{
						1 - \link[u]{0,\mat{\theta}}{}F_{u}
					}
					\right\vert\tr{\mat{\gamma}}\mat{X}=u
				}
				\geq 
				\expect[\sbracket]{\left.
					\frac{
						1 - \varphi_{0}(\tr{\mat{\gamma}}_{0}\mat{X})F_{u,0}
					}{
						1 - [	\varphi_{0,\mat{\theta}}(\tr{\mat{\gamma}}\tilde{\mat{x}}) + \kappa
						]F_{u}
					}
					\right\vert\tr{\mat{\gamma}}\mat{X}=u
				}
				> 1.
			\end{equation*}
			By a similar construction of $\tilde{\varphi}$ as in the second case, we get $H'(0) > 0$ and leads to a contradiction.%
	\end{proof}

For the proof of Proposition \ref{prop:link_mle_convergence}
we need results on the entropy numbers of certain classes of functions, which we derive below. First we introduce some notation that will be used in the following series of Lemmas.  Consider a class of functions $\family{F}$ equipped with a norm $\norm{\cdot}{}$. For $\zeta > 0$, the bracketing number $\bracketnum{\family{F}}{\norm{\cdot}{}} = N$ is the minimal number of pairs of functions $\{[{f_{j}^{L}, f_{j}^{U}}]\colon j = 1, \cdots, N\}$ such that $\Vert{f_{j}^{U} - f_{j}^{L}}\Vert\leq \zeta$, for all $j = 1, \cdots, N$ and for each $f \in \family{F}$, there is a $j \in \cbracket{1, \cdots, N}$ such that $f_{j}^{L} \leq f \leq f_{j}^{U}$. The $\zeta$--entropy with bracketing of $\family{F}$ is the logarithm of the bracketing number, $\bracketentropy{\family{F}}{\norm{\cdot}{}} = \log\bracketnum{\family{F}}{\norm{\cdot}{}}$. $\mathbb{P}_{0}$ denotes the distribution of $\rbracket{Y, \Delta, \mat{X}, \mat{Z}}$. We will consider the following classes of functions.
\begin{itemize}
	\item
	$\family{M}$ is the class of all non-decreasing bounded functions on $\real$.
	\item
	$\family{F}$ is the class of functions $F\rbracket{y, \mat{z}} = 
	1 - \exp\sbracket{-\function{\Lambda}{}{}{y}\exp\rbracket{\tr{\mat{\beta}}\mat{z}}}$, $y \in [0,\tau_{0}]$ and $\mat{z} \in \family{Z}$, where  $\Lambda \in \tilde{\family{D}}$ and $\mat{\beta} \in \family{B}$.
	\item
	$\family{G}_{\epsilon^{\prime}}$ is the class of functions $g\rbracket{y, \delta, \mat{x}, \mat{z}} = \rbracket{1 - \delta}\log\sbracket{1 - \link[\tr{\mat{\gamma}}\mat{x}]{}{}\function{F}{}{}{y,\mat{z}}}$, $y \in [0, \tau_{0}]$, $\delta \in \cbracket{0, 1}$, $\mat{x} \in \family{X}$ and $\mat{z} \in \family{Z}$, where $\mat{\gamma} \in \family{S}_{d-1}$, $\varphi \in \family{M}_{\epsilon^{\prime}}$ and $F \in \family{F}$,
	\item
	$\family{H}_{\epsilon^{\prime}}$ is the class of functions $h\rbracket{\delta, \mat{x}} = \delta\log\link[\tr{\mat{\gamma}}\mat{x}]{}{}$, $\delta \in \cbracket{0, 1}$ and $\mat{x} \in \family{X}$, where $\mat{\gamma} \in \family{S}_{d-1}$ and $\varphi \in \family{M}_{\epsilon^{\prime}}$.
	\item
	$\family{L}_{\epsilon^{\prime}}$ is the class of functions $$l\rbracket{y, \delta, \mat{x}, \mat{z}} = \delta\log\link[\tr{\mat{\gamma}}\mat{x}]{}{} + \rbracket{1 - \delta}\log\sbracket{1 - \link[\tr{\mat{\gamma}}\mat{x}]{}{}\function{F}{}{}{y,\mat{z}}},$$ $y \in [0, \tau_{0}]$, $\delta \in \cbracket{0, 1}$, $\mat{x} \in \family{X}$ and $\mat{z} \in \family{Z}$, where $\mat{\gamma} \in \family{S}_{d-1}$, $\varphi \in \family{M}_{\epsilon^{\prime}}$ and $F_{} \in \family{F}$.
	\item
	$\tilde{\family{L}}_{\epsilon^{\prime}}$ is the class of functions
	\begin{align*}
		\tilde{l}\rbracket{y, \delta, \mat{x}, \mat{z}} &= \delta\sbracket{
			\log\varphi_1(\tr{\mat{\gamma}}\mat{x}) -
			\log\varphi_2(\tr{\mat{\gamma}}\mat{x})
		}\\
		& + (1 - \delta)\{
			\log[1 - \varphi_1(\tr{\mat{\gamma}}\mat{x})F(y,\mat{z})] -
			\log[1 - \varphi_1(\tr{\mat{\gamma}}\mat{x})F(y,\mat{z})]\}
		,
	\end{align*}
	for	$y \in [0, \tau_{0}]$, $\delta \in \cbracket{0, 1}$, $\mat{x} \in \family{X}$, $\mat{z} \in \family{Z}$, $\mat{\gamma} \in \family{S}_{d-1}$, $\varphi_{1}, \varphi_{2} \in \family{M}_{\epsilon^{\prime}}$ and $F\in \family{F}$.
\end{itemize}

We begin with a result from \cite{VW1996} 
and we use it to construct $\zeta$-brackets for the other classes.
\begin{lemma}[Theorem 2.7.5 of \cite{VW1996}
	]
	\label{lemma:bracket_entropy_M}
	There exists a constant $A > 0$ such that
	\begin{equation*}
		\bracketentropy{\family{M}}{\norm{\cdot}{r,Q}} \leq \frac{A}{\zeta},
	\end{equation*}
	for all $\zeta > 0$, $r \geq 1$, and all probability measures $Q$ on $\real$, where $\norm{\cdot}{r,Q}$ is the $L_{r}$--norm corresponding to $Q$.
\end{lemma}

\begin{lemma}
	\label{lemma:bracket_entropy_F}
	Suppose that Assumptions \ref{enum:bounded_support} and \ref{enum:bound_dens} 
	are satisfied 
	Let $\zeta > 0$. There exists a constant $C_{1} > 0$ depending on $\bar{q}_{2}$, and $r_{2}$ only such that
	\begin{equation*}
		\bracketentropy{\family{F}}{\norm{\cdot}{1,\mathbb{P}_{0}}} \leq \frac{C_{1}\rbracket{q + 1}}{\zeta}.
	\end{equation*}
	
	Furthermore, 
	there exists a constant $C_{2} > 0$ depending on $\bar{q}_{1}$, and $r_{1}$ and a constant $C_{3} > 0$ depending on $\epsilon^{\prime}, \bar{q}_{1}, \bar{q}_{2}, r_{1}$, and $r_{2}$ such that
	\begin{equation*}
		\bracketentropy{\family{G}_{\epsilon^{\prime}}}{\norm{\cdot}{1,\mathbb{P}_{0}}} \leq \frac{C_{2}d + C_{3}\rbracket{q + 1}}{\zeta}.
	\end{equation*}
	
	Moreover, there exists a constant $C_{4} > 0$ depending on $\epsilon^{\prime}, \bar{q}_{1}$, and $r_{1}$ such that
	\begin{equation*}
		\bracketentropy{\family{H}_{\epsilon^{\prime}}}{\norm{\cdot}{1,\mathbb{P}_{0}}} \leq \frac{C_{4}\rbracket{d + 1}}{\zeta}.
	\end{equation*}
\end{lemma}

\begin{proof}
	Let $\zeta_{\mat{\beta}} \in (0, 1)$. Since $\family{B}$ is a compact subset of $\real[q]$, it can be covered by $N_{1}$ balls with radius $\zeta_{\mat{\beta}}$, where $N_{1} \leq (A_1/\zeta_{\mat{\beta}})^{q}$ for some constant $A_1 > 0$. Let $\cbracket{\mat{\beta}_{1}, \cdots, \mat{\beta}_{N_{1}}}$ be the centers of such balls. Consider $f\rbracket{y, \mat{z}} = 1 - \exp\sbracket{-\function{\Lambda}{}{}{y}\exp\rbracket{\tr{\mat{\beta}}\mat{z}}} \in \family{F}$, $y \in [0,\tau_{0}]$ and $\mat{z} \in \family{Z}$, for some $\Lambda \in \tilde{\family{D}}$ and $\mat{\beta} \in \family{B}$. We can find $j \in \cbracket{1, \cdots, N_{1}}$ such that $\norm{\mat{\beta} - \mat{\beta}_{j}}{2} \leq \zeta_{\mat{\beta}}$. By the monotonicity of the exponential function and the Cauchy-Schwarz inequality, we have
	\begin{equation*}
		\exp\rbracket{\tr{\mat{\beta}}_{j}\mat{z} - \zeta_{\mat{\beta}}r_{2}} 
		\leq 
		\exp\rbracket{\tr{\mat{\beta}}\mat{z}} 
		\leq
		\exp\rbracket{\tr{\mat{\beta}}_{j}\mat{z} + \zeta_{\mat{\beta}}r_{2}},
	\end{equation*}
	for all $\mat{z} \in \family{Z}$. Let $\zeta_{\Lambda} > 0$. By Lemma \ref{lemma:bracket_entropy_M},  the class $\tilde{\family{D}}$ can be covered by $\zeta_{\Lambda}$--brackets $[\Lambda_{l}^{L}, \Lambda_{l}^{U}]$, $l = 1, \cdots, N_{2}$, such that
	\begin{equation*}
		\int_{[0, \tau_{0}]}\sbracket{\function{\Lambda}{l}{U}{y} - \function{\Lambda}{l}{L}{y}}^{2}\dd \function{Q}{Y}{}{y} \leq \zeta_{\Lambda}^{2},
	\end{equation*}
	where ${Q}_{Y}$ denotes the distribution of $Y$ and $N_{2} \leq \exp{\rbracket{A_2/\zeta_{\Lambda}}}$ for some constant $A_2 > 0$. 
	We note that $\Lambda_{l}^{L}$ and $\Lambda_{l}^{U}$ can always be taken to be bounded below by 0 and bounded above by the uniform upper bound $M$ of the class $\tilde{\mathcal{D}}$. Otherwise, we can take $\Lambda_{l}^{L} \vee 0$ and $\Lambda_{l}^{U} \wedge M$.
	Then, we have
	\begin{equation*}
		1 - \exp\left[-\Lambda_l^L(y)\exp(\tr{\mat{\beta}}_{j}\mat{z} - \zeta_{\mat{\beta}}r_{2})\right]
		\leq 
		f(y, \mat{z})
		\leq
		1 - \exp\sbracket{-\Lambda_l^U(y)\exp(\tr{\mat{\beta}}_{j}\mat{z} +\zeta_{\mat{\beta}}r_{2})},
	\end{equation*}
	for some $l = 1, \cdots, N_{2}$, and for all $y \in [0, \tau_{0}]$ and $\mat{z} \in \family{Z}$. We then show that, for certain choices of $\zeta_{\Lambda}$ and $\zeta_{\mat{\beta}}$,  the brackets 
	\[
	\begin{split}
[f_l^L,f_l^U]=&\left[1 - \exp\cbracket{-\Lambda_l^L(y)\exp(\tr{\mat{\beta}}_{j}\mat{z} - \zeta_{\mat{\beta}}r_{2})},\right.\\
& \quad 1 - \left.\exp\cbracket{-\Lambda_l^U(y)\exp(\tr{\mat{\beta}}_{j}\mat{z} + \zeta_{\mat{\beta}}r_{2})}\right]
	\end{split}
		\]
	are $\zeta$-brackets with respect to $\Vert\cdot\Vert_{1,\mathbb{P}_{0}}$ that cover $\family{F}$. We have 
	\begin{align}
		\begin{split}
			\label{eq:bracket_size_F_L1}
			\Vert f_l^U-f_l^L\Vert_{1,\mathbb{P}_{0}}=	&\int_{[0, \tau_{0}] \times \family{Z}}
			\bigg|
			\exp\sbracket{-\function{\Lambda}{l}{U}{y}\exp\rbracket{\tr{\mat{\beta}}_{j}\mat{z} + \zeta_{\mat{\beta}}r_{2}}}\\%
			&- \exp\sbracket{-\function{\Lambda}{l}{L}{y}\exp\rbracket{\tr{\mat{\beta}}_{j}\mat{z} - \zeta_{\mat{\beta}}r_{2}}}
			\bigg|
			\dd Q_{Y, \mat{Z}}\left(y, \mat{z}\right).
		\end{split}
	\end{align}
	Considering the integrand of \eref{eq:bracket_size_F_L1}, by the mean value theorem we obtain
	\begin{align*}
		\begin{split}
			&\bigg|
			\exp\sbracket{-\function{\Lambda}{l}{U}{y}\exp\rbracket{\tr{\mat{\beta}}_{j}\mat{z} + \zeta_{\mat{\beta}}r_{2}}}
			- \exp\sbracket{-\function{\Lambda}{l}{L}{y}\exp\rbracket{\tr{\mat{\beta}}_{j}\mat{z} - \zeta_{\mat{\beta}}r_{2}}}
			\bigg|\\%
			&=
			\bigg|
			\exp\rbracket{-\xi}\sbracket{
				\function{\Lambda}{l}{L}{y}\exp\rbracket{\tr{\mat{\beta}}_{j}\mat{z} - \zeta_{\mat{\beta}}r_{2}}
				-
				\function{\Lambda}{l}{U}{y}\exp\rbracket{\tr{\mat{\beta}}_{j}\mat{z} + \zeta_{\mat{\beta}}r_{2}}
			}
			\bigg|\\
			&\leq 
			\bigg|
			\function{\Lambda}{l}{L}{y}\exp\rbracket{\tr{\mat{\beta}}_{j}\mat{z} - \zeta_{\mat{\beta}}r_{2}}
			-
			\function{\Lambda}{l}{U}{y}\exp\rbracket{\tr{\mat{\beta}}_{j}\mat{z} + \zeta_{\mat{\beta}}r_{2}}
			\bigg|,
		\end{split}
	\end{align*}
	for some $\xi \in (\function{\Lambda}{l}{L}{y}\exp\rbracket{\tr{\mat{\beta}}_{j}\mat{z} - \zeta_{\mat{\beta}}r_{2}}, \function{\Lambda}{l}{U}{y}\exp\rbracket{\tr{\mat{\beta}}_{j}\mat{z} + \zeta_{\mat{\beta}}r_{2}})$. As a result, from  Minkowski inequality it follows that
	\begin{align}
		\begin{split}
			\label{eq:bound_F_L1}
			\Vert f_l^U-f_l^L\Vert_{1,\mathbb{P}_{0}}
			&\leq
			\int_{[0, \tau_{0}] \times \family{Z}}		
			\left|
			\function{\Lambda}{l}{L}{y} - \function{\Lambda}{}{}{y}
			\right|\exp(\tr{\mat{\beta}}_{j}\mat{z} - \zeta_{\mat{\beta}}r_{2})
			\dd Q_{Y, \mat{Z}}\left(y, \mat{z}\right)\\%
			&+
			\int_{[0, \tau_{0}] \times \family{Z}}
			\Lambda(y)
			\left|
			\exp(\tr{\mat{\beta}}_{j}\mat{z} - \zeta_{\mat{\beta}}r_{2})
			-
			\exp(\tr{\mat{\beta}}_{j}\mat{z} + \zeta_{\mat{\beta}}r_{2})
			\right|
			\dd Q_{Y, \mat{Z}}\left(y, \mat{z}\right)\\%
			&+
			\int_{[0, \tau_{0}] \times \family{Z}}		
			\left|
			\function{\Lambda}{}{}{y} - \function{\Lambda}{l}{U}{y}
			\right|\exp(\tr{\mat{\beta}}_{j}\mat{z} + \zeta_{\mat{\beta}}r_{2})
			\dd Q_{Y, \mat{Z}}\left(y, \mat{z}\right),%
		\end{split}
	\end{align}
	Consider the first integral on the right hand side of the last inequality in \eref{eq:bound_F_L1}. Using the Cauchy-Schwarz inequality, we have
	\begin{align*}
		\begin{split}
			&\int_{[0, \tau_{0}] \times \family{Z}}
			\left|
			\function{\Lambda}{l}{L}{y} - \function{\Lambda}{}{}{y}
			\right|\exp\rbracket{\tr{\mat{\beta}}_{j}\mat{z} - \zeta_{\mat{\beta}}r_{2}}
			\dd Q_{Y, \mat{Z}}\left(y, \mat{z}\right)\\
			&\leq
			\cbracket{
				\int_{[0, \tau_{0}]}\sbracket{
					\function{\Lambda}{l}{L}{y} - \function{\Lambda}{}{}{y}
				}^{2}\dd Q_{Y}\rbracket{y}
				\cdot
				\int_{\family{Z}}\sbracket{
					\exp\rbracket{\tr{\mat{\beta}}_{j}\mat{z} - \zeta_{\mat{\beta}}r_{2}}
				}^{2}\dd Q_{\mat{Z}}\rbracket{\mat{z}}
			}^{\frac{1}{2}}\\
			&\leq
			\exp\rbracket{b}\zeta_{\Lambda} 
			\leq \exp\rbracket{b + r_{2}}\zeta_{\Lambda},
		\end{split}
	\end{align*}
	where $b$ is such that the support of $\tr{\mat{\beta}}_{j}\mat{Z}$ is included in an interval $\sbracket{a,b}$ for all $j$. 
	Similarly, we also have
	\begin{equation*}
		\int_{[0, \tau_{0}] \times \family{Z}}
		\left|
		\function{\Lambda}{}{}{y} - \function{\Lambda}{l}{U}{y}
		\right|\exp\rbracket{\tr{\mat{\beta}}_{j}\mat{z} + \zeta_{\mat{\beta}}r_{2}}
		\dd Q_{Y, \mat{Z}}\left(y, \mat{z}\right) \leq \exp\rbracket{b + r_{2}}\zeta_{\Lambda}.
	\end{equation*}
	For the second integral on the right hand side of  \eref{eq:bound_F_L1}, we can show that
	\begin{align*}
		\begin{split}
			&\int_{[0, \tau_{0}] \times \family{Z}}			
			\function{\Lambda}{}{}{y}
			\left|
			\exp\rbracket{\tr{\mat{\beta}}_{j}\mat{z} - \zeta_{\mat{\beta}}r_{2}}
			-
			\exp\rbracket{\tr{\mat{\beta}}_{j}\mat{z} + \zeta_{\mat{\beta}}r_{2}}
			\right|
			\dd Q_{Y, \mat{Z}}\left(y, \mat{z}\right)\\
			&\leq
			\bar{q}_{2}M
			\int_{a}^{b}
			\exp\rbracket{t + \zeta_{\mat{\beta}}r_{2}}
			-
			\exp\rbracket{t - \zeta_{\mat{\beta}}r_{2}}
			\dd t
			\\
			&=
			\bar{q}_{2}M\rbracket{
				\int_{b - \zeta_{\mat{\beta}}r_{2}}^{b + \zeta_{\mat{\beta}}r_{2}}
				\exp\rbracket{u}
				\dd u
				-
				\int_{a - \zeta_{\mat{\beta}}r_{2}}^{a + \zeta_{\mat{\beta}}r_{2}}
				\exp\rbracket{u}
				\dd u
			}\\
			&\leq
			2M\bar{q}_{2}\sbracket{
				\exp\rbracket{b + r_{2}} - \exp\rbracket{a - r_{2}}
			}r_{2}\zeta_{\mat{\beta}}\\
			&\leq
			2M\bar{q}_{2}r_{2}\exp\rbracket{b + r_{2}}\zeta_{\mat{\beta}},
		\end{split}
	\end{align*}
	where $M$ denotes the upper bound of the class of functions $\tilde{\mathcal{D}}$.
	If we take $\zeta_{\Lambda} = \frac{\zeta}{3\exp\rbracket{b + r_{2}}}$ and $\zeta_{\mat{\beta}} = 
	\frac{\zeta}{6M\bar{q}_2r_2\exp(b + r_{2})}$ , we obtain
	\begin{align*}
		\Vert f_l^U-f_l^L\Vert_{1,\mathbb{P}_{0}}\leq
		\zeta.
	\end{align*}
	Hence, using $\log x \leq x - 1$, for $x > 0$, we have
	\begin{align*}
		\begin{split}
			\bracketentropy[\zeta]{\family{F}}{\norm{\cdot}{1,\mathbb{P}_{0}}}
			&\leq
			\log N_{1} + \log N_{2}\\%
			&\leq
			q\log\rbracket{\frac{6A_1M\bar{q}_2r_2\exp\rbracket{b + r_{2}}}{\zeta}} +
			\frac{3A_2\exp\rbracket{b + r_{2}}}{\zeta}\\
			&\leq
			\frac{C_1\rbracket{q + 1}}{\zeta}.
		\end{split}
	\end{align*}
	
	Next we construct brackets for the class of functions $\family{G}_{\epsilon'}$ to show the second assertion of the Lemma. Let $\zeta_{\mat{\gamma}} > 0$. By Lemma 7.5 of \cite{BDJ2019}, $\family{S}_{d-1}$ can be covered by $N_{3}$ neighborhoods with diameter at most $\zeta_{\mat{\gamma}}$, where $N_{3} \leq \rbracket{A_3/\zeta_{\mat{\gamma}}}^{d}$ with a constant $A_3 > 0$. Let $\cbracket{\mat{\gamma}_{1}, \cdots, \mat{\gamma}_{N_{3}}}$ be elements of each these neighborhoods. Let $\zeta_{\varphi} > 0$ and consider $\zeta_{\varphi}$--brackets $[\varphi_{k}^{L}, \varphi_{k}^{U}]$, $k = 1, \cdots, N_{4}$, covering the class $\family{M}_{\epsilon^{\prime}}$ such that
	\begin{equation*}
		\int_{\real}
		\left|
		\function{\varphi}{k}{U}{t} - \function{\varphi}{k}{L}{t}
		\right|
		\dd\function{Q}{i}{\pm}{t} \leq \zeta_{\varphi},\quad k=1,\cdots,N_{4},
	\end{equation*}
	where $\function{Q}{i}{\pm}{t}$ denotes the distribution of $\tr{\mat{\gamma}}_{i}\mat{X} \pm \zeta_{\mat{\gamma}}r_{1}$, $i=1,\cdots,N_{3}$, and $N_{4} \leq \exp{\rbracket{A_4/\zeta_{\varphi}}}$ for some constant $A_4 > 0$, by Lemma \ref{lemma:bracket_entropy_M}. We note that $\varphi_{k}^{L}$ and $\varphi_{k}^{U}$ can always be taken to be bounded below by $\epsilon^{\prime}$ and bounded above by $1 - \epsilon^{\prime}$, respectively. Otherwise, we can take $\varphi_{k}^{L} \vee \epsilon^{\prime}$ and $\varphi_{k}^{U} \wedge 1 - \epsilon^{\prime}$. Using the monotonicity of $\varphi$ and the Cauchy-Schwarz inequality, we have, for all $\mat{x} \in \family{X}$,
	\begin{equation*}
		\function{\varphi}{k}{L}{\tr{\mat{\gamma}}_{i}\mat{x} - \zeta_{\mat{\gamma}}r_{1}}
		\leq 
		\function{\varphi}{}{}{\tr{\mat{\gamma}}\mat{x}}
		\leq
		\function{\varphi}{k}{U}{\tr{\mat{\gamma}}_{i}\mat{x} + \zeta_{\mat{\gamma}}r_{1}},
	\end{equation*}
	for some $i = 1,\cdots,N_{3}$ and $k = 1,\cdots,N_{4}$. Using the result from the first statement of the Lemma, we can consider a $\zeta_F$-bracket $[F_{j}^{L}, F_{j}^{U}]$, $j = 1, \cdots, N$, covering the class $\family{F}$ such that
	\begin{equation*}
		\int\left|\function{F}{j}{U}{y|\mat{z}} - \function{F}{j}{L}{y|\mat{z}}\right|\dd\mathbb{P}_{0}\left(y, \delta, \mat{x}, \mat{z}\right) \leq \zeta_F,
		\quad j=1, \cdots, N,
	\end{equation*}
	where $N \leq \exp[C_1(q+1)/\zeta_F]$. Then, we have
	\begin{align*}
		\begin{split}
			\log\sbracket{1 - \function{\varphi}{k}{U}{\tr{\mat{\gamma}}_{i}\mat{x} + \zeta_{\mat{\gamma}}r_{1}} \function{F}{j}{U}{y|\mat{z}}} 
			&\leq
			\log\sbracket{1 - \function{\varphi}{}{}{\tr{\mat{\gamma}}\mat{x}} \function{F}{u}{}{y|\mat{z}}}\\%
			&\leq
			\log\sbracket{1 - \function{\varphi}{k}{L}{\tr{\mat{\gamma}}_{i}\mat{x} - \zeta_{\mat{\gamma}}r_{1}} \function{F}{j}{L}{y|\mat{z}}},
		\end{split}
	\end{align*}
	for some $j=1,\cdots,N$. We then show that, for a certain choice of $\zeta_F$, the brackets
	\[
	\begin{split}
		[g_{i,j,k}^L,g_{i,j,k}^U]&=\left[(1-\delta)\log\left\{1 - \function{\varphi}{k}{U}{\tr{\mat{\gamma}}_{i}\mat{x} + \zeta_{\mat{\gamma}}r_{1}} \function{F}{j}{U}{y|\mat{z}}\right\},\right.\\
		&\qquad\left.(1-\delta) \log\cbracket{1 - \function{\varphi}{k}{L}{\tr{\mat{\gamma}}_{i}\mat{x} - \zeta_{\mat{\gamma}}r_{1}} \function{F}{j}{L}{y|\mat{z}}}\right]
	\end{split}
	\] are $\zeta$-brackets with respect to $\Vert\cdot\Vert_{1,\mathbb{P}_{0}}$ for the class $\family{G}_{\epsilon'}$. 	Using the Minkowski inequality, we have 
	\begin{equation}
		\label{eq:bracket_size_censored_L1_triangle}
		\begin{split}
			\Vert g_{i,j,k}^U-g_{i,j,k}^L\Vert_{1,\mathbb{P}_{0}}
			&\leq \int_{[0, \tau_{0}] \times \family{X} \times \family{Z}}
			\bigg|
			\log\sbracket{1 - \function{\varphi}{k}{U}{\tr{\mat{\gamma}}_{i}\mat{x} - \zeta_{\mat{\gamma}}r_{1}} \function{F}{j}{L}{y|\mat{z}}}\\%
			&\qquad-
			\log\sbracket{1 - \function{\varphi}{k}{U}{\tr{\mat{\gamma}}_{i}\mat{x} + \zeta_{\mat{\gamma}}r_{1}} \function{F}{j}{U}{y|\mat{z}}}
			\bigg|
			\dd Q_{Y, \mat{X}, \mat{Z}}\left(y, \mat{x}, \mat{z}\right)\\%
			&\quad+ \int_{[0, \tau_{0}] \times \family{X} \times \family{Z}}
			\bigg|
			\log\sbracket{1 - \function{\varphi}{k}{L}{\tr{\mat{\gamma}}_{i}\mat{x} - \zeta_{\mat{\gamma}}r_{1}} \function{F}{j}{L}{y|\mat{z}}}\\%
			&\qquad-
			\log\sbracket{1 - \function{\varphi}{k}{U}{\tr{\mat{\gamma}}_{i}\mat{x} - \zeta_{\mat{\gamma}}r_{1}} \function{F}{j}{L}{y|\mat{z}}}
			\bigg|
			\dd Q_{Y, \mat{X}, \mat{Z}}\left(y, \mat{x}, \mat{z}\right).\\%
		\end{split}
	\end{equation}
	For the second integral in \eref{eq:bracket_size_censored_L1_triangle} we have
	\begin{align*}
		&\int_{[0, \tau_{0}] \times \family{X} \times \family{Z}}
		\bigg|
		\log\sbracket{1 - \varphi_k^L(\tr{\mat{\gamma}}_{i}\mat{x} - \zeta_{\mat{\gamma}}r_{1}) F_j^L(y|\mat{z})}\\%
		&\qquad\qquad\qquad-
		\log\sbracket{1 - \varphi_k^U(\tr{\mat{\gamma}}_{i}\mat{x} - \zeta_{\mat{\gamma}}r_{1}) F_j^L(y|\mat{z})}
		\bigg|
		\dd Q_{Y, \mat{X}, \mat{Z}}\left(y, \mat{x}, \mat{z}\right)\\
		&\leq{
			\frac{1}{\epsilon^{\prime}}	
			\int_{[0, \tau_{0}] \times \family{X} \times \family{Z}}
			\left|
			\varphi_k^U(\tr{\mat{\gamma}}_{i}\mat{x} - \zeta_{\mat{\gamma}}r_{1})
			-
		\varphi_k^L(\tr{\mat{\gamma}}_{i}\mat{x} - \zeta_{\mat{\gamma}}r_{1})\right|
			F_j^L(y|\mat{z})
			\dd Q_{Y, \mat{X}, \mat{Z}}\left(y, \mat{x}, \mat{z}\right)}\\
		&\leq 
		\frac{1}{\epsilon^{\prime}}\int_{\real}
		\bigg|
		\function{\varphi}{k}{U}{t} -
		\function{\varphi}{k}{L}{t}
		\bigg|
		\dd Q_{i}^{-}\left(t\right)\leq
		{\frac{\zeta_{\varphi}}{\epsilon^{\prime}}},
	\end{align*}
	where we have used {the mean value theorem,  $0 \leq {F}_{j}^{L} \leq 1$ and $\epsilon^{\prime} \leq {\varphi}_{k}^{L}, {\varphi}_{k}^{U} \leq 1 - \epsilon^{\prime}$}. 
	For the integrand of the first integral in \eref{eq:bracket_size_censored_L1_triangle}, we have
	\begin{align*}
		&\bigg|
		\log\sbracket{1 - \function{\varphi}{k}{U}{\tr{\mat{\gamma}}_{i}\mat{x} - \zeta_{\mat{\gamma}}r_{1}} \function{F}{j}{L}{y|\mat{z}}}
		-
		\log\sbracket{1 - \function{\varphi}{k}{U}{\tr{\mat{\gamma}}_{i}\mat{x} + \zeta_{\mat{\gamma}}r_{1}} \function{F}{j}{U}{y|\mat{z}}}
		\bigg|\\
		&\leq{
			\frac{1}{\epsilon^{\prime}}\left|
			\function{\varphi}{k}{U}{\tr{\mat{\gamma}}_{i}\mat{x} + \zeta_{\mat{\gamma}}r_{1}} \function{F}{j}{U}{y|\mat{z}}
			-
			\function{\varphi}{k}{U}{\tr{\mat{\gamma}}_{i}\mat{x} - \zeta_{\mat{\gamma}}r_{1}} \function{F}{j}{L}{y|\mat{z}}\right|
		}\\
		&\leq
		\frac{1}{\epsilon^{\prime}}
		\function{\varphi}{k}{U}{\tr{\mat{\gamma}}_{i}\mat{x} + \zeta_{\mat{\gamma}}r_{1}}
		\left|	\function{F}{j}{U}{y|\mat{z}} - \function{F}{j}{L}{y|\mat{z}}
		\right|\\
		&\quad+ 	\frac{1}{\epsilon^{\prime}}\function{F}{j}{L}{y|\mat{z}}\left|
		\function{\varphi}{k}{U}{\tr{\mat{\gamma}}_{i}\mat{x} + \zeta_{\mat{\gamma}}r_{1}} -
		\function{\varphi}{k}{U}{\tr{\mat{\gamma}}_{i}\mat{x} - \zeta_{\mat{\gamma}}r_{1}}
		\right|\\
		&\leq
		\frac{1}{\epsilon^{\prime}}
		\left|	\function{F}{j}{U}{y|\mat{z}} - \function{F}{j}{L}{y|\mat{z}}
		\right|+ 	\frac{1}{\epsilon^{\prime}}\left|
		\function{\varphi}{k}{U}{\tr{\mat{\gamma}}_{i}\mat{x} + \zeta_{\mat{\gamma}}r_{1}} -
		\function{\varphi}{k}{U}{\tr{\mat{\gamma}}_{i}\mat{x} - \zeta_{\mat{\gamma}}r_{1}}
		\right|.
	\end{align*}
	Note that, by the monotonicity of $\varphi_k^U$, we have
	\begin{align}
		\begin{split}
			\label{eq:link_diff_bound}
			&	\int_{\family{X}}
			\left|
			\function{\varphi}{k}{U}{\tr{\mat{\gamma}}_{i}\mat{x} + \zeta_{\mat{\gamma}}r_{1}} -
			\function{\varphi}{k}{U}{\tr{\mat{\gamma}}_{i}\mat{x} - \zeta_{\mat{\gamma}}r_{1}}
			\right|\dd Q_{\mat{X}}\left(\mat{x}\right)\\
			&=
			\int_{\family{X}}
			\function{\varphi}{k}{U}{\tr{\mat{\gamma}}_{i}\mat{x} + \zeta_{\mat{\gamma}}r_{1}} -
			\function{\varphi}{k}{U}{\tr{\mat{\gamma}}_{i}\mat{x} - \zeta_{\mat{\gamma}}r_{1}}
			\dd Q_{\mat{X}}\left(\mat{x}\right)\\
			&\leq
			\bar{q}_{1}
			\int_{-r_{1}}^{r_{1}}\left[
			\function{\varphi}{k}{U}{t + \zeta_{\mat{\gamma}}r_{1}} -
			\function{\varphi}{k}{U}{t - \zeta_{\mat{\gamma}}r_{1}}\right]
			\dd t\\
			&=
			\bar{q}_{1}\rbracket{
				\int_{r_{1} - \zeta_{\mat{\gamma}}r_{1}}^{r_{1} + \zeta_{\mat{\gamma}}r_{1}}
				\function{\varphi}{k}{U}{u}
				\dd u -
				\int_{-r_{1} - \zeta_{\mat{\gamma}}r_{1}}^{-r_{1} + \zeta_{\mat{\gamma}}r_{1}}
				\function{\varphi}{k}{U}{u}
				\dd u
			}\\
			&=
			2r_{1}\bar{q}_{1}\rbracket{1 - 2\epsilon^{\prime}}\zeta_{\mat{\gamma}}.
		\end{split}
	\end{align}
	Hence, the first integral in \eref{eq:bracket_size_censored_L1_triangle} can be bounded by 
	\begin{align*}
		\frac{1}{\epsilon^{\prime}}\left[\int
		\left|
		F_j^U(y|\mat{z}) - F_j^L(y|\mat{z})
		\right|\dd \mathbb{P}_{0}\left(y, \delta,\mat{x}, \mat{z}\right)+
		2r_{1}\bar{q}_{1}\rbracket{1 - 2\epsilon^{\prime}}\zeta_{\mat{\gamma}}\right]\leq
		\frac{
				\zeta_F + 2r_{1}\bar{q}_{1}\zeta_{\mat{\gamma}}}{\epsilon^{\prime}}
		.
	\end{align*}
	Choosing $\zeta_{\varphi} = \epsilon'\zeta/3$, $\zeta_F=\epsilon'\zeta/3$, $\zeta_{\mat{\gamma}} = \epsilon'\zeta/6r_{1}\bar{q}_{1}$, we get $	\Vert g_{i,j,k}^U-g_{i,j,k}^L\Vert_{1,\mathbb{P}_{0}}\leq\zeta$.   Hence, using $\log x \leq x$, we have
	\begin{align*}
		\begin{split}
			\bracketentropy[\zeta]{\family{G}_{\epsilon^{\prime}}}{\norm{\cdot}{1,\mathbb{P}_{0}}}
			&\leq
			\log N + \log N_{3} + \log N_{4}\\%
			&\leq
			\frac{C_1\rbracket{q + 1}}{\zeta_F} +
			d\log\rbracket{\frac{A_3}{\zeta_{\mat\gamma}}} + \frac{A_4}{\zeta_\varphi}\\
			&\leq
			\frac{C_2d+C_3(q+1)}{\zeta},
		\end{split}
	\end{align*}
	for some positive constants $C_2,C_3$. 
	
	To prove the third statement of the Lemma, we can use a similar argument as for the second statement and show that the brackets $[\log\function{\varphi}{k}{L}{\tr{\mat{\gamma}}_{i}\mat{x} - \zeta_{\mat{\gamma}}r_{1}},$ $ \log\function{\varphi}{k}{U}{\tr{\mat{\gamma}}_{i}\mat{x} + \zeta_{\mat{\gamma}}r_{1}}]$ for $k\in\{1,\dots,N_4\}$, $i\in\{1,\dots,N_3\}$ are $\zeta$-brackets with respect to $\Vert\cdot\Vert_{1,\mathbb{P}_{0}}$ if we choose 
	$\zeta_{\varphi} = \epsilon'\zeta/2$ and $\zeta_{\mat{\gamma}} = \epsilon'\zeta/4r_{1}\bar{q}_{1}$.
	Hence
	\begin{align*}
		\begin{split}
			\bracketentropy[\zeta]{\family{H}_{\epsilon^{\prime}}}{\norm{\cdot}{1,\mathbb{P}_{0}}}
			\leq
			\log N_{3} + \log N_{4}
			\leq
			d\log\rbracket{\frac{A_3}{\zeta_{\mat\gamma}}} + \frac{A_4}{\zeta_\varphi}
			\leq
			\frac{C_4\rbracket{d + 1}}{\zeta}.
		\end{split}
	\end{align*}
\end{proof}

\begin{lemma}
	\label{lemma:bracket_entropy_L}
	Suppose that Assumptions \ref{enum:bounded_support} and \ref{enum:bound_dens} 	are satisfied.
	Let $\zeta > 0$. There exists a constant $A_{1} > 0$ depending on $\epsilon^{\prime}, \bar{q}_{1}$, and $r_{1}$ and a constant $A_{2} > 0$ depending on $\epsilon^{\prime}, \bar{q}_{1}, \bar{q}_{2}, r_{1}$, and $r_{2}$, such that
	\begin{equation*}
		\bracketentropy{\family{L}_{\epsilon^{\prime}}}{\norm{\cdot}{1,\mathbb{P}_{0}}} \leq \frac{A_{1}\rbracket{d + 1} + A_{2}\rbracket{q + 1}}{\zeta}.
	\end{equation*}
	
	Moreover, there exists a constant $A_{3} > 0$ depending on $\epsilon^{\prime}, \bar{q}_{1}$, and $r_{1}$ and a constant $A_{4} > 0$ depending on $\epsilon^{\prime}, \bar{q}_{1}, \bar{q}_{2}, r_{1}$, and $r_{2}$, such that
	\begin{equation*}
		\bracketentropy{\tilde{\family{L}}_{\epsilon^{\prime}}}{\norm{\cdot}{1,\mathbb{P}_{0}}} \leq \frac{A_{3}\rbracket{d + 1} + A_{4}\rbracket{q + 1}}{\zeta}.
	\end{equation*}
\end{lemma}

\begin{proof}
	Let $\zeta > 0$. Using the last two statements of Lemma \ref{lemma:bracket_entropy_F}, we can consider $\zeta/2$-brackets $[h_{i}^{L},h_{i}^{U}]$, $i = 1,\cdots,N_{1}$, covering the class $\family{H}_{\epsilon^{\prime}}$
	where $N_{1} \leq \exp\rbracket{\frac{2C_{4}\rbracket{d + 1}}{\zeta}}$. Also, consider $\zeta/2$-brackets $[g_{j}^{L},g_{j}^{U}]$, $j = 1,\cdots,N_{2}$, covering the class $\family{G}_{\epsilon^{\prime}}$, 
	where $N_{2} \leq \exp\rbracket{\frac{2C_{2}d+ 2C_{3}\rbracket{q + 1}}{\zeta}}$. Then, for any $l \in \family{L}_{\epsilon^{\prime}}$, we have
	\begin{equation*}
		\function{h}{i}{L}{\delta, \mat{x}} + \function{g}{j}{L}{y, \delta, \mat{x}, \mat{z}}
		\leq
		\function{l}{}{}{y, \delta, \mat{x}, \mat{z}}
		\leq
		\function{h}{i}{U}{\delta, \mat{x}} + \function{g}{j}{U}{y, \delta, \mat{x}, \mat{z}},
	\end{equation*}
	for some $i = 1, \cdots N_{1}$ and $j = 1, \cdots, N_{2}$. Using Minkowski inequality, we obtain
	\begin{align*}
		\begin{split}
			&\int
			\left|
			\function{h}{i}{U}{\delta, \mat{x}} + \function{g}{j}{U}{y, \delta, \mat{x}, \mat{z}} - \function{h}{i}{L}{\delta, \mat{x}} - \function{g}{j}{L}{y, \delta, \mat{x}, \mat{z}}
			\right|
			\dd \function{\mathbb{P}}{0}{}{y, \delta, \mat{x}, \mat{z}}\\
			&\leq
			\int
			\left|
			\function{h}{i}{U}{\delta, \mat{x}} -
			\function{h}{i}{L}{\delta, \mat{x}}
			\right|
			\dd\function{\mathbb{P}}{0}{}{y, \delta, \mat{x}, \mat{z}}\\
			&\quad+
			\int
			\left|
			\function{g}{j}{U}{y, \delta, \mat{x}, \mat{z}} -
			\function{g}{j}{L}{y, \delta, \mat{x}, \mat{z}}
			\right|
			\dd\function{\mathbb{P}}{0}{}{y, \delta, \mat{x}, \mat{z}}\\
			&\leq \zeta/2+\zeta/2=\zeta.
		\end{split}
	\end{align*}
	Hence, the $\zeta$-brackets $[	h_i^U+ g_j^U,h_i^L+g_j^L]$, $i=1,\dots,N_1$, $j=1,\dots,,N_2$, cover $\mathcal{L}_{\epsilon'}$,  and
	\begin{align*}
		\begin{split}
			\bracketentropy[\zeta]{\family{L}_{\epsilon^{\prime}}}{\norm{\cdot}{1,\mathbb{P}_{0}}}
			\leq
			\log N_{1} + \log N_{2}
			\leq
			\frac{A_{1} \rbracket{d + 1} + A_{2}\rbracket{q + 1}}{\zeta}.
		\end{split}
	\end{align*}
	
	The second statement of the Lemma can be shown with the same type of argument. 	
\end{proof}

\begin{lemma}
	\label{lemma:cont_monotone_cover}
	Let $Q$ be any probability measures on $\real$ and denote by $\norm{\cdot}{r,Q}$ the $L_{r}$-norm corresponding to $Q$. Let $\tilde{\family{M}}_{\epsilon^{\prime}}$ be the class of continuous monotone functions on $\real$ with values in $\sbracket{\epsilon^{\prime}, 1 - \epsilon^{\prime}}$. For any $\zeta > 0$, there exist $M \leq \exp\rbracket{\frac{A}{\zeta}}$ balls of radius of $\zeta$ with respect to the $L_{r}$-norm and centered in $\tilde{\family{M}}_{\epsilon^{\prime}}$ such that their union covers $\tilde{\family{M}}_{\epsilon^{\prime}}$. The constant  $A > 0$ does not depend on $\zeta$.
\end{lemma}
\begin{proof}[Proof of Lemma \ref{lemma:cont_monotone_cover}]
	By Lemma \ref{lemma:bracket_entropy_M} and Lemma 2.1 in \cite{van2000empirical},
	$\family{M}_{\epsilon^{\prime}}$ can be covered by $N$ balls of radius of $\zeta/2$, where $N \leq \exp\rbracket{A/\zeta}$ for some constant $A > 0$. Let $\cbracket{\varphi_{1}, \cdots, \varphi_{N}}$ be the centers of each of these balls which are not necessarily elements of $\tilde{\family{M}}_{\epsilon'}$. 	Next we construct new balls whose union covers $\tilde{\family{M}}_{\epsilon^{\prime}}$.
	For each $i \in \cbracket{1, \cdots, N}$, if there is a continuous monotone function $\tilde{\varphi}_{i}$ in the ball $B_{\varphi_{i}}(\zeta/2)$, we define a new ball, $B_{\tilde{\varphi}_{i}}({\zeta})$, centered at $\tilde{\varphi}_{i}$ with a radius of $\zeta$. Otherwise, if the ball $B_{\varphi_{i}}\rbracket{\zeta/2}$ does not contain any continuous monotone function we eliminate it. 
	Thus we have at most $M \leq N \leq \exp\rbracket{A/{\zeta}}$ of such balls with radius $\zeta$. Next we show that they cover $\tilde{\family{M}}_{\epsilon'}.$ 
	
	Consider $\varphi \in \tilde{\family{M}}_{\epsilon^{\prime}} \subset \family{M}_{\epsilon^{\prime}}$. We can find $i \in \cbracket{1, \cdots, N}$ such that $\norm{\varphi - \varphi_{i}}{r,Q} \leq \zeta/2$. 
	Let $\tilde{\varphi}_{j}$ be center of the new ball constructed using the preceding argument that corresponds to $B_{\varphi_{i}}\rbracket{\zeta/2}$. We have $\norm{{\varphi} - \tilde{\varphi}_{j}}{r,Q} \leq \norm{{\varphi} - \varphi_{i}}{r,Q} + \norm{\varphi_{i} - \tilde{\varphi}_{j}}{r,Q} \leq \zeta$. Hence ${\varphi} \in B_{\tilde{\varphi}_{j}}\rbracket{\zeta}$ for some $j \in \cbracket{1, \cdots M}$. We conclude that $\tilde{\family{M}}_{\epsilon^{\prime}} \subset \bigcup_{j=1}^{M} B_{\tilde{\varphi}_{j}}\rbracket{\zeta}$.
\end{proof}	
	\begin{proof}[Proof of Proposition \ref{prop:link_mle_convergence}]
		Recall that for fixed $\mat{\theta}$, $\varphi_{0,\mat{\theta}} = \argmax_{\varphi \in \family{M}_{\epsilon^{\prime}}}l_{\mat{\theta}}\rbracket{\varphi}$, where 
		\begin{equation*}
			l_{\mat{\theta}}\rbracket{\varphi} 
			= 
			\expect[\sbracket]{l\rbracket{Y, \Delta, \mat{X}, \mat{Z}; \mat{\theta}, \varphi}}
			=
			\int 
			l\rbracket{y, \delta, \mat{x}, \mat{z}; \mat{\theta}, \varphi}
			\dd\mathbb{P}_{0}\rbracket{y, \delta, \mat{x}, \mat{z}},
		\end{equation*}
		and $\mathbb{P}_{0}$ denotes the distribution of $\rbracket{Y, \Delta, \mat{X}, \mat{Z}}$. On the other hand, for fixed $\mat{\theta}$, we have defined $\hat{\varphi}_{n,\mat{\theta}} 
		= \argmax_{\varphi \in \family{M}_{\epsilon^{\prime}}}
		\function{l}{n}{}{\mat{\gamma},\mat{\beta},\Lambda,\varphi}$, where
		\begin{equation*}
			\function{l}{n}{}{\mat{\gamma},\mat{\beta},\Lambda,\varphi}
			=
			\int 
			l\rbracket{y, \delta, \mat{x}, \mat{z}; \mat{\theta}, \varphi}
			\dd\empirical{y, \delta, \mat{x}, \mat{z}},
		\end{equation*}
		and $\mathbb{P}_{n}$ denotes the empirical distribution of the $\rbracket{y_{i}, \delta_{i}, \mat{x}_{i}, \mat{z}_{i}}$. We first show that
		\begin{align}
			\label{eq:link_mle_l2_dist}
			\begin{split}
				&\int 
				l\rbracket{y, \delta, \mat{x}, \mat{z}; \mat{\theta}, \hat{\varphi}_{n,\mat{\theta}}} - l\rbracket{y, \delta, \mat{x}, \mat{z}; \mat{\theta}, \varphi_{0,\mat{\theta}}}
				\dd\rbracket{\mathbb{P}_{n} - \mathbb{P}_{0}}\rbracket{y, \delta, \mat{x}, \mat{z}}\\%
				&\quad\geq
				\tilde{c} \int_{\family{X}}\cbracket{
					\function{\hat{\varphi}}{n,\mat{\theta}}{}{\tr{\mat{\gamma}}\mat{x}} - \function{\varphi}{0,\mat{\theta}}{}{\tr{\mat{\gamma}}\mat{x}}
				}^{2}\dd\function{Q}{\mat{X}}{}{\mat{x}}=\tilde{c}\Vert\hat\varphi_{n,\mat\theta}-\varphi_{0,\mat{\theta}} \Vert_{Q_{\mat{X}},\mat\gamma},
			\end{split}
		\end{align}
		for some constant $\tilde{c} > 0$.
		Then it suffices to show that
		\begin{equation}
			\label{eq:llh_link_convergence}
			\sup_{\mat{\theta}\in\tilde\Theta}
			\int 
			l\rbracket{y, \delta, \mat{x}, \mat{z}; \mat{\theta}, \hat{\varphi}_{n,\mat{\theta}}} - l\rbracket{y, \delta, \mat{x}, \mat{z}; \mat{\theta}, \varphi_{0,\mat{\theta}}}
			\dd\rbracket{\mathbb{P}_{n} - \mathbb{P}_{0}}\rbracket{y, \delta, \mat{x}, \mat{z}} 
		\end{equation}
		converges almost surely to zero. 
		
		For fixed $\mat{\theta}\in\tilde\Theta$ and ${\varphi} \in \family{M}_{\epsilon^{\prime}}$, define $H\rbracket{t} = \expect[\sbracket]{l_{\mat{\theta}}\rbracket{t{\varphi} + \rbracket{1 - t}\varphi_{0,\mat{\theta}}}}$, where $0 < t < 1$. We have
		\begin{align*}
			\begin{split}
				&\int 
				l\rbracket{y, \delta, \mat{x}, \mat{z}; \mat{\theta}, \varphi} - l\rbracket{y, \delta, \mat{x}, \mat{z}; \mat{\theta}, \varphi_{0,\mat{\theta}}}
				\dd\mathbb{P}_{0}\rbracket{y, \delta, \mat{x}, \mat{z}}\\
				&=
				H\rbracket{1} - H\rbracket{0}
				=
				H'(0) +
				\frac{1}{2}H''(t^*)
			\end{split}
		\end{align*}
		for some $t^{\ast} \in \rbracket{0, 1}$. By the definition of $\varphi_{0,\mat{\theta}}$ and the convexity of $\family{M}_{\epsilon^{\prime}}$, $H'(0)\leq 0$ for all $\varphi \in \family{M}_{\epsilon^{\prime}}$. From \eqref{eq:second_deriv_expectation} and \eqref{eq:second_deriv_expectation_lb}, we have
		\begin{align*}
			\begin{split}
				H''(t)
				&\leq
				-\expect[\sbracket]{
					\link[\tr{\mat{\gamma}}_{0}\mat{X}]{0}{}F_{u,0}
					\cbracket{1 - \link[\tr{\mat{\gamma}}_{0}\mat{X}]{0}{}F_{u,0}}
					\cbracket{\link[\tr{\mat{\gamma}}\mat{X}]{}{} - \link[\tr{\mat{\gamma}}\mat{X}]{0,\mat{\theta}}{}}^{2}
				}\\%
				&\leq
				-c\epsilon^{2}\int_{\family{X}}\cbracket{
					\function{\varphi}{}{}{\tr{\mat{\gamma}}\mat{x}} - \function{\varphi}{0,\mat{\theta}}{}{\tr{\mat{\gamma}}\mat{x}}
				}^{2}\dd\function{Q}{\mat{X}}{}{\mat{x}}.
			\end{split}
		\end{align*}
		Then,
		\begin{equation*}
			\int 
			l\rbracket{y, \delta, \mat{x}, \mat{z}; \mat{\theta}, \varphi} - l\rbracket{y, \delta, \mat{x}, \mat{z}; \mat{\theta}, \varphi_{0,\mat{\theta}}}
			\dd\mathbb{P}_{0}\rbracket{y, \delta, \mat{x}, \mat{z}}\leq
			-\tilde{c}\Vert\hat\varphi_{n,\mat\theta}-\varphi_{0,\mat{\theta}} \Vert_{Q_{\mat{X}},\mat\gamma},
		\end{equation*}
		where $\tilde{c} = \frac12c\epsilon^{2} > 0$.
		By the definition of $\hat{\varphi}_{n,\mat{\theta}}$, we have
		\begin{equation*}
			\int 
			l\rbracket{y, \delta, \mat{x}, \mat{z}; \mat{\theta}, \hat{\varphi}_{n,\mat{\theta}}} - l\rbracket{y, \delta, \mat{x}, \mat{z}; \mat{\theta}, \varphi_{0,\mat{\theta}}}
			\dd\empirical{y, \delta, \mat{x}, \mat{z}}
			\geq 0.
		\end{equation*}
		Combining the two inequalities above and the fact that $\hat{\varphi}_{n,\mat{\theta}} \in \family{M}_{\epsilon^{\prime}}$ gives \eqref{eq:link_mle_l2_dist}.%

		To obtain the almost sure convergence of the expression in \eqref{eq:llh_link_convergence}, we consider the class $\tilde{\family{L}}_{\epsilon^{\prime}}$ of functions of the form 
		\begin{equation*}
			\label{eqn:class_tilde_L}
			\begin{split}
				\tilde{l}\rbracket{y, \delta, \mat{x}, \mat{z}} &= \delta\sbracket{
					\log\link[\tr{\mat{\gamma}}\mat{x}]{1}{} -
					\log\link[\tr{\mat{\gamma}}\mat{x}]{2}{}
				}\\
				&\quad + \rbracket{1 - \delta}\cbracket{
					\log\sbracket{1 - \varphi_1(\tr{\mat{\gamma}}\mat{x})F_u(y|\mat{z})} -
					\log\sbracket{1 - \varphi_2(\tr{\mat{\gamma}}\mat{x})F_u(y|\mat{z})}
				},
			\end{split}
		\end{equation*}
		for some  $\varphi_{1}, \varphi_{2} \in \family{M}_{\epsilon^{\prime}}$, $F_{u}(y|\mat{z})=1-\exp\{-\Lambda(y)\exp(\mat{\beta}^T\mat{z})\}$ and $(\mat{\gamma},\mat\beta,\Lambda)\in\tilde\Theta$. With a series of Lemmas in Appendix we show that $\tilde{\family{L}}_{\epsilon^{\prime}}$ is a Glivenko-Cantelli class of functions (see Lemma~\ref{lemma:bracket_entropy_L}). Note also that the class is uniformly bounded since for any {$\tilde{l} \in \tilde{\family{L}}_{\epsilon^{\prime}}$, $\big|\tilde{l}\big| \leq -\log{\epsilon^{\prime}} < \infty$}. As a result,
		\begin{equation*}
			\sup_{\tilde{l} \in \tilde{\family{L}}_{\epsilon^{\prime}}}
			\bigg|
			\int 
			\tilde{l}\rbracket{y, \delta, \mat{x}, \mat{z}}\dd\rbracket{\mathbb{P}_{n} - \mathbb{P}_{0}}\rbracket{y, \delta, \mat{x}, \mat{z}}
			\bigg|
			\to 0\quad{a.s. \text{ as } n\to\infty}.
		\end{equation*}
		Since $\hat{\varphi}_{n,\mat{\theta}}, \varphi_{0,\mat{\theta}} \in \family{M}_{\epsilon^{\prime}}$, we complete the proof by the following inequality,
		\begin{align*}
			\begin{split}
				\sup_{\mat{\theta}\in\tilde\Theta}
				&\int 
				l\rbracket{y, \delta, \mat{x}, \mat{z}; \mat{\theta}, \hat{\varphi}_{n,\mat{\theta}}} - l\rbracket{y, \delta, \mat{x}, \mat{z}; \mat{\theta}, \varphi_{0,\mat{\theta}}}
				\dd\rbracket{\mathbb{P}_{n} - \mathbb{P}_{0}}\rbracket{y, \delta, \mat{x}, \mat{z}}\\
				&\leq
				\sup_{\tilde{l} \in \tilde{\family{L}}_{\epsilon^{\prime}}}
				\bigg|
				\int 
				\tilde{l}\rbracket{y, \delta, \mat{x}, \mat{z}}
				\dd\rbracket{\mathbb{P}_{n} - \mathbb{P}_{0}}\rbracket{y, \delta, \mat{x}, \mat{z}}
				\bigg| 	\to 0\quad{a.s. \text{ as } n\to\infty}.
			\end{split}
		\end{align*}
	\end{proof}
	
	\begin{proof}[Proof of Proposition \ref{prop:smooth_estimator_convergence}]
		{For fixed $\mat{\theta} \in \tilde\Theta$ and $\mat{x} \in \family{X}$}, using the properties of the kernel density $k$ and a change of variable, we have
		\[
		\begin{split}
			\function{\hat{\varphi}}{n,\mat{\theta}}{s}{\tr{\mat{\gamma}}\mat{x}} - \function{\varphi}{0,\mat{\theta}}{}{\tr{\mat{\gamma}}\mat{x}}
			= &
			\int_{-1}^{1}
			\function{k}{}{}{u}\sbracket{
				\function{\hat{\varphi}}{n,\mat{\theta}}{}{\tr{\mat{\gamma}}\mat{x} - hu} - \function{\varphi}{0,\mat{\theta}}{}{\tr{\mat{\gamma}}\mat{x} - hu}
			}\dd u \\%
			&\quad+
			\int_{-1}^{1}
			\function{k}{}{}{v}\sbracket{
				\function{\varphi}{0,\mat{\theta}}{}{\tr{\mat{\gamma}}\mat{x} - hv} - 
				\function{\varphi}{0,\mat{\theta}}{}{\tr{\mat{\gamma}}\mat{x}}
			}
			\dd v.
		\end{split}
		\]
		Then, we have
		\begin{align}
			\begin{split}
				\label{eq:smooth_link_l2_triangle}
				&\cbracket{
					\int_{\family{X}}
					\sbracket{
						\function{\hat{\varphi}}{n,\mat{\theta}}{s}{\tr{\mat{\gamma}}\mat{x}} - \function{\varphi}{0,\mat{\theta}}{}{\tr{\mat{\gamma}}\mat{x}}
					}^{2}
					\dd\function{Q}{\mat{X}}{}{\mat{x}}
				}^{\frac{1}{2}}\\%
				&\leq
				\cbracket{
					\int_{\family{X}}
					\sbracket{
						\int_{-1}^{1}
						\function{k}{}{}{u}\cbracket{
							\function{\hat{\varphi}}{n,\mat{\theta}}{}{\tr{\mat{\gamma}}\mat{x} - hu} - \function{\varphi}{0,\mat{\theta}}{}{\tr{\mat{\gamma}}\mat{x} - hu}
						}\dd u
					}^{2}
					\dd\function{Q}{\mat{X}}{}{\mat{x}}
				}^{\frac{1}{2}}\\%
				&+
				\cbracket{
					\int_{\family{X}}
					\sbracket{
						\int_{-1}^{1}
						\function{k}{}{}{v}\cbracket{
							\function{\varphi}{0,\mat{\theta}}{}{\tr{\mat{\gamma}}\mat{x} - hv} - 
							\function{\varphi}{0,\mat{\theta}}{}{\tr{\mat{\gamma}}\mat{x}}
						}\dd v
					}^{2}
					\dd\function{Q}{\mat{X}}{}{\mat{x}}
				}^{\frac{1}{2}}{= \rbracket{I} + \rbracket{II}}.
			\end{split}
		\end{align}
		We proceed by showing that {$\rbracket{I}$ and $\rbracket{II}$}  converge to zero uniformly on $\tilde\Theta$ with probability one.%
		
		Consider the inner integral of {$\rbracket{I}$} in \eqref{eq:smooth_link_l2_triangle} and using Jensen's inequality gives
		\begin{align*}
			&\sbracket{
				\int_{-1}^{1}
				\function{k}{}{}{u}\cbracket{
					\function{\hat{\varphi}}{n,\mat{\theta}}{}{\tr{\mat{\gamma}}\mat{x} - hu} - \function{\varphi}{0,\mat{\theta}}{}{\tr{\mat{\gamma}}\mat{x} - hu}
				}\dd u
			}^{2}\\%
			&\leq
			K^{2}\int_{-1}^{1}
			\cbracket{
				\function{\hat{\varphi}}{n,\mat{\theta}}{}{\tr{\mat{\gamma}}\mat{x} - hu} - \function{\varphi}{0,\mat{\theta}}{}{\tr{\mat{\gamma}}\mat{x} - hu}
			}^{2}\dd u,
		\end{align*}
		where $K$ is the upper bound of the kernel function $k$. By assumption (A1)(v) and for $\mat\gamma$ in a neighborhood of $\mat\gamma_0$ we have $\family{I}_{\mat\gamma}=[\ubar{I}_{\mat\gamma},\bar{I}_{\mat\gamma}]$.
		Then, we can write
		\begin{align}
			\begin{split}
				\label{eq:smooth_link_l2_triangle_1st}
				\sup_{\mat{\theta}\in\tilde\Theta}\,(I)^2
				&\leq
				K^{2}\sup_{\mat{\theta}\in\tilde\Theta}\int_{\family{X}}
				\int_{-1}^{1}
				\cbracket{
					\function{\hat{\varphi}}{n,\mat{\theta}}{}{\tr{\mat{\gamma}}\mat{x} - hu} - \function{\varphi}{0,\mat{\theta}}{}{\tr{\mat{\gamma}}\mat{x} - hu}
				}^{2}
				\dd u\,
				\dd\function{Q}{\mat{X}}{}{\mat{x}}\\%
				&=
				K^{2}\sup_{\mat{\theta}\in\tilde\Theta}\int_{\family{I}_{\mat{\gamma}}}
				\int_{-1}^{1}
				\cbracket{
					\function{\hat{\varphi}}{n,\mat{\theta}}{}{v - hu} - \function{\varphi}{0,\mat{\theta}}{}{v - hu}
				}^{2}
				\dd u\,
				\dd\function{G}{\tr{\mat{\gamma}}\mat{X}}{}{v}\\%
				&=
				K^{2}\sup_{\mat{\theta}\in\tilde\Theta}\int_{\family{I}_{\mat{\gamma}}}\frac{1}{h}
				\int_{v - h}^{v + h}
				\cbracket{
					\function{\hat{\varphi}}{n,\mat{\theta}}{}{s} - \function{\varphi}{0,\mat{\theta}}{}{s}
				}^{2}
				\dd s\,
				\dd\function{G}{\tr{\mat{\gamma}}\mat{X}}{}{v}\\%
				&=
				K^{2}\sup_{\mat{\theta}\in\tilde\Theta}\int_{\ubar{I}_{\mat\gamma} - h}^{\bar{I}_{\mat\gamma} + h}\cbracket{
					\function{\hat{\varphi}}{n,\mat{\theta}}{}{s} - \function{\varphi}{0,\mat{\theta}}{}{s}
				}^{2}\frac{1}{h}	\int_{s - h}^{s + h}
				\dd\function{G}{\tr{\mat{\gamma}}\mat{X}}{}{v}\,
				\dd s,
			\end{split}
		\end{align}
		where $G_{\tr{\mat{\gamma}}\mat{X}}$ denotes the distribution function of $\tr{\mat{\gamma}}\mat{X}$. 
		Recall that $h = h_{n} \converge 0$ as $n \converge \infty$.  For $s\in[\ubar{I}_{\mat\gamma}+h,\bar{I}_{\mat\gamma}-h]$, since the density function $g_{\tr{\mat{\gamma}}\mat{X}}$ is continuous, by the mean value theorem, there exists $\xi_{n} \in  \rbracket{s - h_{n}, s + h_{n}}$ such that
		\begin{equation}
			\label{eq:smooth_link_mvt}
			\frac{1}{h}	\int_{s - h}^{s + h}
			\dd\function{G}{\tr{\mat{\gamma}}\mat{X}}{}{v}=	\frac{1}{h}\left\{\function{G}{\tr{\mat{\gamma}}\mat{X}}{}{s+h}-\function{G}{\tr{\mat{\gamma}}\mat{X}}{}{s-h}\right\}
			= 2\function{g}{\tr{\mat{\gamma}}\mat{X}}{}{\xi_{n}}.
		\end{equation}	
		By Assumption (A9) it follows that 
		\begin{equation}
			\label{eq:index_density_convergence}
			\sup_{\mat{\theta}\in\tilde\Theta}\sup_{s\in[\ubar{I}_{\mat\gamma}+h,\bar{I}_{\mat\gamma}-h] }\left|
			\frac{1}{h}	\int_{s - h}^{s + h}
			\dd\function{G}{\tr{\mat{\gamma}}\mat{X}}{}{v}- 2\function{g}{\tr{\mat{\gamma}}\mat{X}}{}{s}
			\right| \converge 0\quad \text{ as } n\to\infty.
		\end{equation}
		Since the length $\family{I}_{\mat\gamma}$ is uniformly bounded, $\lbrace{
			\function{\hat{\varphi}}{n,\mat{\theta}}{}{s} - \function{\varphi}{0,\mat{\theta}}{}{s}
		}\rbrace^{2} \leq (1 - 2\epsilon^{\prime})^{2}$, from  \eqref{eq:smooth_link_l2_triangle_1st}, \eqref{eq:index_density_convergence} and Proposition \ref{prop:link_mle_convergence} we obtain
		\begin{equation}
			\label{eqn:sup_I}
			\begin{aligned}
			&	\sup_{\mat{\theta}\in\tilde\Theta}\,(I)^2\leq \tilde{K}h+	2K^2\sup_{\mat{\theta}\in\tilde\Theta}\int_{\family{I}_{\mat\gamma}}
				\cbracket{
					\function{\hat{\varphi}}{n,\mat{\theta}}{}{s} - \function{\varphi}{0,\mat{\theta}}{}{s}
				}^{2}
				\dd\function{G}{\tr{\mat{\gamma}}\mat{X}}{}{s}\\%
				&+K^2
				\sup_{\mat{\theta}\in\tilde\Theta}\int_{\ubar{I}_{\mat\gamma} + h}^{\bar{I}_{\mat\gamma} -h}
				\cbracket{
					\hat{\varphi}_{n,\mat{\theta}}(s) - \varphi_{0,\mat{\theta}}(s)
				}^{2}\left|
				\int_{s - h}^{s + h}\frac{1}{h}
				\dd G_{\tr{\mat{\gamma}}\mat{X}}(v)-2g_{\tr{\mat{\gamma}}\mat{X}}(s)\right|\,
				\dd s\converge[a.s.] 0,
			\end{aligned}
		\end{equation}
		where the term $\tilde{K}h$ comes from dealing with the boundary regions of the integrals and the constant $\tilde{K}$ depends on the uniform bound of $\function{g}{\tr{\mat{\gamma}}\mat{X}}{}{\cdot}$. 
		This means that
		the first term $(I)$ in \eqref{eq:smooth_link_l2_triangle} converges to zero uniformly on $\tilde\Theta$ with probability one. 
		
		{Next, we deal with $\rbracket{II}$ in \eqref{eq:smooth_link_l2_triangle}, which is a deterministic term.} With a similar argument as above, one can show that 
		\begin{align}
			\begin{split}
				\label{eq:smooth_link_l2_triangle_2nd}
				\sup_{\mat{\theta}\in\tilde\Theta}\,(II)^2
				&\leq
				K^{2}\sup_{\mat{\theta}\in\tilde\Theta}\int_{\family{X}}
				\int_{-1}^{1}
				\cbracket{
					\function{\varphi}{0,\mat{\theta}}{}{\tr{\mat{\gamma}}\mat{x} - hv} - \function{\varphi}{0,\mat{\theta}}{}{\tr{\mat{\gamma}}\mat{x}}
				}^{2}
				\dd v\,
				\dd\function{Q}{\mat{X}}{}{\mat{x}}.		
			\end{split}
		\end{align}
		For a fixed $\mat{\theta} \in \tilde\Theta$, the inner integral in the previous display converges to zero since $\varphi_{0,\mat{\theta}}$ is continuous and $h \converge 0$. However, this does not hold uniformly on $\tilde\Theta$. To circumvent the non-uniformity issue, we construct a covering of the space of bounded continuous monotone functions by a finite number of balls and approximate $\varphi_{0,\mat{\theta}}$ by one of the centers of such balls.  Let $\zeta > 0$ and let $[{\tilde{A}, \tilde{B}}]$ be an interval that contains all the intervals $[{\ubar{I}_{\mat\gamma} - h, \bar{I}_{\gamma} + h}]$. Denote by $Q
		$ the uniform distribution on $[{\tilde{A}, \tilde{B}}]$. By Lemma B.4, 
		where $M \leq \exp\sbracket{A/r}$ with a constant $A > 0$.  For a fixed $\mat{\theta} \in \tilde\Theta$, using the Minkowski inequality, we have
		\begin{align}
			\begin{split}
				\label{eq:smooth_link_l2_triangle_2nd_split}
				&\cbracket{
					\int_{\family{X}}
					\int_{-1}^{1}
					\cbracket{
						\function{\varphi}{0,\mat{\theta}}{}{\tr{\mat{\gamma}}\mat{x} - hv} - \function{\varphi}{0,\mat{\theta}}{}{\tr{\mat{\gamma}}\mat{x}}
					}^{2}
					\dd v\,
					\dd\function{Q}{\mat{X}}{}{\mat{x}}
				}^{\frac{1}{2}}	\\%
				&\leq
				\cbracket{
					\int_{\family{X}}
					\int_{-1}^{1}
					\cbracket{
						\function{\varphi}{0,\mat{\theta}}{}{\tr{\mat{\gamma}}\mat{x} - hv} - \function{\tilde{\varphi}}{i(\mat{\theta})}{}{\tr{\mat{\gamma}}\mat{x} - hv}
					}^{2}
					\dd v\,
					\dd\function{Q}{\mat{X}}{}{\mat{x}}
				}^{\frac{1}{2}}\\%
				&\quad +
				\cbracket{
					\int_{\family{X}}
					\int_{-1}^{1}
					\cbracket{
						\function{\tilde{\varphi}}{i(\mat{\theta})}{}{\tr{\mat{\gamma}}\mat{x} - hv} - \function{\tilde{\varphi}}{i(\mat{\theta})}{}{\tr{\mat{\gamma}}\mat{x}}
					}^{2}
					\dd v\,
					\dd\function{Q}{\mat{X}}{}{\mat{x}}
				}^{\frac{1}{2}}\\%
				&\quad +
				\cbracket{
					2	\int_{\family{X}}
					\cbracket{
						\function{\tilde{\varphi}}{i(\mat{\theta})}{}{\tr{\mat{\gamma}}\mat{x}} -
						\function{\varphi}{0,\mat{\theta}}{}{\tr{\mat{\gamma}}\mat{x}}
					}^{2}
					\dd\function{Q}{\mat{X}}{}{\mat{x}}
				}^{\frac{1}{2}}{= \rbracket{III} + \rbracket{IV} + \rbracket{V}
				}
			\end{split}
		\end{align}
		We will show that there exists $N$ independent of $\mat\theta$ such that each of these three terms are smaller than $\zeta/3$ for $n>N$ and any $\mat\theta\in\tilde\Theta$. Since this holds for any $\zeta>0$, we can then conclude that $(II)$ converges to zero uniformly in $\mat\theta$. 
		Consider the square of {$\rbracket{III}$} in \eqref{eq:smooth_link_l2_triangle_2nd_split}. As in  \eqref{eq:smooth_link_l2_triangle_1st}-\eqref{eqn:sup_I}, we obtain
		\begin{align*}
			\begin{split}
				&\int_{\family{X}}
				\int_{-1}^{1}
				\cbracket{
					\function{{\varphi}}{0,\mat{\theta}}{}{\tr{\mat{\gamma}}\mat{x} - hv} - \function{\tilde{\varphi}}{i(\mat{\theta})}{}{\tr{\mat{\gamma}}\mat{x} - hv}
				}^{2}
				\dd v\,
				\dd\function{Q}{\mat{X}}{}{\mat{x}}\\%
				&\leq c\left\{h+\int_{I_{\mat\gamma}}
				\cbracket{
					\function{{\varphi}}{0,\mat{\theta}}{}{s} - \function{\tilde{\varphi}}{i(\mat{\theta})}{}{s}
				}^{2}\dd\function{G}{\tr{\mat{\gamma}}\mat{X}}{}{s}\right.\\
				&\quad\left.+	\sup_{s\in[\ubar{I}_{\mat\gamma}+h,\bar{I}_{\mat\gamma}-h] }\left|
				\frac{1}{h}	\int_{s - h}^{s + h}
				\dd\function{G}{\tr{\mat{\gamma}}\mat{X}}{}{v}- 2\function{g}{\tr{\mat{\gamma}}\mat{X}}{}{s}
				\right| 			\right\}\\
				&\leq c\left(h+\bar{q}_1(\tilde{B}-\tilde{A})r^2+2\sup_{s\in[\ubar{I}_{\mat\gamma}+h,\bar{I}_{\mat\gamma}-h] }\left|\function{g}{\tr{\mat{\gamma}}\mat{X}}{}{\xi_{n,s}}-\function{g}{\tr{\mat{\gamma}}\mat{X}}{}{s}\right|\right),
			\end{split}
		\end{align*}
		where $s-h\leq \xi_{n,s}\leq s+h$, $\bar{q}_1$ denotes the uniform upper bound on the density functions $g_{\mat\gamma^T\mat{X}}(\cdot)$, and $c$ is a positive constant. By uniform equicontinuity of the family of functions $\{g_{\tr{\mat{\gamma}}\mat{X}}\}$ it follows that there exists $N_1$ such that,  $$\sup_{s\in[\ubar{I}_{\mat\gamma}+h,\bar{I}_{\mat\gamma}-h] }\left|\function{g}{\tr{\mat{\gamma}}\mat{X}}{}{\xi_{n,s}}-\function{g}{\tr{\mat{\gamma}}\mat{X}}{}{s}\right|\leq \frac{\zeta^2}{54c}$$ 
		for all $\mat\theta\in\tilde\Theta$, $h\leq \zeta^2/27c$ whenever $n\geq N_1$. If we  choose $r={\zeta/\sqrt{27c\bar{q}_1(\tilde{B}-\tilde{A})}}$ we obtain $(III)^2\leq \zeta^2/9$ for all $\mat\theta\in\tilde\Theta$ and $n>N_1$.
		In the same way, for the square of {$\rbracket{V}$} in \eqref{eq:smooth_link_l2_triangle_2nd_split} we  obtain, for $n>N_2$,
		\begin{equation}
			\label{eq:smooth_link_l2_triangle_2nd_splited_3rd}
			(V)^2 {\leq} 2\sup_{\mat{\theta}\in\tilde\Theta}
			\int_{\family{X}}
			\cbracket{
				\function{\tilde{\varphi}}{i(\mat{\theta})}{}{\tr{\mat{\gamma}}\mat{x}} -
				\function{\varphi}{0,\mat{\theta}}{}{\tr{\mat{\gamma}}\mat{x}}
			}^{2}
			\dd\function{Q}{\mat{X}}{}{\mat{x}}
			\leq
			{\frac{\zeta^2}{9}}.
		\end{equation}
		For the square of {\rbracket{IV}} in \eqref{eq:smooth_link_l2_triangle_2nd_split}, we have
		\begin{align}
			\begin{split}
				\label{eq:smooth_link_l2_triangle_2nd_splited_2nd}
				&	
				\int_{\family{X}}
				\int_{-1}^{1}
				\cbracket{
					\function{\tilde{\varphi}}{i(\mat{\theta})}{}{\tr{\mat{\gamma}}\mat{x} - hv} - \function{\tilde{\varphi}}{i(\mat{\theta})}{}{\tr{\mat{\gamma}}\mat{x}}
				}^{2}
				\dd v\,
				\dd\function{Q}{\mat{X}}{}{\mat{x}}\\%
				&\leq
				\bar{q}_{1}
				\int_{\tilde{A}}^{\tilde{B}}
				\int_{-1}^{1}
				\cbracket{
					\function{\tilde{\varphi}}{i(\mat{\theta})}{}{u - hv} - \function{\tilde{\varphi}}{i(\mat{\theta})}{}{u}
				}^{2}
				\dd v\,\dd u.\\%
			\end{split}
		\end{align}
		Since the functions	  $\tilde{\varphi}_{i(\mat{\theta})}$ are continuous on a compact, they are uniformly continuous and since we are dealing with a finite number $M$ of such functions the family $\{\tilde{\varphi}_{i(\mat{\theta})}: i=1,\dots,M\}$ is uniformly equicontinuous. Therefore, since 
		$h = h_{n} \converge 0$ as $n \converge \infty$, we can find an integer $N_3$ such that $$|\function{\tilde{\varphi}}{i(\mat{\theta})}{}{u - hv} - \function{\tilde{\varphi}}{i(\mat{\theta})}{}{u}| \leq \zeta /\sqrt{ 18\bar{q}_{1}({\tilde{B} - \tilde{A}})}$$ for all $u\in[\tilde{A},\tilde{B}]$, $v\in[-1,1]$ and $n \geq N_3$. Consequently,  $(IV)^2\leq \zeta^2/9$ for $n \geq N_3$ and any $\mat\theta\in\tilde\Theta$. 
		Combining these results, we have for $n \geq \max\{N_1,N_2,N_3\}$,
		\begin{equation*}
			\sup_{\mat{\theta}\in\tilde\Theta}\int_{\family{X}}
			\int_{-1}^{1}
			\cbracket{
				\function{\varphi}{0,\mat{\theta}}{}{\tr{\mat{\gamma}}\mat{x} - hv} - \function{\varphi}{0,\mat{\theta}}{}{\tr{\mat{\gamma}}\mat{x}}
			}^{2}
			\dd v\,
			\dd\function{Q}{\mat{X}}{}{\mat{x}} \leq \zeta^2,
		\end{equation*}
		from which the uniform almost sure convergence of $(II)$ follows. 
	\end{proof}
	
	\begin{proof}[Proof of Theorem \ref{prop:theta_estimator_consistency}]
		This can be proved by applying Theorem 1 of \cite{DK2020}, which boils down to verifying the required conditions. First we define some notation. 
		We equip the space $\tilde\Theta= \Gamma \times \mathcal{B} \times \tilde{\family{D}}$ with the metric  
		$d$ which is defined by
		\begin{equation*}
			d\rbracket{\mat{\theta}_{1}, \mat{\theta}_{2}} = \max\cbracket{\norm{\mat{\gamma}_{1} - \mat{\gamma}_{2}}{2}, \norm{\mat{\beta}_{1} - \mat{\beta}_{2}}{2}, \sup_{t \in \sbracket{0, \tau_{0}}}|\Lambda_{1}\rbracket{t} - \Lambda_{2}\rbracket{t}|},
		\end{equation*}
		where $\mat{\theta}_{j} = \rbracket{\mat{\gamma}_{j}, \mat{\beta}_{j}, \Lambda_{j}} \in \tilde\Theta$, $j = 1, 2$.
		Let $\family{H}$ denote the infinite-dimensional parameter space for the link,  defined by
		$
		\family{H} = \{\map{h}{\family{X} \times \tilde\Theta}{\sbracket{\epsilon^{\prime}, 1-\epsilon^{\prime}}},\,h\rbracket{\mat{x}, \mat{\theta}} = {\varphi}_{\mat{\theta}}({\tr{\mat{\gamma}}\mat{x}})\,\text{for some}~\varphi_{\mat{\theta}}\in\family{M}_{\epsilon^{\prime}}\}. 
		$
		We equip $\family{H}$ with the metric
		\begin{equation*}
			\function{d}{\family{H}}{}{h_{1}, h_{2}} = \sup_{\mat{\theta} \in \tilde\Theta}\rbracket{
				\int_{\family{X}}\cbracket{
					\function{\varphi}{1,\mat{\theta}}{}{\tr{\mat{\gamma}}\mat{x}} - 
					\function{\varphi}{2,\mat{\theta}}{}{\tr{\mat{\gamma}}\mat{x}}
				}^{2}\dd\function{Q}{\mat{X}}{}{\mat{x}}
			}^{\frac{1}{2}},
		\end{equation*}
		where $\function{h}{j}{}{\mat{x}, \mat{\theta}} = \function{\varphi}{j, \mat{\theta}}{}{\tr{\mat{\gamma}}\mat{x}}$ for some $\varphi_{j,\mat{\theta}} \in \family{M}_{\epsilon^{\prime}}, j=1,2$. Let $\function{h}{0}{}{\mat{x}, \mat{\theta}} = \function{\varphi}{0, \mat{\theta}}{}{\tr{\mat{\gamma}}\mat{x}}$, where $\varphi_{0,\mat{\theta}}$ is defined in \eqref{eqn:phi_theta}. Note that, because of the Kullback-Leibler inequality, 
		\[
		\expect[\sbracket]{l\rbracket{Y, \Delta, \mat{X}, \mat{Z}; \mat{\theta}, \varphi}} \leq \expect[\sbracket]{l\rbracket{Y, \Delta, \mat{X}, \mat{Z}; \mat{\theta}_{0}, \varphi_{0}}}
		\]
		for all $\mat{\theta} \in \tilde\Theta$, $\varphi \in \family{M}_{\epsilon^{\prime}}$ and, under the identifiability conditions, the equality holds only for $\mat{\theta} = \mat{\theta}_{0}$ and $\varphi = \varphi_{0}$.  By Proposition \ref{prop:true_link} it follows that $
		\mat{\theta}_{0} = \argmax_{\mat{\theta} \in \tilde\Theta}\expect[\sbracket]{l\rbracket{Y, \Delta, \mat{X}, \mat{Z}; \mat{\theta}, \varphi_{0,\mat{\theta}}}}.
		$ Then, the non-parametric estimator of $h_{0}$ is $\hat{h}(\mat{x}, \mat{\theta}) = \function{\hat{\varphi}}{n,\mat{\theta}}{s}{\tr{\mat{\gamma}}\mat{x}}$, where $\hat{\varphi}_{n,\mat{\theta}}^{s}$ is our smoothed monotone link estimator defined in \eqref{eq:smooth_isotonic_estimator}.
		Consider i.i.d. realization $\rbracket{y_{i},\delta_{i},\mat{x}_{i},\mat{z}_{i}}$, $i=1,\cdots,n$ of $\rbracket{Y, \Delta, \mat{X}, \mat{Z}}$. Let $h \in \family{H}$ such that $\function{h}{}{}{\mat{x}, \mat{\theta}} = \function{\varphi}{\mat{\theta}}{}{\tr{\mat{\gamma}}\mat{x}}$ for some $\varphi_{\mat{\theta}} \in \family{M}_{\epsilon^{\prime}}$. We define \[
		\function{M}{}{}{\mat{\theta}, h} = \expect[\sbracket]{l\rbracket{Y, \Delta, \mat{X}, \mat{Z}; \mat{\theta}, \varphi_{\mat{\theta}}}}\quad\text{and}\quad\function{M}{n}{}{\mat{\theta}, h} = \frac{1}{n}\sum_{i=1}^{n}l\rbracket{Y, \Delta, \mat{X}, \mat{Z}; \mat{\theta}, \varphi_{\mat{\theta}}}.
		\]
		   Next we verify the required conditions (A1)-(A5) of Theorem 1 of \cite{DK2020}. 
		
		Condition (A1)  is satisfied by definition of $\hat{\mat{\theta}}_{n}$.
		For condition (A2), note that $\function{M}{}{}{\mat{\theta}, h_{0}} = \expect[\sbracket]{l\rbracket{Y, \Delta, \mat{X}, \mat{Z}; \mat{\theta}, \varphi_{0,\mat{\theta}}}}$. Under the model identifiability assumptions, the negative Kullback-Leibler divergence $\function{M}{}{}{\mat{\theta}_{0}, h_{0}} - \function{M}{}{}{\mat{\theta}, h_{0}}$ attains its maximum uniquely at $\mat{\theta}_{0}$. Since $\function{M}{}{}{\mat{\theta}, h_{0}} $ is continuous with respect to $\mat\theta$, on each compact set $\tilde\Theta\cap\{\mat\theta:d(\mat\theta,\mat\theta_0)\geq \delta\}$, it will obtain a maximum strictly smaller than $\function{M}{}{}{\mat{\theta}_{0}, h_{0}}$, which indicates that (A2) is fulfilled.
		Condition (A3) follows immediately from Proposition \ref{prop:smooth_estimator_convergence} and the fact that $\hat{\varphi}^s_{n,\mat{\theta}} \in \family{M}_{\epsilon^{\prime}}$ for all $\mat{\theta} \in \tilde\Theta$.
		Using Remark (ii) of \cite{DK2020},	Condition (A4) follows from the fact that the family of uniformly bounded functions
		\[
		\begin{split}
			\family{L}_{\epsilon^{\prime}}&=\left\{l\rbracket{y, \delta, \mat{x}, \mat{z}} = \rbracket{1 - \delta}\log\sbracket{1 - \varphi(\tr{\mat{\gamma}}\mat{x})\left\{ 1 - \exp\rbracket{-\function{\Lambda}{}{}{y}e^{\tr{\mat{\beta}}\mat{z}}}\right\}}\right.\\
			&\qquad\qquad\qquad\qquad+ \delta\log\link[\tr{\mat{\gamma}}\mat{x}]{}{};
			(\mat\gamma,\mat\beta,\Lambda)\in\tilde\Theta, \, \varphi\in\family{M}_{\epsilon'}\bigg\}
		\end{split}
		\]
		is Glivenko-Cantelli (see Lemma~\ref{lemma:bracket_entropy_L}). 
		For condition (A5), note that for fixed $\mat{\theta} \in \tilde\Theta$, we have, from the proof of Proposition \ref{prop:link_maximizer}, that
		\begin{equation*}
			|\function{M}{}{}{\mat{\theta}, h} - \function{M}{}{}{\mat{\theta}, h_{0}}| \leq \frac{2}{\epsilon^{\prime}}\rbracket{
				\int_{\family{X}}\cbracket{
					\function{\varphi}{\mat{\theta}}{}{\tr{\mat{\gamma}}\mat{x}} - 
					\function{\varphi}{0,\mat{\theta}}{}{\tr{\mat{\gamma}}\mat{x}}
				}^{2}\dd\function{Q}{\mat{X}}{}{\mat{x}}
			}^{\frac{1}{2}},
		\end{equation*}
		where $\function{h}{}{}{\mat{x}, \mat{\theta}} = \function{\varphi}{\mat{\theta}}{}{\tr{\mat{\gamma}}\mat{x}}$ for some $\varphi_{\mat{\theta}} \in \family{M}_{\epsilon^{\prime}}$. Then it follows  immediately from the definition of $d_{\family{H}}$ that  
		$ \sup_{\mat{\theta} \in \Theta} |\function{M}{}{}{\mat{\theta}, h} - \function{M}{}{}{\mat{\theta}, h_{0}}|$ converges to zero as $d_{\family{H}}\rbracket{h, h_{0}} \converge 0$.
	\end{proof}

	{
		\begin{proof}[Proof of Corollary \ref{prop:link_s_u_consistency}]
			We first show the convergence of link estimator $\hat{p}_{n}$. Using Minkowski inequality and $\varphi_{0,{\mat{\theta}}_{0}} = \varphi_{0}$, we have
			\begin{align*}
				\begin{split}
					&\cbracket{
						\int_{\family{X}}\cbracket{
							{\hat{\varphi}}_{n,\hat{\mat{\theta}}_{n}}^{s}({\tr{\hat{\mat{\gamma}}}_{n}\mat{x}}) - 	\link[\tr{\mat{\gamma}}_{0}\mat{x}]{0}{}
						}^{2}\dd\function{Q}{\mat{X}}{}{\mat{x}}
					}^{\frac{1}{2}}\\
					&\leq
					\cbracket{
						\int_{\family{X}}\cbracket{
							{\hat{\varphi}}_{n,\hat{\mat{\theta}}_{n}}^{s}({\tr{\hat{\mat{\gamma}}}_{n}\mat{x}}) - 	{{\varphi}}_{0,\hat{\mat{\theta}}_{n}}({\tr{\hat{\mat{\gamma}}}_{n}\mat{x}})
						}^{2}\dd\function{Q}{\mat{X}}{}{\mat{x}}
					}^{\frac{1}{2}}\\%
					&\quad+
					\cbracket{
						\int_{\family{X}}\cbracket{
							{{\varphi}}_{0,\hat{\mat{\theta}}_{n}}({\tr{\hat{\mat{\gamma}}}_{n}\mat{x}}) - 	
							{{\varphi}}_{0,{\mat{\theta}}_{0}}({\tr{\hat{\mat{\gamma}}}_{n}\mat{x}})
						}^{2}\dd\function{Q}{\mat{X}}{}{\mat{x}}
					}^{\frac{1}{2}}\\%
					&\quad+
					\cbracket{
						\int_{\family{X}}\cbracket{
							{{\varphi}}_{0}({\tr{\hat{\mat{\gamma}}}_{n}\mat{x}}) - 	\link[\tr{\mat{\gamma}}_{0}\mat{x}]{0}{}
						}^{2}\dd\function{Q}{\mat{X}}{}{\mat{x}}
					}^{\frac{1}{2}}\\%
					& = (I) + (II) + (III).
				\end{split}
			\end{align*}
			
			Consider the square of $(I)$, using the result of Proposition \ref{prop:smooth_estimator_convergence}, we have
			\[
				(I)^{2} \leq \sup_{\mat{\theta}\in\tilde\Theta}
				\int_{\family{X}}\cbracket{
					\function{\hat{\varphi}}{n,\mat{\theta}}{s}{\tr{\mat{\gamma}}\mat{x}} - 	\function{\varphi}{0,\mat{\theta}}{}{\tr{\mat{\gamma}}\mat{x}}
				}^{2}\dd\function{Q}{\mat{X}}{}{\mat{x}} \converge[P] 0\quad\text{as}~n\to\infty.
			\]
			
			For the square of $(III)$, by the assumption that $\varphi'_{0}$ is uniformly bounded on $\family{I}_{0}$ and using the mean value theorem followed by the Cauchy-Schwarz inequality, we have, for every $\mat{x}\in\family{X}$,
			\[
				\vert
					\varphi_{0}(\tr{\hat{\mat{\gamma}}}_{n}\mat{x}) - \varphi_{0}(\tr{{\mat{\gamma}}}_{0}\mat{x})
				\vert
				\leq c_{1}r_{1}\norm{\hat{\mat{\gamma}}_{n} - \mat{\gamma}_{0}}{2},
			\]
			where $c_{1}$ a positive constant that does not depend on $\mat{x}$. Therefore, using the result of Theorem \ref{prop:theta_estimator_consistency}, we have
			\begin{align*}
				\begin{split}
					(III)^{2}&=
					\int_{\family{X}}\cbracket{
						{{\varphi}}_{0}({\tr{\hat{\mat{\gamma}}}_{n}\mat{x}}) - \link[\tr{\mat{\gamma}}_{0}\mat{x}]{0}{}
					}^{2}\dd\function{Q}{\mat{X}}{}{\mat{x}}\leq
					c_{1}^{2}r_{1}^{2}\norm{\hat{\mat{\gamma}}_{n} - \mat{\gamma}_{0}}{2}^{2}
					\converge[P]0.
				\end{split}
			\end{align*}
			
			For the square of $(II)$, we apply the argmax continuous mapping theorem \cite[Theorem 3.2.2]{van2000empirical} to show $\sup_{x\in\real}\vert{\varphi}_{0,\hat{\mat{\theta}}_{n}}(x) - {{\varphi}}_{0,{\mat{\theta}}_{0}}(x)\vert \converge[P] 0$ if $\hat{\mat{\theta}}_{n}\converge[P]{\mat{\theta}}_{0}$. Given this uniform convergence, the convergence of $(II)$ to zero in probability follows immediately. Let
			\begin{equation*}
				m\rbracket{y, \delta, \mat{x}, \mat{z}; \mat{\theta}, \varphi} = \delta\log\link{}{} + \rbracket{1 - \delta}\log\sbracket{1 - \link{}{}\uncuredsurvival[F]{\mat{z}}{y}{}}.
			\end{equation*}
			Define $\mathbb{M}_{n}(\varphi)=\mathbb{E}[{m({Y, \Delta, \mat{X}, \mat{Z}; \hat{\mat{\theta}}_{n}, \varphi})}]$ and $\mathbb{M}(\varphi)=\mathbb{E}[{m({Y, \Delta, \mat{X}, \mat{Z}; {\mat{\theta}}_{0}, \varphi})}]$, where $\varphi$ belongs to the metric space $H=(\family{M}_{\epsilon^\prime}, \norm{\cdot}{\infty})$ with $\norm{\varphi}{\infty}=\sup_{x\in\real}\vert\varphi(x)\vert$. By the definition of $\varphi_{0,\mat{\theta}}$ in \eqref{eqn:phi_theta}, we have $\varphi_{0,\hat{\mat{\theta}}_{n}}=\argmax_{\varphi \in \family{M}_{\epsilon^{\prime}}}\mathbb{M}_{n}(\varphi)$ and $\varphi_{0,{\mat{\theta}}_{0}}=\argmax_{\varphi \in \family{M}_{\epsilon^{\prime}}}\mathbb{M}(\varphi)$. We note that the continuity of the sample paths $\varphi \mapsto \mathbb{M}(\varphi)$ on $H$ follows from a similar argument in the proof of Proposition \ref{prop:link_maximizer}. $\varphi_{0,{\mat{\theta}}_{0}}$ being tight and being the unique maximizer of $\mathbb{M}$, $\varphi_{0,\hat{\mat{\theta}}_{n}}$ being uniformly tight and maximizing $\mathbb{M}_{n}$ follow from their definitions. Therefore, it remains to show $\mathbb{M}_{n}$ converges weakly to $\mathbb{M}$ in $\ell^{\infty}(K)$ for every compact $K \subset H$. By Theorem \ref{prop:theta_estimator_consistency} and since every $\varphi\in\family{M}_{\epsilon^\prime}$ are uniformly bounded, we can show a stronger statement that $\sup_{\varphi\in\family{M}_{\epsilon^{\prime}}}\vert\mathbb{M}_{n}(\varphi)-\mathbb{M}(\varphi)\vert\converge[P]0$ and hence the weak convergence condition holds.			
			For a fixed $\varphi \in \family{M}_{\epsilon^\prime}$, we have
			\begin{align*}
				\begin{split}
					&\vert\mathbb{M}_{n}(\varphi) - \mathbb{M}(\varphi)\vert\\%
					&\leq
					\left\vert\expect[\sbracket]{
						\Delta\log\frac{
							\varphi(\tr{\hat{\mat{\gamma}}}_{n}\mat{X})
						}{
							\varphi(\tr{\mat{\gamma}}_{0}\mat{X})
						}
					}\right\vert
					+
					\left\vert\expect[\sbracket]{
						(1 - \Delta)\log\frac{
							1 - \varphi(\tr{\hat{\mat{\gamma}}}_{n}\mat{X})F_{u}(Y|\mat{Z};\hat{\mat{\beta}}_{n},\hat{\Lambda}_{n})
						}{
							1 - \varphi(\tr{\mat{\gamma}}_{0}\mat{X})F_{u}(Y|\mat{Z};\mat{\beta}_{0},\Lambda_{0})
						}
					}\right\vert\\%
					&= (IV) + (V),
				\end{split}
			\end{align*}
			where $F_{u}(t|\mat{z};\mat{\beta},\Lambda) = 1 - \exp[{-\Lambda(t)\exp({\tr{\mat{\beta}}\mat{z}}})]$.
						For $(IV)$, using a similar argument as in the proof of Proposition \ref{prop:link_maximizer}, we have
			\begin{equation*}
				(IV)\leq\frac{1}{\epsilon^{\prime}}\left(\expect[\sbracket]{
					\left\lbrace
						\varphi(\tr{\hat{\mat{\gamma}}}_{n}\mat{X}) - 
						\varphi(\tr{\mat{\gamma}}_{0}\mat{X})
					\right\rbrace^{2}
				}\right)^{\frac{1}{2}}.
			\end{equation*}
			To show that the term on the right-hand side of the above inequality converges to zero in probability, let $\zeta > 0$ and consider $\zeta$--brackets $[\varphi_{k}^{L}, \varphi_{k}^{U}]$, $k = 1, \cdots, N$, covering the class $\family{M}_{\epsilon^{\prime}}$ such that
			\begin{equation*}
				\int_{\family{I}_{0}}
				\cbracket{
					\function{\varphi}{k}{U}{t} - \function{\varphi}{k}{L}{t}
				}^{2}
				\dd G_{\tr{\mat{\gamma}}_{0}\mat{X}}(t) \leq \zeta,\quad k=1,\cdots,N,
			\end{equation*}
			where $N\leq\exp(A/\zeta)$ for some constant $A >0$, by Lemma \ref{lemma:bracket_entropy_M}. Using a similar argument in the proof of Lemma \ref{lemma:bracket_entropy_F}, we have, 
			\begin{align*}
				\begin{split}
					&\left(\int_{\family{X}}
					\cbracket{
						\varphi(\tr{\hat{\mat{\gamma}}}_{n}\mat{x}) - 
						\varphi(\tr{\mat{\gamma}}_{0}\mat{x})
					}^{2}
					\dd G_{\mat{X}}(\mat{x})\right)^{\frac{1}{2}}\\%
					&\leq
					\left(\int_{\family{X}}
					\cbracket{
						\varphi_{k}^{U}(\tr{\mat{\gamma}}_{0}\mat{x}+r_{1}\Vert\mat{e}_{n}\Vert_{2}) - 
						\varphi_{k}^{U}(\tr{\mat{\gamma}}_{0}\mat{x})
					}^{2}
					\dd G_{\mat{X}}(\mat{x})\right)^{\frac{1}{2}}\\%
					&\quad+
					\left(\int_{\family{X}}
					\cbracket{
						\varphi_{k}^{U}(\tr{\mat{\gamma}}_{0}\mat{x}) - 
						\varphi_{k}^{L}(\tr{\mat{\gamma}}_{0}\mat{x})
					}^{2}
					\dd G_{\mat{X}}(\mat{x})\right)^{\frac{1}{2}}=(A) + (B),
				\end{split}
			\end{align*}
			for some $k=1,\cdots,N$, where $\mat{e}_{n}=\hat{\mat{\gamma}}_{n} - \mat{\gamma}_{0}$.
			For the square of $(A)$, using a similar argument in \eqref{eq:link_diff_bound} in the proof of Lemma \ref{lemma:bracket_entropy_F}, we have
			\begin{align*}
				\begin{split}
					(A)^{2} 
					&\leq
					\bar{q}_{1}(1-2\epsilon^{\prime})\rbracket{
						\int_{-r_{1}+r_{1}\norm{\mat{e}_{n}}{2}}^{r_{1}+ r_{1}\norm{\mat{e}_{n}}{2}}
							\function{\varphi}{k}{U}{t}
						\dd t - 
						\int_{-r_{1}}^{r_{1}}
							\function{\varphi}{k}{U}{t}
						\dd t
					}\\%
					&=
					\bar{q}_{1}(1-2\epsilon^{\prime})\rbracket{
						\int_{r_{1}}^{r_{1}+r_{1}\norm{\mat{e}_{n}}{2}}
							\function{\varphi}{k}{U}{t}
						\dd t - 
						\int_{-r_{1}}^{-r_{1}+r_{1}\norm{\mat{e}_{n}}{2}}
							\function{\varphi}{k}{U}{t}
						\dd t
					}\\%
					&\leq
					\bar{q}_{1}r_{1}(1-2\epsilon^{\prime})^{2}\norm{\mat{e}_{n}}{2}.
				\end{split}
			\end{align*}
		For the square of $(B)$, since $[\varphi_{k}^{L}, \varphi_{k}^{U}]$, $k = 1, \cdots, N$, are $\zeta$--brackets covering $\family{M}_{\epsilon^{\prime}}$, choosing $\zeta=\norm{\mat{e}_{n}}{2}$ we have $(B)^{2}\leq\norm{\mat{e}_{n}}{2}$. Note that $\zeta$ depends on $n$ and on the realization of the data $\omega$ but the argument can be followed for any fixed $n$ and $\omega$.
			Therefore, by Theorem \ref{prop:theta_estimator_consistency} together with the results on $(A)$ and $(B)$, we obtain
						\begin{align*}
				\begin{split}
					&\sup_{\varphi\in\family{M}_{\epsilon^{\prime}}}\left\vert\expect[\sbracket]{
						\Delta\log\frac{
							\varphi(\tr{\hat{\mat{\gamma}}}_{n}\mat{X})
						}{
							\varphi(\tr{\mat{\gamma}}_{0}\mat{X})
						}
					}\right\vert\\%
					&\leq\frac{1}{\epsilon^{\prime}}
					\left\lbrace
					\sup_{k}
					\left(\int_{\family{X}}
					\cbracket{
						\varphi_{k}^{U}(\tr{\mat{\gamma}}_{0}\mat{x}+r_{1}\Vert\mat{e}_{n}\Vert_{2}) - 
						\varphi_{k}^{U}(\tr{\mat{\gamma}}_{0}\mat{x})
					}^{2}
					\dd G_{\mat{X}}(\mat{x})\right)^{\frac{1}{2}}
					\right.\\%
					&\left.\qquad+
					\sup_{k}
					\left(\int_{\family{X}}
					\cbracket{
						\varphi_{k}^{U}(\tr{\mat{\gamma}}_{0}\mat{x}) - 
						\varphi_{k}^{L}(\tr{\mat{\gamma}}_{0}\mat{x})
					}^{2}
					\dd G_{\mat{X}}(\mat{x})\right)^{\frac{1}{2}}\right\rbrace\\%
					&\leq\frac{1}{\epsilon^{\prime}}
					\cbracket{
						\rbracket{1 + (1-2\epsilon^{\prime})\sqrt{\bar{q}_{1}r_{1}}}
						\sqrt{\norm{\mat{e}_{n}}{2}}
					}
					\converge[P]0\quad\text{as}~n\to\infty.
				\end{split}
			\end{align*}
			Consider $(V)$, using the mean value theorem, the Cauchy-Schwarz inequality, $\epsilon^{\prime} \leq \varphi \leq 1-\epsilon^{\prime}$, $F_{u} \leq 1$, and the Minkowski inequality, we have
			\begin{align*}
				\begin{split}
					(V) 
					&\leq 
					\frac{1}{\epsilon^{\prime}}\left\lbrace
						\left(
							\expect[\sbracket]{\cbracket{
									\varphi(\tr{\hat{\mat{\gamma}}}_{n}\mat{X}) - 
									\varphi(\tr{\mat{\gamma}}_{0}\mat{X})
								}^{2}F_{u}^{2}(Y|\mat{Z};\mat{\beta}_{0},\Lambda_{0})}
						\right)^\frac{1}{2}
					\right.\\%
					&\left.\qquad+
						\left(
							\expect[\sbracket]{\cbracket{
									F_{u}(Y|\mat{Z};\hat{\mat{\beta}}_{n},\hat{\Lambda}_{n}) - 
									F_{u}(Y|\mat{Z};\mat{\beta}_{0},\Lambda_{0})
								}^{2}\varphi^{2}(\tr{\hat{\mat{\gamma}}}_{n}\mat{X})}
						\right)^\frac{1}{2}
					\right\rbrace\\%
					&\leq
					\frac{1}{\epsilon^{\prime}}
					\left(
					\expect[\sbracket]{\cbracket{
							\varphi(\tr{\hat{\mat{\gamma}}}_{n}\mat{X}) - 
							\varphi(\tr{\mat{\gamma}}_{0}\mat{X})
						}^{2}}
					\right)^\frac{1}{2}\\%
					&\qquad+
					\frac{1 - \epsilon^{\prime}}{\epsilon^{\prime}}
					\left(
					\expect[\sbracket]{\cbracket{
							F_{u}(Y|\mat{Z};\hat{\mat{\beta}}_{n},\hat{\Lambda}_{n}) - 
							F_{u}(Y|\mat{Z};\mat{\beta}_{0},\Lambda_{0})
						}^{2}}
					\right)^\frac{1}{2}\\%
					&= (C) + (D)
				\end{split}
			\end{align*}
			The convergence of $(C)$ to zero in probability follows from the previous results on $(IV)$.			
			For the square of $(D)$, using the mean value theorem, $F_{u} \leq 1$, and the Minkowski inequality, we have
			\begin{align*}
				\begin{split}
					\sbracket{\frac{\epsilon^{\prime}}{1 - \epsilon^{\prime}}(D)}^{2}
					&\leq
					\left(
						\expect[\sbracket]{\cbracket{
								\hat{\Lambda}_{n}(Y) - \Lambda_{0}(Y)
							}^{2}\exp(2\tr{\hat{\mat{\beta}}}_{n}\mat{Z})
					}\right)^\frac{1}{2}\\%
					&\quad+
					\left(
					\expect[\sbracket]{\cbracket{
							\exp(\tr{\hat{\mat{\beta}}}_{n}\mat{Z}) - \exp(\tr{\mat{\beta}}_{0}\mat{Z})
						}^{2}\Lambda_{0}^{2}(Y)
					}\right)^\frac{1}{2}\\
					&= (E) + (F).
				\end{split}
			\end{align*}
			For $(E)$, we note that, for every $\mat{z}\in\family{Z}$, 
			\[
				\exp(2\tr{\hat{\mat{\beta}}}_{n}\mat{z}) \leq \exp(2\tr{\mat{\beta}}_{0}\mat{z} + 2r_{2}\norm{\mat{e}^{\prime}_{n}}{2}) \leq c_{2}\exp(\norm{\mat{e}^{\prime}_{n}}{2}),
			\]
			where $\mat{e}^{\prime}_{n} = \hat{\mat{\beta}}_n-\mat\beta_0$ and $c_{2}$ is a positive constant. The existence of the constant $c_2$ follows from assumptions \ref{enum:bounded_support} and \ref{enum:compact_param_space}. By Theorem \ref{prop:theta_estimator_consistency} that $\Vert\hat{\mat{\beta}}_n-\mat\beta_0\Vert_2$ and $\sup_{t \in \sbracket{0, \tau_{0}}}|\hat\Lambda_{n}(t) - \Lambda_{0}(t)|$ converge to zero in probability, $(E)$ converges to zero in probability.
			For $(F)$, we note that $\Lambda_{0}$ is strictly increasing and $\Lambda_{0}(\tau_{0}) <\infty$ by
			\asref{enum:Lambda}. Using the mean value theorem, the Cauchy-Schwarz inequality, assumptions \ref{enum:bounded_support} and \ref{enum:compact_param_space}, and Theorem \ref{prop:theta_estimator_consistency}, we have
			\[
				(F)^{2} \leq c_{3}\exp(\norm{\mat{e}^{\prime}_{n}}{2})\norm{\mat{e}^{\prime}_{n}}{2}
				\converge[P]0\quad\text{as}~n\to\infty,
			\]
			where $c_{3}$ is a positive constant. Thus, combining the above results on $(IV)$ and $(V)$, we have $\sup_{\varphi\in\family{M}_{\epsilon^{\prime}}}\vert\mathbb{M}_{n}(\varphi)-\mathbb{M}(\varphi)\vert\converge[P]0$ as $n\to\infty$.
			
			To show the convergence of the latency estimator $\hat{S}_{u}$, using the mean value theorem,  $F_{u} \leq 1$, and the triangle inequality, we have, for every $t\in[0,\tau_{0}]$ and $\mat{z}\in\family{Z}$,
			\[
				\left\vert
				\hat{S}_{u}(t|\mat{z}) - S_{u}(t|\mat{z})
				\right\vert
				\leq
				\left\vert
				[\hat{\Lambda}_{n}(t) - \Lambda_{0}(t)]\exp(\tr{\hat{\mat{\beta}}}_{n}\mat{z})
				\right\vert
				+
				\left\vert
				[\exp(\tr{\hat{\mat{\beta}}}_{n}\mat{z}) - \exp(\tr{\mat{\beta}}_{0}\mat{z})]\Lambda_{0}(t)
				\right\vert.
			\]
			Using similar arguments as in $(E)$ and $(F)$, we have 
			\[
				\sup_{t \in \sbracket{0, \tau_{0}}}
				\left\vert
				\hat{S}_{u}(t|\mat{z}) - S_{u}(t|\mat{z})
				\right\vert
				\converge[P] 0\quad\text{as}~n\to\infty.
			\]
		\end{proof}
	}

	\section{Additional simulation results}
	\subsection{Selection of truncation parameters}
	\label{sec:appendix_setting_B}
	When the link function has a small range as in experiment B of the simulation study, it is better to determine the truncation parameters for the link function in a data driven way instead of taking fixed bounds $\epsilon'=10^{-6}$ and $1-\epsilon'$. In practice, determining the truncation parameters from the data is challenging. We use ideas form the range-regularized isotonic regression problem studied in \cite{LR2017}. Specifically, we replace the uniform bound restricted MLE in step \ref{eq:link_MLE_Mstep} of Algorithm \ref{algo:link_estimator_algo} by considering the following range-regularized isotonic regression problem 
	\begin{equation}
		\label{eq:BIR}
		\begin{aligned}
			\minimize_{\mat{y}, a, b} \quad& \sum_{i=1}^{n}\rbracket{w_{i} - y_{i}}^{2} + \mu\rbracket{b - a}\\%
			\text{subject to} \quad& a \leq y_{1} \leq y_{2} \leq \cdots \leq y_{n} \leq b,
		\end{aligned}
	\end{equation}
	where $a = \underset{i = 1, \cdots, n}{\inf}y_{i}$, $b = \underset{i = 1, \cdots, n}{\sup}y_{i}$, and $\mu$ is the regularization parameter for shrinking the range $b - a$. To reduce computational cost, the data driven truncation parameters are computed only at the first iteration of the EM algorithm. We applied the bounded isotonic regression algorithm proposed by \cite{LR2017} to obtain solutions over the regularization path. The optimal regularization parameter and hence the lower and upper truncation, are determined by a $K$-fold cross-validation with the expected prediction error for cure probability (EPECP) \citep{JSP2017} as the performance metric. EPECP is a weight adjusted Brier score for assessing the prediction accuracy of a mixture cure model in predicting cure probability. \cite{JSP2017} proposed an estimator for EPECP and studied the statistical properties of this estimator.%
	
	In Tables~\ref{tab:sim_results_B_const_link}-\ref{tab:sim_results_B_nonconst_link} below we provide extra simulation results for experiment B using this choice of data-dependent truncation parameters. 
	\begin{table}[h]
		\centering
		\caption[Simulation results]{Simulation results for experiment B when the truncation parameters are equal \label{tab:sim_results_B_const_link}}
		{
			\renewcommand{\arraystretch}{1}
			\small
			\setlength{\tabcolsep}{2pt}
			\begin{tabular}{cccl@{\extracolsep{6pt}}cc@{\extracolsep{6pt}}cc}
	\hline
	\multirow{2}{*}{Sample Size} & \multirow{2}{*}{$\lambda_{C}$} & \multirow{2}{*}{\shortstack{Number of\\Cases}} & \multirow{2}{*}{Method} & \multicolumn{2}{c}{$\text{MSE}\rbracket{\hat{p},p_{0}}$} & \multicolumn{2}{c}{$\hat{\mat{\gamma}}$} \\ \cline{5-8} 
	&                                &                                  &                         & Mean                       & Variance                    & Bias               & Variance            \\ \hline
	\multirow{4}{*}{250}         & \multirow{2}{*}{0.1}           & \multirow{2}{*}{129}             & mSIC                    & 0.00350                    & 1.77E-06                    & 1.03314            & 0.23414             \\
	&                                &                                  & SIC                     & 0.00666                    & 8.47E-05                    & 1.10141            & 0.30039             \\ \cline{2-8} 
	& \multirow{2}{*}{0.4}           & \multirow{2}{*}{131}             & mSIC                    & 0.00515                    & 9.97E-06                    & 0.95056            & 0.21264             \\
	&                                &                                  & SIC                     & 0.00936                    & 1.51E-04                    & 1.17386            & 0.29076             \\ \hline
	\multirow{4}{*}{500}         & \multirow{2}{*}{0.1}           & \multirow{2}{*}{101}             & mSIC                    & 0.00295                    & 3.26E-07                    & 0.87385            & 0.16862             \\
	&                                &                                  & SIC                     & 0.00400                    & 2.94E-05                    & 0.98576            & 0.30896             \\ \cline{2-8} 
	& \multirow{2}{*}{0.4}           & \multirow{2}{*}{116}             & mSIC                    & 0.00489                    & 4.57E-06                    & 0.84465            & 0.18364             \\
	&                                &                                  & SIC                     & 0.00504                    & 3.19E-05                    & 1.11335            & 0.29095             \\ \hline
\end{tabular}
		}
	\end{table}
	
	\begin{table}[h]
		\centering
		\caption[Simulation results]{Simulation results for experiment B when the truncation parameters are different \label{tab:sim_results_B_nonconst_link}}
		{
			\renewcommand{\arraystretch}{1}
			\small
			\setlength{\tabcolsep}{2pt}
			\begin{tabular}{cccl@{\extracolsep{6pt}}cc@{\extracolsep{6pt}}cc}
	\hline
	\multirow{2}{*}{Sample Size} & \multirow{2}{*}{$\lambda_{C}$} & \multirow{2}{*}{\shortstack{Number of\\Cases}} & \multirow{2}{*}{Method} & \multicolumn{2}{c}{$\text{MSE}\rbracket{\hat{p},p_{0}}$} & \multicolumn{2}{c}{$\hat{\mat{\gamma}}$} \\ \cline{5-6} \cline{7-8}
	&                                &                                  &                         & Mean                       & Variance                    & Bias               & Variance            \\ \hline
	\multirow{4}{*}{250}         & \multirow{2}{*}{0.1}           & \multirow{2}{*}{371}             & mSIC                  & 0.00818                    & 7.55E-05                    & 0.98945            & 0.18702             \\
	&                                &                                  & SIC                     & 0.00888                    & 1.00E-04                    & 1.00346            & 0.25033             \\ \cline{2-8} 
	& \multirow{2}{*}{0.4}           & \multirow{2}{*}{369}             & mSIC                    & 0.01052                    & 1.22E-04                    & 0.94831            & 0.16934             \\
	&                                &                                  & SIC                     & 0.01079                    & 1.28E-04                    & 1.00270            & 0.23379             \\ \hline
	\multirow{4}{*}{500}         & \multirow{2}{*}{0.1}           & \multirow{2}{*}{399}             & mSIC                    & 0.00468                    & 2.75E-05                    & 0.88273            & 0.17155             \\
	&                                &                                  & SIC                     & 0.00531                    & 4.46E-05                    & 0.90003            & 0.24909             \\ \cline{2-8} 
	& \multirow{2}{*}{0.4}           & \multirow{2}{*}{384}             & mSIC                    & 0.00622                    & 4.38E-05                    & 0.87617            & 0.16895             \\
	&                                &                                  & SIC                     & 0.00601                    & 4.80E-05                    & 0.92918            & 0.23441             \\ \hline
\end{tabular}
		}
	\end{table}

{
	\subsection{Non-monotone true link function}
	\label{sec:non_monotone}
To investigate how the method performs when the true link function is not monotone, we consider an additional simulation experiment with the following non-monotone link function,
	\begin{equation*}
		\function{\varphi}{D}{}{u} = \frac{
			\exp\sbracket{\psi(u - c)}
		}{
			1 + \exp\sbracket{\psi(u - c)}
		},
	\end{equation*}
	where $c$ is an intercept term and $\psi(x)= 0.5x^{3}-0.1x^{2}-0.8x+1$. \fref{fig:link_non_monotone} depicts such function over $[-4, 4]$ when the intercept term $c = 0$. \tref{tab:sim_settings_additional} shows the parameters  $\mat{\gamma}_{0}$, $\mat{\beta}_{0}$, and $\lambda_{C}$ of the additional simulation. The averages of the cure proportion, the censoring rate and the proportion of observations in the plateau are also reported in the same table. The truncation parameter $\epsilon^\prime$ is set to $10^{-6}$. Other settings, such as covariates, baseline distribution of the uncured subjects, censoring time distribution, remain the same as described in \sref{sec:sim}. \tref{tab:sim_results_additional} summarizes the simulation results, including the MSE of the link estimates, bias and variance of the coefficient estimates, for both SIC and mSIC methods. As expected, when the true link function is non-monotone, the SIC method performs better in estimating the link and $\mat{\gamma}$. In general, 
	the effect would depend on the amount of deviation from the monotonicity assumption and mSIC would still perform well for small deviations. On the other hand, we observe that both methods behave similarly in estimating $\mat{\beta}$ meaning that the latency component is not very sensitive to non-monotonicity of the link function.
	\begin{figure}
		\centering
		\includegraphics[scale=0.2]{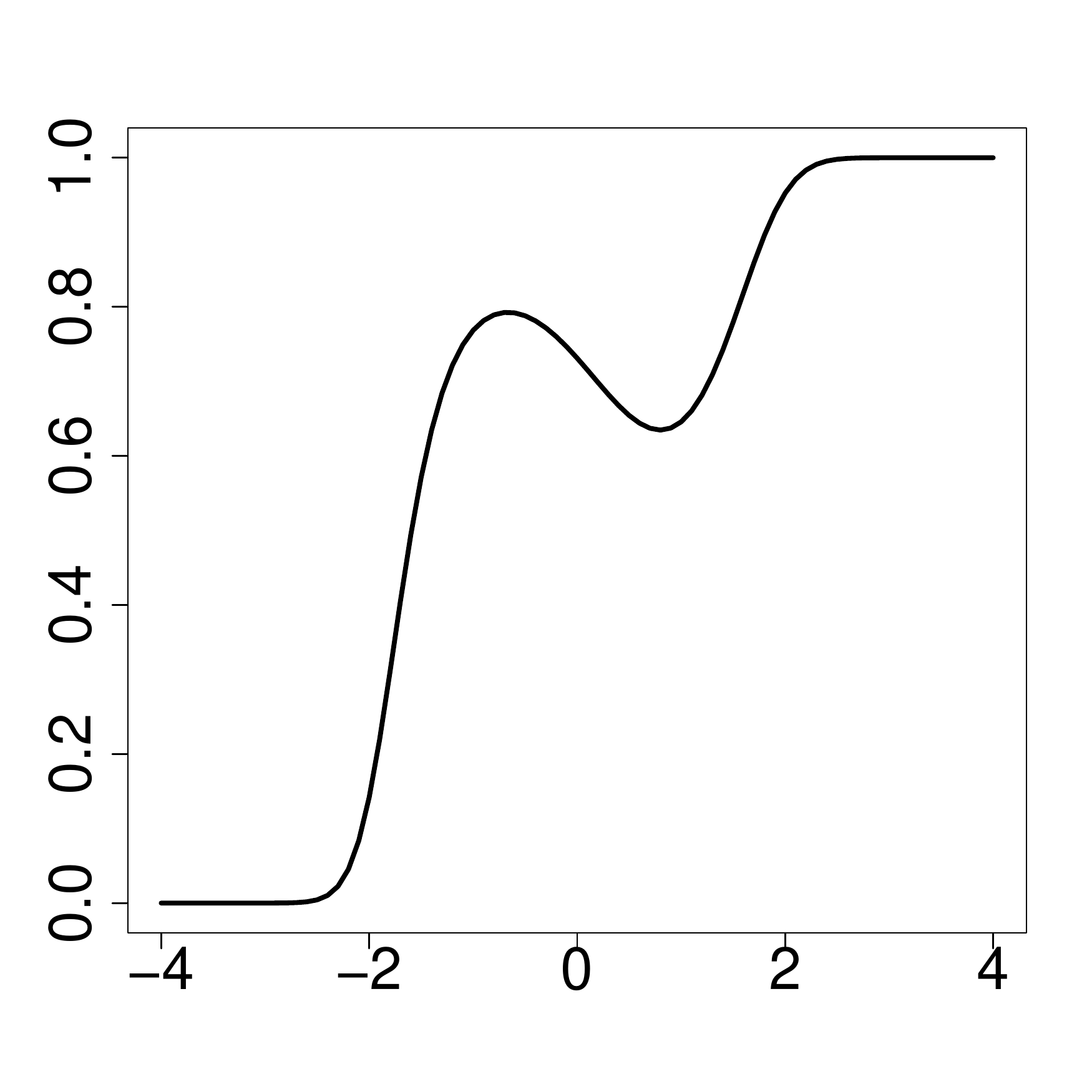}
		\caption{The non-monotone link when the intercept term $c=0$\label{fig:link_non_monotone}}
	\end{figure}
	
	\begin{table}
		\centering
		\caption[Additional simulation settings]{Additional simulation settings\label{tab:sim_settings_additional}}
		{
			\renewcommand{\arraystretch}{1}
			\small
			\setlength{\tabcolsep}{2pt}

\begin{tabular}{cc@{\extracolsep{3pt}}ccccc@{\extracolsep{3pt}}cc@{\extracolsep{3pt}}cccc}
	\hline
	Expt.                  & $c$ & ${\gamma}_{01}$     & ${\gamma}_{02}$    & ${\gamma}_{03}$     & ${\gamma}_{04}$    & ${\beta}_{01}$   & ${\beta}_{02}$  & $\lambda_{C}$ & \shortstack{Cure\\prop.}         & \shortstack{Cens.\\rate} & Plateau \\ \hline
	\multirow{2}{*}{D} & \multirow{2}{*}{-0.5} & \multirow{2}{*}{0.6718} & \multirow{2}{*}{0.2896} & \multirow{2}{*}{-0.1547}  & \multirow{2}{*}{0.6640} & \multirow{2}{*}{-0.4} & \multirow{2}{*}{-0.6} & 0.1           & \multirow{2}{*}{0.2679} & 0.3347         & 0.2079  \\
	&                      &                          &                         &                          &                         &                       &                      & 0.25           &                         & 0.4215         & 0.1463  \\ \hline
\end{tabular}

		}
	\end{table}

	\begin{table}
		\centering
		\caption[Additional simulation results]{Additional simulation results\label{tab:sim_results_additional}} 
		{
			\renewcommand{\arraystretch}{1}
			\small
			\setlength{\tabcolsep}{2pt}
			\begin{tabular}{cccl@{\extracolsep{6pt}}cc@{\extracolsep{6pt}}cc@{\extracolsep{6pt}}cc}
	\hline
	\multirow{2}{*}{Expt.} & \multirow{2}{*}{Size} & \multirow{2}{*}{$\lambda_{C}$} & \multirow{2}{*}{Method} & \multicolumn{2}{c}{$\text{MSE}\rbracket{\hat{p},p_{0}}$} & \multicolumn{2}{c}{$\hat{\mat{\gamma}}$} & \multicolumn{2}{c}{$\hat{\mat{\beta}}$} \\ \cline{5-6} \cline{7-8} \cline{9-10}
	&                              &                                &                         & Mean                       & Variance                    & Bias               & Variance            & Bias               & Variance           \\ \hline
	\multirow{8}{*}{D}          & \multirow{4}{*}{250}         & \multirow{2}{*}{0.1}           & mSIC                    & 0.01921                    & 1.24E-04                    & 0.6041             & 0.1425              & 0.2896            & 0.0301              \\
	&                              &                                & SIC                     & 0.01160                    & 9.05E-05                    & 0.5012             & 0.1283              & 0.2894            & 0.0300              \\ \cline{3-10} 
	&                              & \multirow{2}{*}{0.25}          & mSIC                    & 0.02312                    & 1.45E-04                    & 0.7073             & 0.1959              & 0.3241            & 0.0377              \\
	&                              &                                & SIC                     & 0.01457                    & 1.46E-04                    & 0.6338             & 0.1816              & 0.3190            & 0.0375              \\ \cline{2-10} 
	& \multirow{4}{*}{500}         & \multirow{2}{*}{0.1}           & mSIC                    & 0.01122                    & 5.06E-05                    & 0.4313             & 0.0813              & 0.2121            & 0.0147              \\
	&                              &                                & SIC                     & 0.00601                    & 2.69E-05                    & 0.3250             & 0.0551              & 0.2127            & 0.0150              \\ \cline{3-10} 
	&                              & \multirow{2}{*}{0.25}          & mSIC                    & 0.01324                    & 6.41E-05                    & 0.5126             & 0.1344              & 0.2349            & 0.0169              \\
	&                              &                                & SIC                     & 0.00818                    & 6.73E-05                    & 0.4236             & 0.0973              & 0.2348            & 0.0171              \\ \hline
\end{tabular}
		}
	\end{table}
	
}

{
	\subsection{Sensitivity to the choice of bandwidth}
	\label{sec:bandwidth}
	We re-consider Experiment A and investigate the influence of bandwidth on the mean squared error (MSE) of the estimated cure probability as mentioned in \sref{sec:sim}. In particular, we set the bandwidth parameter at the $k$-th iteration as $h_k = mr_{k}n^{-1/5}$, where $r_{k}$ is the range of the index $\mat{\gamma}^{T}\mat{X}$ computed at the $k$-th iteration and $m \in \lbrace0.25,0.5,0.75,1,1.25,1.5,2,2.5,3\rbrace$. The averages MSE over the 500 simulated datasets for each setting are computed and depicted in \fref{fig:mse_link_bw_sens}. Recall that in \sref{sec:sim} the bandwidth is set to $r_{k}n^{-1/5}$, i.e., $m=1$, and the mSIC method performs better than SIC. This figure shows that the MSE reaches minimum when $m=2$ indicating that there is still room for improvement on estimating the link for mSIC. However, $m=1$ is a satisfactory choice  that does not the increase computational cost of the method. Even without an optimal bandwidth mSIC outperforms the SIC method. Actually for almost all the considered values of $m$ the average MSE of mSIC remains below the one of SIC (given in Table \ref{tab:sim_results}). In general we observe that mSIC is more stable with respect to the choice of the bandwidth than the unconstrained estimator. Furthermore, the influence of the choice of bandwidth reduces when the sample size increases to 500.
	
	\begin{figure}[h]
		\centering
		\includegraphics[scale=0.45]{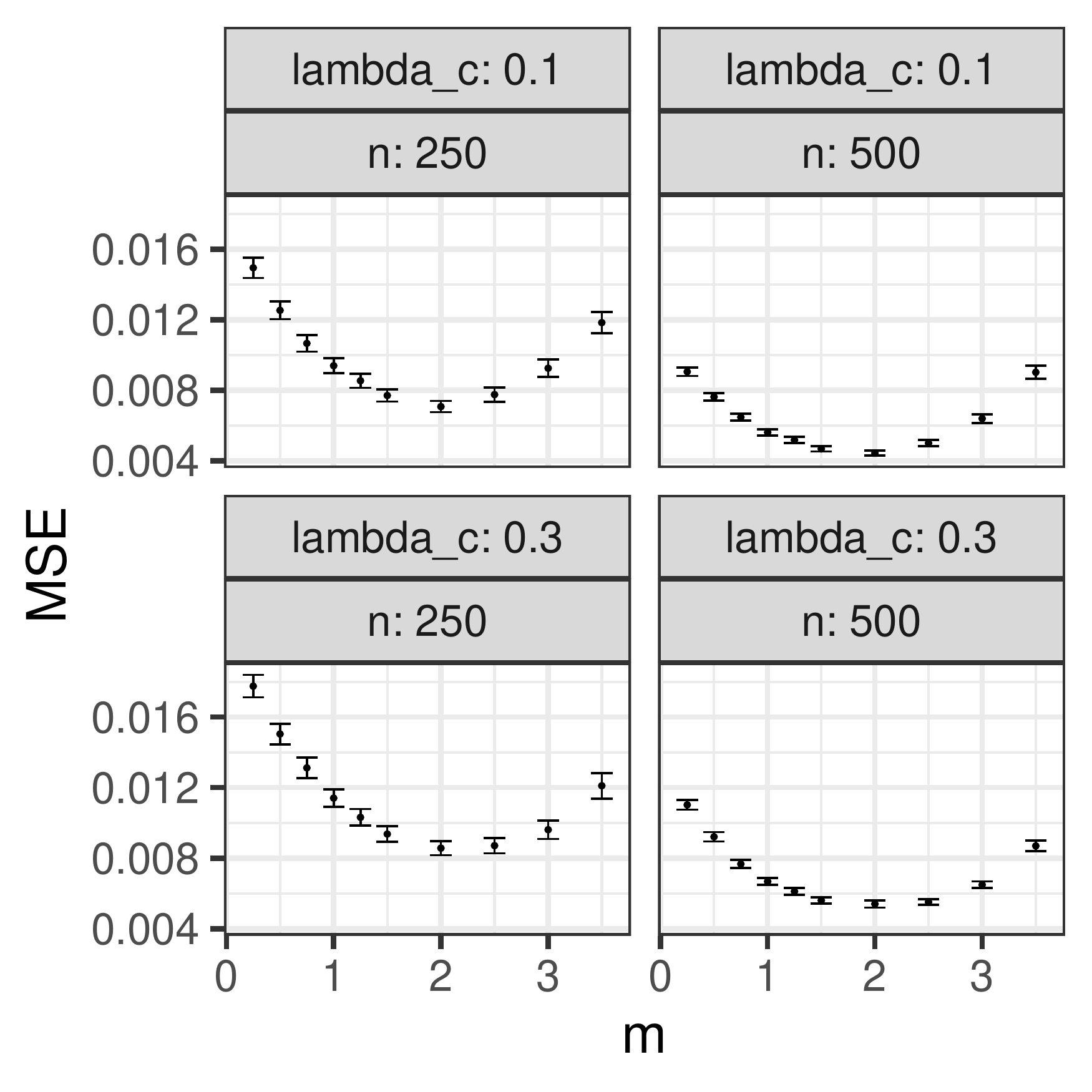}
		\caption{$\text{MSE}\rbracket{\hat{p},p_{0}}$ against different bandwidths $h=mrn^{-1/5}$, where $r$ is the range of the index $\mat{\gamma}^{T}\mat{X}$ and $m \in \lbrace0.25,0.5,0.75,1,1.25,1.5,2,2.5,3\rbrace$\label{fig:mse_link_bw_sens}}
	\end{figure}
}

{
	\subsection{Estimation without smoothing the monotone link estimator}
	\label{sec:smoothing}
	We re-consider Experiment A and study the estimation of $\mat{\gamma}$ and the link $p$ when the non-smoothed monotone link estimator is used, i.e. without applying step~2 in the estimation method described in Section \ref{sec:model_est}. Without smoothing the monotone link estimator, the M-step of the EM algorithm in \eqref{eq:gamma_mle} to estimate $\mat{\gamma}$ can be replaced by the score approach. The estimate of $\mat{\gamma}$ is obtained by computing the zero-crossings of 
	\begin{equation}
		\label{eq:gamma_score}
		\sum_{i=1}^{n} 					
		\left[\frac{w_{i}}{\varphi_{\mat\gamma}(\tr{\mat{\gamma}}\mat{x}_{i})} - \frac{1 - w_{i}}{1 - \varphi_{\mat\gamma}(\tr{\mat{\gamma}}\mat{x}_{i})}\right]\mat{x}_{i},
	\end{equation}
	where $\varphi_{\mat\gamma}$ denotes the monotone link estimate $\hat\varphi_{n,\mat\theta}$ for $\mat{\theta}=(\mat\gamma,\mat\beta^{(m-1)},\Lambda^{(m-1)})$, or minimizing the squared norm of \eqref{eq:gamma_score} over $\family{S}_{d-1}$. \tref{tab:sim_score} shows MSE of the estimated cure probability, bias and variance of the estimated $\mat{\gamma}$. In summary, the mSIC method with the smoothed monotone link estimates performs better in estimating $\mat{\gamma}$ and  the link $p$, among all simulation settings, when comparing with the score approach where the non-smoothed monotone link estimator is used.
	
	\begin{table}
	\centering
	\caption[Simulation results of the score approach]{Simulation results of the score approach\label{tab:sim_score}} 
	{
		\renewcommand{\arraystretch}{1}
		\small
		\setlength{\tabcolsep}{2pt}
		\begin{tabular}{ccl@{\extracolsep{6pt}}cc@{\extracolsep{6pt}}cc}
	\hline
	\multirow{2}{*}{Size} & \multirow{2}{*}{$\lambda_{C}$} & \multirow{2}{*}{Method} & \multicolumn{2}{c}{$\text{MSE}\rbracket{\hat{p},p_{0}}$} & \multicolumn{2}{c}{$\hat{\mat{\gamma}}$} \\ \cline{4-5} \cline{6-7} 
	&                                &                         & Mean                       & Variance                    & Bias               & Variance            \\ \hline
	\multirow{4}{*}{250}  & \multirow{2}{*}{0.1}           & mSIC                    & 0.00939                    & 4.39E-05                    & 0.57368            & 0.05038             \\
	&                                & mSIC (score)            & 0.01645                    & 7.95E-05                    & 0.59970            & 0.05342             \\ \cline{2-7} 
	& \multirow{2}{*}{0.3}           & mSIC                    & 0.01141                    & 5.92E-05                    & 0.61280            & 0.05145             \\
	&                                & mSIC (score)            & 0.01983                    & 1.21E-04                    & 0.64170            & 0.05902             \\ \hline
	\multirow{4}{*}{500}  & \multirow{2}{*}{0.1}           & mSIC                    & 0.00562                    & 1.50E-05                    & 0.43937            & 0.03810             \\
	&                                & mSIC (score)            & 0.01043                    & 3.25E-05                    & 0.49248            & 0.04135             \\ \cline{2-7} 
	& \multirow{2}{*}{0.3}           & mSIC                    & 0.00668                    & 2.01E-05                    & 0.48043            & 0.03989             \\
	&                                & mSIC (score)            & 0.01198                    & 4.08E-05                    & 0.49994            & 0.04055             \\ \hline
\end{tabular}
	}
	\end{table}
}

{	\subsection{Real data application revisited}
	\label{sec:app_application}
		In \sref{sec:application}, we compared the performance of the three models (LC, SIC, and mSIC) in predicting the uncure probability using the prediction error (PE) and observed that mSIC behaved the best. However, even the link estimate of the SIC method in \fref{fig:seer_links} is monotone, which makes it intriguing to understand the reason behind the better performance of mSIC.
	We looked at the link estimates over the 10 random splits for each method, which are shown in \fref{fig:seer_links_split}. We observe that the link estimates of the SIC method are in general not monotone and it was a coincidence that the particular random split considered in \sref{sec:application} gave a monotone estimator. This shows the need to impose the monotonicity assumption in the estimation procedure. Out of the 10 random splits, one in particular leads to a non-monotone link estimate that differs considerably from the others and also reaches value 1 for index in the range $(0.81, 0.87)$. For this case, we observe that the support of the fitted index $\tr{\hat{\mat{\gamma}}}_{n}\mat{x}^{\text{train}}_{i}$ for SIC, a histogram of which is shown in  Figure \ref{fig:seer_sic_index_hist_wo_bw_adj}, is divided into two disjoint intervals. The lack of observations in the region around $(0.81, 0.87)$ leads to the strange behavior of link estimate. Recall that the SIC method employs a leave-one-out cross-validation approach to search the optimal bandwidth over the interval $[0.4,1]$. For this particular case, the selected bandwidth was the upper bound~$1$. We increased the upper bound of the search interval from $1$ to $5$ and fitted again the SIC model for the same 10 random splits of the data. The average PE (standard deviation) for SIC with bandwidth adjustment is 90.54 (5.82), which improves slightly as compared to the SIC models without adjusting the bandwidth, with average PE (standard deviation) 91.84 (9.51). \fref{fig:seer_sic_index_hist_bw_adj} shows that the fitted index after adjusting the search of bandwidth is not divided into two disjoint intervals. This suggests that the cross-validation applied to choose the bandwidth of the SIC method is not stable in practice and the search interval should depend on the range of the index, hence should be data dependent. In contrast, we showed that the mSIC method performs well even with a simple  bandwidth choice.
\fref{fig:seer_sic_link_bw_adj} depicts the link estimates of the SIC method with the bandwidth adjustment.  Again, the link estimates are non-monotone in general but possess less variability as compared to the un-adjusted ones. The estimated links using SIC are more volatile in the region where the fitted index is sparse (see the $+$ sign in Figures \ref{fig:seer_links_split} and \ref{fig:seer_sic_link}) compared to  mSIC. Imposing the monotonicity constraint in mSIC increases the stability of the link estimate particularly in the region of index that possesses less available observations.
		\begin{figure}
		\centering
		\includegraphics[scale=0.55]{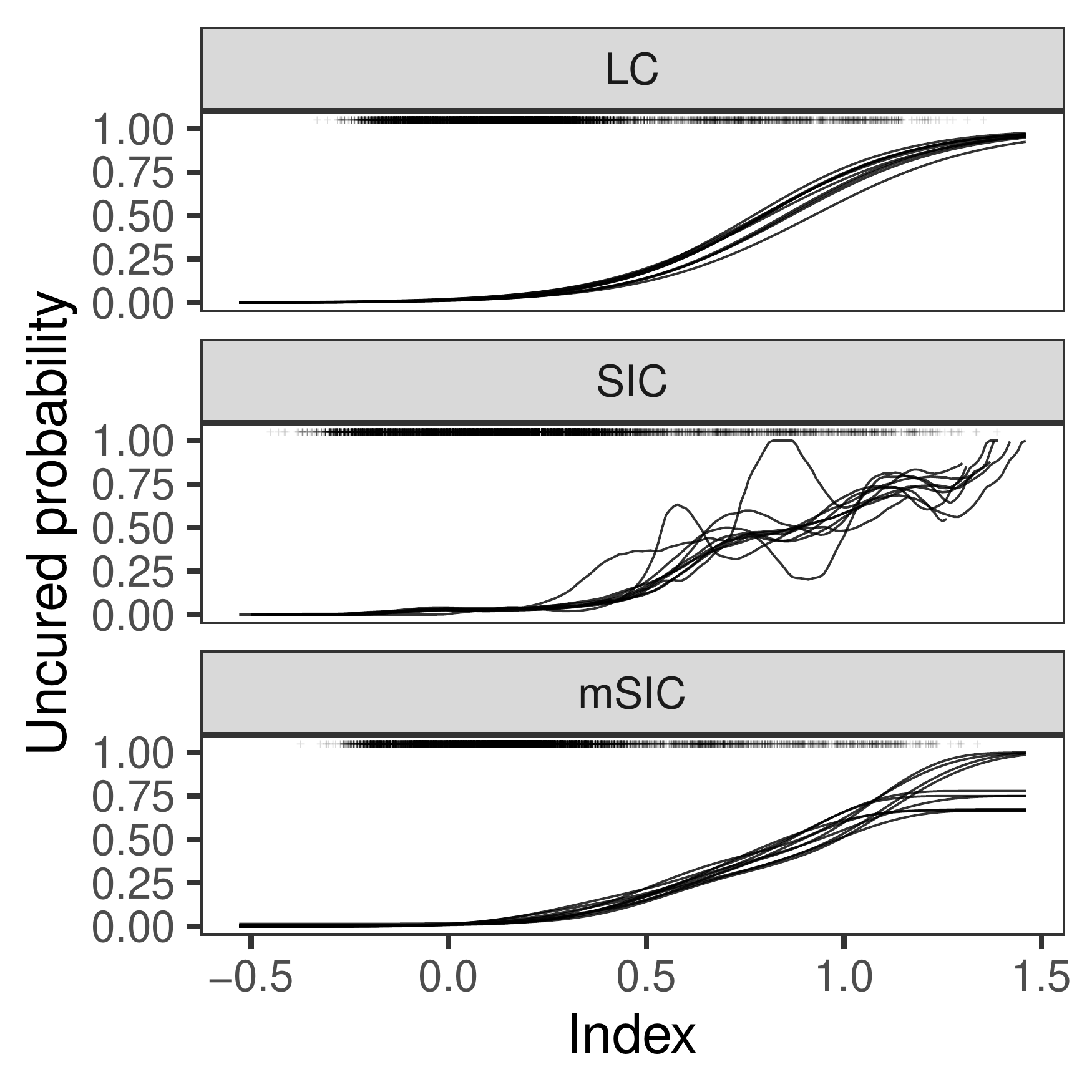}
		\caption{Link estimates over the 10 random splits ($+$ indicates the density of the fitted index )\label{fig:seer_links_split}}
	\end{figure}
	\begin{figure}
		\centering
		\begin{subfigure}[b]{0.49\textwidth}
			\centering
			\includegraphics[width=\columnwidth]{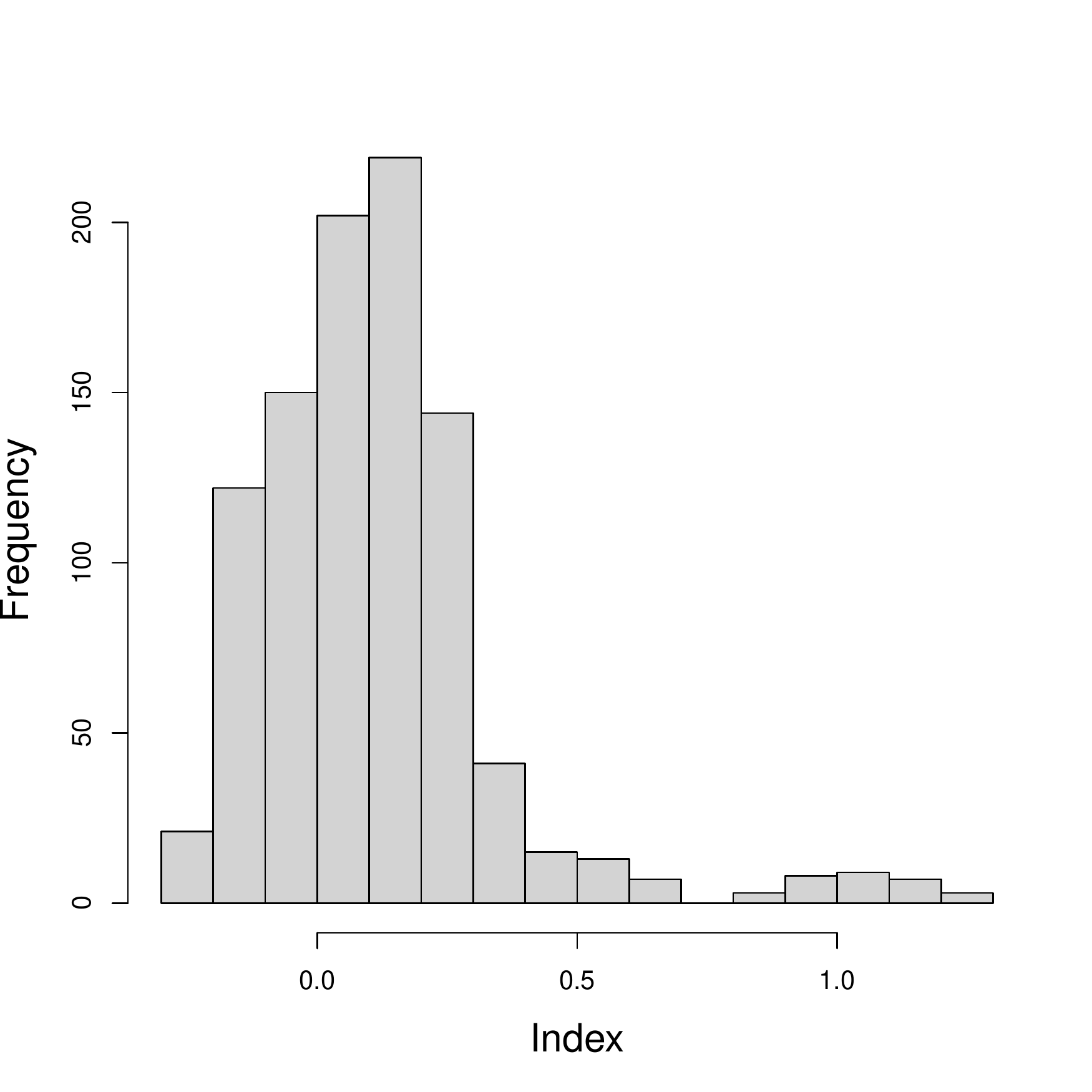}
			\caption{Without bandwidth adjustment\label{fig:seer_sic_index_hist_wo_bw_adj}}
		\end{subfigure}
		\hfill
		\begin{subfigure}[b]{0.49\textwidth}
			\centering
			\includegraphics[width=\columnwidth]{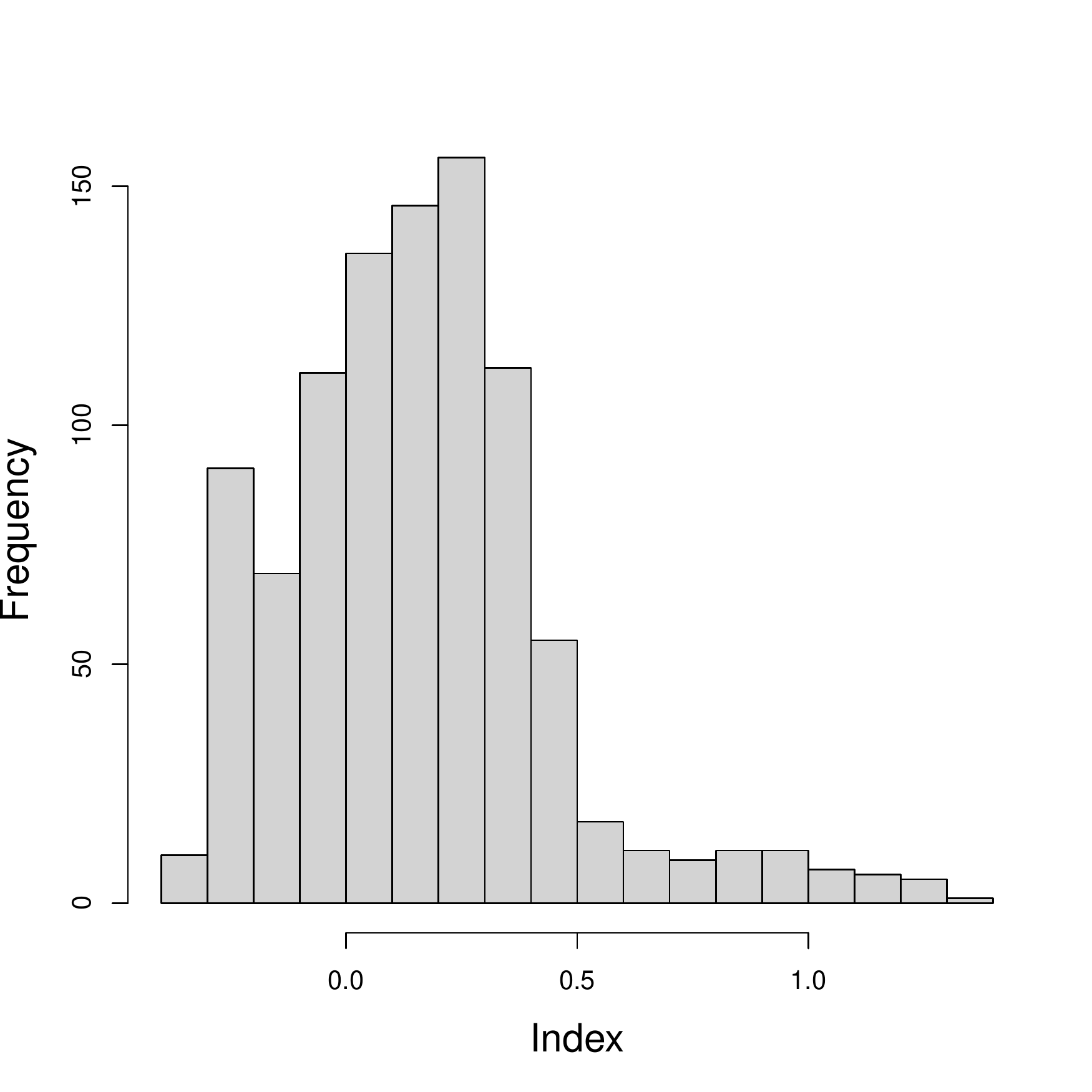}
			\caption{With bandwidth adjustment\label{fig:seer_sic_index_hist_bw_adj}}
		\end{subfigure}
		\caption{Histogram of the fitted index of the selected SIC model\label{fig:seer_sic_index_hist}}
	\end{figure}
	\begin{figure}
		\centering
		\begin{subfigure}[b]{0.49\textwidth}
			\centering
			\includegraphics[width=\columnwidth]{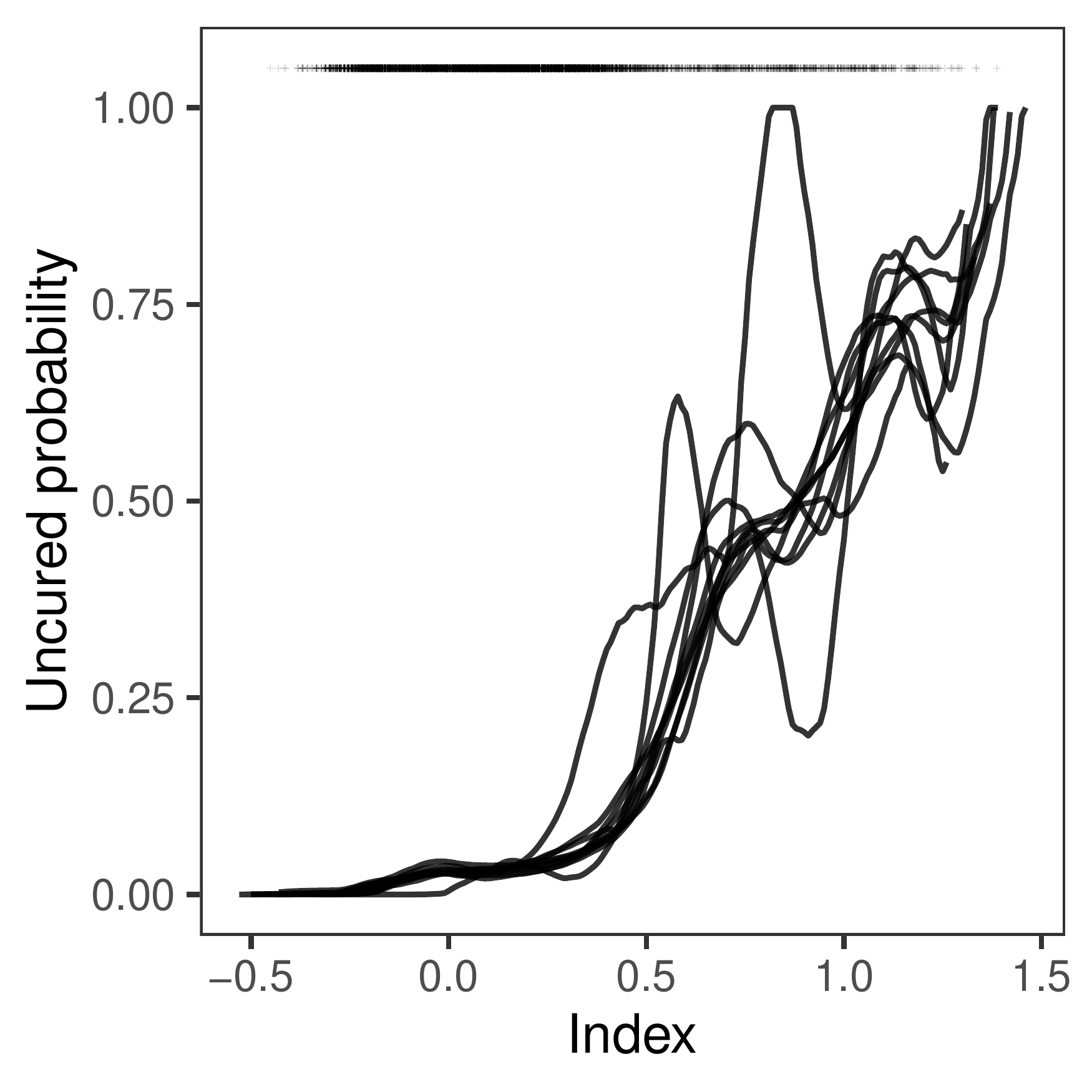}
			\caption{Without bandwidth adjustment\label{fig:seer_sic_link_wo_bw_adj}}
		\end{subfigure}
		\hfill
		\begin{subfigure}[b]{0.49\textwidth}
			\centering
			\includegraphics[width=\columnwidth]{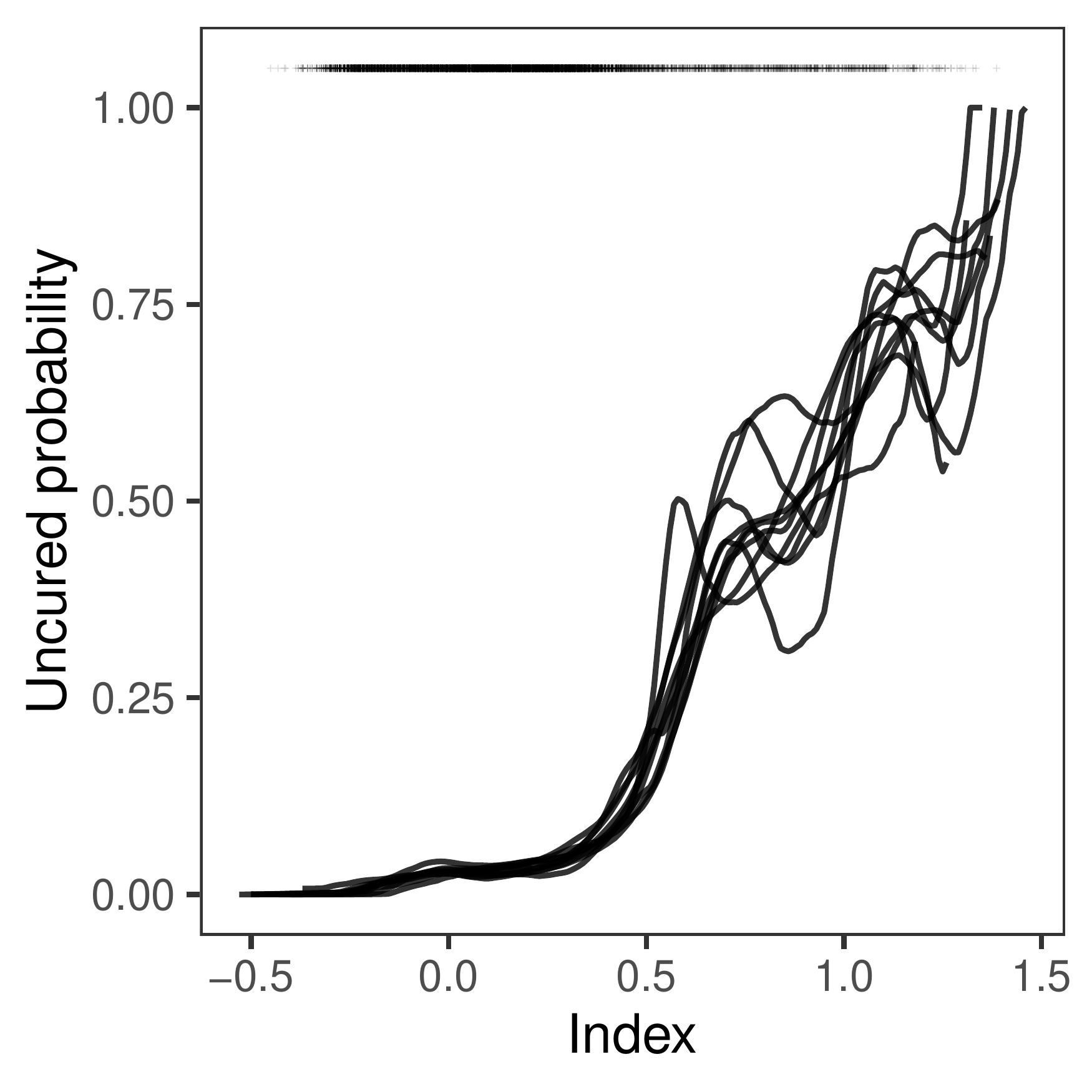}
			\caption{With bandwidth adjustment\label{fig:seer_sic_link_bw_adj}}
		\end{subfigure}
		\caption{SIC link estimates over the 10 random splits ($+$ indicates the density of the fitted index )\label{fig:seer_sic_link}}
	\end{figure}
}

\subsection{Algorithms}
\label{sec:algorithms}
\begin{algorithm}[H]
	\caption{Monotone Link Estimator $\hat{\varphi}_{n,\mat{\theta}}$ for fixed $\mat{\theta} = \rbracket{\mat{\gamma}, \mat{\beta}, \Lambda}$ \label{algo:link_estimator_algo}}
	\begin{algorithmic}[1]
		\Require
		\Statex
		\begin{description}[itemsep=0mm]
			\item[Observed data]
			$\cbracket{\rbracket{y_{i},\delta_{i},\mat{x}_{i},\mat{z}_{i}}, i=1,\cdots,n}$
			\item[Fixed parameter values]
			${\mat{\theta}} = \rbracket{{\mat{\gamma}}, {\mat{\beta}}, {\Lambda}}$
			\item[Initial link] $\hat\varphi^{(0)}$
			\item[Trunction parameter]
			$\epsilon^{\prime} > 0$
		\end{description}
		\State
		Let $\function{{S}}{u}{}{y_{i}|\mat{z}_{i}} = \exp\sbracket{-{\Lambda}^{}\rbracket{y_{i}}\exp\rbracket{\tr{{\mat{\beta}}}\mat{z}_{i}}}$ and ${F}_{u} = 1 - {S}_{u}$
		\State $\hat\varphi^{(0)}_{n,\mat{\theta}}\gets\hat\varphi^{(0)}$ and $k\gets 0$
		\Repeat
		\State $k\gets k+1$
		\State
		$\function{\hat{p}}{}{}{\mat{x}_{i}} \gets
		\function{\hat{\varphi}}{n,\mat{\theta}}{(k-1)}{\mat{\gamma}^{\prime}\mat{x}_{i}}$ for all $i = 1, \cdots, n$
		\Statex \quad
		\textbf{E-step}
		\State
		${w}_{i} \gets \delta_{i} + \rbracket{1 - \delta_{i}}
		\frac{
			\function{{\hat{p}}}{}{}{\mat{x}_{i}}\function{{S}}{u}{}{y_{i}|\mat{z}_{i}}
		}{
			1-\function{{\hat{p}}}{}{}{\mat{x}_{i}}\function{{F}}{u}{}{y_{i}|\mat{z}_{i}}
		}$ for all $i = 1, \cdots, n$
		\Statex \quad
		\textbf{M-step}
		\State\label{eq:link_MLE_Mstep}
		$\hat{\varphi}_{{n,\mat{\theta}}}^{(k)} \gets \argmax_{\varphi \in \family{M}_{\epsilon^{\prime}}} \sum_{i=1}^{n} \sbracket{{w}_{i} \log\link[\tr{{\mat{\gamma}}}\mat{x}_{i}]{}{} + \rbracket{1 - {w}_{i}}\log\cbracket{1 - \link[\tr{{\mat{\gamma}}}\mat{x}_{i}]{}{}}}$
		\Statex \Comment{Using uniform bound restricted MLE}
		\Until{Termination criterion is satisfied }
		\State\Return
		$\hat{\varphi}_{{n,\mat{\theta}}}^{(k)}$
	\end{algorithmic}
\end{algorithm}	
\begin{algorithm}[H]
	\caption{Model Estimation \label{algo:model_estimation_algo}}
	\begin{algorithmic}[1]
		\Require
		\Statex
		\begin{description}[itemsep=0mm]
			\item[Observed data]
			$\cbracket{\rbracket{y_{i},\delta_{i},\mat{x}_{i},\mat{z}_{i}}, i=1,\cdots,n}$
			\item[Trunction parameter]
			$\epsilon^{\prime} > 0$
			\item[Kernel function] $k$
			\item[Bandwidth] $h > 0$
		\end{description}
		\State 
		Initialize $\hat{\mat{\gamma}}^{(0)}$, $\hat{\mat{\beta}}^{(0)}$, 
		$\hat{\Lambda}^{(0)}$ and $\hat{\varphi}^{(0)}$; m $\gets$ 0
		\State
		$\hat{\varphi} \gets \argmax_{\varphi \in \family{M}_{\epsilon^{\prime}}}
		\function{L}{n}{}{\hat{\mat{\gamma}}^{(0)},\hat{\mat{\beta}}^{(0)},\hat{\Lambda}^{(0)},\varphi}$
		\Comment{Using Algorithm \ref{algo:link_estimator_algo} with $\hat{\mat{\theta}}^{(0)}$ and $\hat{\varphi}^{(0)}$}
		\State
		$\function{\hat{\varphi}}{}{(0)}{u} \gets \int\frac{1}{h}\function{k}{}{}{\frac{u - t}{h}}\function{\hat{\varphi}}{}{}{t}dt$ 
		\Comment{Smoothing as in \eqref{eq:smooth_isotonic_estimator}
		}
		\Repeat
		\State
		$m \gets m + 1$
		\Statex \quad	\textbf{E-step}
		\State
		$\function{\hat{p}}{}{}{\mat{x}_{i}}\gets\function{\hat{\varphi}}{}{(m-1)}{\hat{\mat{\gamma}}^{(m-1)^T}\mat{x}_{i}}$ for all $i = 1, \cdots, n$
		
		\State
		$\function{\hat{S}}{u}{}{y_{i}|\mat{z}_{i}}\gets \exp\sbracket{-\hat{\Lambda}^{(m-1)}\rbracket{y_{i}}\exp\rbracket{\hat{\mat{\beta}}^{(m-1)^T}\mat{z}_{i}}}$ and $\hat{F}_{u}^{} = 1 - \hat{S}_{u}^{}$
		\State
		$\hat{w}_{i} \gets \delta_{i} + \rbracket{1 - \delta_{i}}
		\frac{
			\function{\hat{p}}{}{}{\mat{x}_{i}}\function{\hat{S}}{u}{}{y_{i}|\mat{z}_{i}}
		}{
			1-\function{\hat{p}}{}{}{\mat{x}_{i}}\function{\hat{F}}{u}{}{y_{i}|\mat{z}_{i}}
		}$ for all $i = 1, \cdots, n$
		\Statex \quad
		\textbf{M-step}
		\State
		{$\hat{\mat{\gamma}}^{(m)} \gets \argmax_{\mat{\gamma} \in \family{S}_{d-1}}\tilde{L}_{nc,1}\rbracket{\mat{\gamma}, \hat\varphi^s_{n,\mat\gamma,\mat\beta^{(m-1)},\Lambda^{(m-1)}}}$}
		\Statex
		\Comment{Using the augmented Lagrangian method}
		\State
		$\rbracket{\hat{\mat{\beta}}^{(m)}, \hat{\Lambda}^{(m)}} \gets
		\argmax_{\mat{\beta}, \Lambda}	\tilde{L}_{nc,2}\rbracket{\mat{\beta}, \Lambda}$
		\Comment{Using the profile likelihood approach}
		\Statex \quad
		\textbf{Update the link}
		\State
		$\hat{\varphi} \gets \argmax_{\varphi \in \family{M}_{\epsilon^{\prime}}}
		\function{L}{n}{}{\hat{\mat{\gamma}}^{(m)},\hat{\mat{\beta}}^{(m)},\hat{\Lambda}^{(m)},\varphi}$
		\Statex
		\Comment{Using Algorithm \ref{algo:link_estimator_algo} with $\hat{\mat{\theta}}^{(m)}$ and $\hat{\varphi}^{(m-1)}$}
		\State
		$\function{\hat{\varphi}}{}{(m)}{u} \gets \int\frac{1}{h}\function{k}{}{}{\frac{u - t}{h}}\function{\hat{\varphi}}{}{}{t}dt$ 
		\Comment{Smoothing as in \eqref{eq:smooth_isotonic_estimator}
		}
		\Until{Termination criterion is satisfied:}
		\State\Return
		$\hat{\mat{\gamma}}^{(m)}$, $\hat{\mat{\beta}}^{(m)}$, $\hat{\Lambda}^{(m)}$, $\hat{\varphi}^{(m)}$
	\end{algorithmic}
\end{algorithm}
	\bibliographystyle{imsart-number} 
	\bibliography{references}
\end{document}